\documentclass[11pt]{article}

\newif\ifdraftversion
\draftversionfalse    % final: hide the to-do list

\usepackage[utf8]{inputenc}
\usepackage[letterpaper,margin=1in]{geometry}   % 6.5in text block, as in sample.pdf
\usepackage{ifthen}
\usepackage{enumerate}
\usepackage[shortlabels]{enumitem}
\usepackage{color}
\usepackage{array,float}
\usepackage{url}
\usepackage{amsfonts,amsmath,amssymb,mathtools,graphicx}
\usepackage{amsbsy}
\usepackage{bm}
\usepackage{pgfplots}
\usepgfplotslibrary{fillbetween}
\pgfplotsset{compat=1.16}
\usepackage{tikz}
\usetikzlibrary{shapes,arrows,arrows.meta,patterns,decorations.markings,calc,positioning}
\usepackage{footnote}
\usepackage{sidecap}
\usepackage{soul}
\usepackage{comment}
\makesavenoteenv{tabular}
\makesavenoteenv{table}
\usepackage[breakable,skins]{tcolorbox}
\usepackage{csquotes}
\usepackage[authoryear,round]{natbib}
\usepackage{ragged2e}
\usepackage{setspace}
\usepackage{amsthm}
\usepackage{thmtools}
\usepackage{thm-restate}
\usepackage{xurl}

\makeatletter
\@ifundefined{newcounteralias}{}{%
  \renewcommand\thmt@autorefsetup{%
    \@xa\def\csname\thmt@envname autorefname\@xa\endcsname\@xa{\thmt@thmname}%
  }%
}
\makeatother

\usepackage{algorithm}
\usepackage[noend]{algpseudocode}   % [noend]: suppress \EndWhile/\EndIf lines (as in sample.pdf)
\allowdisplaybreaks

\graphicspath{{tex-old/}{tex-new/}{./}}

\definecolor{myLightGray}{gray}{0.97}
\definecolor{cornellred}{rgb}{0.7, 0.11, 0.11}
\definecolor{maroon}{rgb}{0.52, 0, 0}
\definecolor{dgreen}{rgb}{0.0, 0.5, 0.0}
\definecolor{ballblue}{rgb}{0.13, 0.67, 0.8}
\definecolor{royalblue(web)}{rgb}{0.25, 0.41, 0.88}
\definecolor{bleudefrance}{rgb}{0.19, 0.55, 0.91}
\definecolor{royalazure}{rgb}{0.0, 0.22, 0.66}

\theoremstyle{plain}
\newtheorem{theorem}{Theorem}[section]
\newtheorem*{theorem*}{Theorem}
\newtheorem{definition}{Definition}[section]

\newtheorem{fact}{Fact}[section]
\newtheorem{proposition}[theorem]{Proposition}

\newtheorem{lemma}[theorem]{Lemma}

\newtheorem{corollary}[theorem]{Corollary}

\newtheorem{example}[theorem]{Example}

\newenvironment{proofof}[1]{\begin{proof}[Proof of #1]}{\end{proof}}

\newenvironment{APPENDIX}[1]{\par\appendix}{\par}

\newcommand{\SingleSpacedXI}{\singlespacing}

\makeatletter
\gdef\@thetitle{}\gdef\@therunningtitle{}\gdef\@therunningauthor{}
\newcommand{\TITLE}[1]{\gdef\@thetitle{#1}\title{#1}}
\newcommand{\RUNTITLE}[1]{\gdef\@therunningtitle{#1}}
\newcommand{\RUNAUTHOR}[1]{\gdef\@therunningauthor{#1}}
\long\gdef\@articleauthors{}
\newcommand{\ARTICLEAUTHORS}[1]{\long\gdef\@articleauthors{#1}}
\long\gdef\@theabstract{}
\newcommand{\ABSTRACT}[1]{\long\gdef\@theabstract{#1}}

\newcounter{authornum}
\newcommand{\AUTHOR}[1]{%
  \ifnum\value{authornum}>0
    \ifhmode\ifodd\value{authornum}\hspace{0.03\textwidth}\else\\[1.6em]\fi\fi
  \fi
  \stepcounter{authornum}%
  \begin{minipage}[t]{0.46\textwidth}\centering
    {\large #1}\\[0.15em]%
}
\newcommand{\AFF}[1]{#1\end{minipage}}
\newcommand{\EMAIL}[1]{\\[0.2em]\texttt{#1}}

\renewcommand{\maketitle}{%
  \thispagestyle{plain}%
  \begin{center}%
    {\fontsize{17}{21}\selectfont\@thetitle\par}%
    \vspace{2.3em}%
    \setcounter{authornum}{0}%
    \@articleauthors\par
  \end{center}%
  \vspace{2.6em}%
  \begin{center}{\bfseries\large Abstract}\end{center}%
  \begin{quote}\@theabstract\end{quote}%
  \vspace{1.2em}%
}
\makeatother

\newtcolorbox{lpbox}[1]{%
  enhanced, colback=myLightGray, colbacktitle=myLightGray,
  colframe=myLightGray, boxrule=0pt, frame hidden,
  coltitle=black, fonttitle=\bfseries, arc=1pt,
  left=8pt, right=8pt, top=5pt, bottom=5pt,
  halign title=center, title={#1}}

\makeatletter
\renewenvironment*{displayquote}
  {\begingroup
   \setlength{\leftmargini}{0.25cm}%
   \linespread{1.3}\selectfont
   \csq@getcargs{\csq@bdquote{}{}}}
  {\csq@edquote\endgroup}
\makeatother

\tikzset{
  > = stealth,
  midarrow/.style={
    line width=1.2pt,
    decoration={markings, mark=at position 0.5 with {\arrow[scale=1.5]{stealth}}},
    postaction={decorate}
  }
}

\newcommand{\E}[1]{\mathbb{E}\left[ #1 \right]}
\newcommand{\PP}[1]{\mathbb{P}\left( #1 \right)}

\newcommand{\offopt}[0]{\textup{OMN}}
\newcommand{\offoptt}[0]{\textup{OMN}(t)}
\newcommand{\offrbrel}[0]{\textup{LP}^{\textup{OMN}}}

\newcommand{\plp}[0]{\textup{LP}^{\textup{POL}}}
\newcommand{\plpm}[0]{\textup{LP}^{\textup{POL}}(\mathcal{M})}
\newcommand{\plpt}[0]{\textup{LP}^{\textup{POL}}(\mathcal{T \times T})}

\newcommand{\plblp}[0]{{\textup{LP}^{\textup{LB}}}}
\newcommand{\dlblp}[0]{{\textup{DLP}^{\textup{LB}}}}

\newcommand{\dlplong}[0]{\textup{DLP}^{\textup{POL}}(\m{T} \times \m{T})}
\newcommand{\dlpm}[0]{\textup{DLP}^{\textup{POL}}(\m{M})}

\newcommand{\bx}{{\bf x}}
\newcommand{\bn}{{\bf n}}

\newcommand{\m}[1]{\mathcal{#1}}
\newcommand{\prefixes}[1]{\text{Prefixes}(#1)}

\newif\ifshowcoauthorcomments
\showcoauthorcommentstrue   % <-- \showcoauthorcommentsfalse hides all comments
\ifshowcoauthorcomments
  \newcommand{\nb}[3]{{\colorbox{#2}{\bfseries\sffamily\scriptsize\textcolor{white}{#1}}}%
   {\textcolor{#2}{\sf\small\textit{#3}}}}
  \newcommand{\na}[1]{\nb{Nick}{blue}{#1}}
  \newcommand{\fs}[1]{\nb{Felipe}{orange}{#1}}
  \newcommand{\pn}[1]{\nb{Pranav}{dgreen}{#1}}
  \newcommand{\rn}[1]{\nb{Rad}{royalazure}{#1}}
  \newcommand{\todo}[1]{\nb{TODO}{red}{#1}}
  
\else
  \newcommand{\nb}[3]{}
  \newcommand{\na}[1]{}
  \newcommand{\fs}[1]{}
  \newcommand{\pn}[1]{}
  \newcommand{\rn}[1]{}
  \newcommand{\todo}[1]{}
  
\fi

\newcommand{\revcolor}[1]{{\color{maroon}#1}}
\usepackage{pdfcolfoot} % separate color stack for footnotes: prevents \revcolor from leaking into the main text when a colored footnote breaks across pages

\newcommand{\inst}{\mathcal{I}}                  % problem instance  I = {T,lambda,mu,r}
\newcommand{\types}{\mathcal{T}}                 % set of agent types
\newcommand{\TT}{\mathcal{T}\times\mathcal{T}}   % all ordered type pairs T x T
\newcommand{\agentset}{\mathcal{A}}              % set of agents  (use \agentset(t))
\newcommand{\hist}{\mathcal{H}}                  % history          (use \hist(t))
\newcommand{\graph}{G}                           % compatibility graph (use \graph(t))
\newcommand{\reals}{\mathbb{R}}
\newcommand{\Rnn}{\mathbb{R}_{\ge 0}}            % non-negative reals

\newcommand{\arate}[1]{\lambda_{#1}}             % arrival rate of type #1
\newcommand{\drate}[1]{\mu_{#1}}                 % abandonment (departure) rate of type #1
\newcommand{\rew}[2]{r_{#1#2}}                   % reward of an (#1,#2) match
\newcommand{\atype}[1]{i_{#1}}                   % type of agent #1
\newcommand{\atime}[1]{t_{#1}}                   % arrival time of agent #1
\newcommand{\soj}[1]{s_{#1}}                     % maximum sojourn time of agent #1
\newcommand{\dtime}[1]{d_{#1}}                   % unmatched-departure time of agent #1
\newcommand{\pol}{\pi}                           % a (generic) policy
\newcommand{\porder}[1]{\succ_{#1}}              % the priority order of type #1 (as a relation)
\newcommand{\val}[1]{V(#1)}          % steady-state reward rate of policy #1
\newcommand{\valt}[1]{V(#1,t)}       % reward rate of policy #1 by time t

\newcommand{\Npres}[1]{N_{#1}}                   % number of present type-#1 agents
\newcommand{\nss}[1]{n_{#1}}                     % E[number present] in steady state
\newcommand{\Xmatch}[2]{X_{#1#2}}                % cumulative number of (#1,#2) matches
\newcommand{\xrate}[2]{x_{#1#2}}                 % steady-state rate of (#1,#2) matches

\newcommand{\gam}[1]{\gamma_{#1}}                % gamma coefficient for set/type #1
\newcommand{\gammaform}[1]{%
  \frac{1-e^{-\sum_{i\in #1}\lambda_i/\mu_i}}{\sum_{i\in #1}\lambda_i/\mu_i}} % closed form

\newcommand{\score}[2]{r_{#1#2}-v_{#1}-v_{#2}}   % matching score of an (#1,#2) match
\newcommand{\half}{\tfrac{1}{2}}
\newcommand{\indic}[1]{\mathbf{1}\!\left(#1\right)}  % indicator of an event  [Version 4: parentheses, matching the 1(...) used in Sections 5--6 and Appendix B]

\usepackage[hypertexnames=false]{hyperref}
\hypersetup{
  colorlinks      = true,
  linkcolor       = cornellred,        % internal links / refs -> cornell-red text
  citecolor       = royalazure,        % citations            -> royal-azure text
  urlcolor        = bleudefrance,      % URLs                 -> blue text
  linkbordercolor = white,   % cornellred box around links / refs
  citebordercolor = white,   % royalazure box around citations
  urlbordercolor  = {0.19 0.55 0.91}   % bleudefrance box around URLs
}
\makeatletter
\AtBeginDocument{\def\@pdfborder{0 0 1}\def\@pdfborderstyle{}}% keep boxes despite colorlinks
\makeatother
\usepackage[capitalize]{cleveref}
\crefname{fact}{Fact}{Facts}\Crefname{fact}{Fact}{Facts}% Version 2: Facts are now referenced with \Cref

\bibpunct[, ]{(}{)}{,}{a}{}{,}%

\usepackage{microtype}
\usepackage{ragged2e}
\usepackage{graphicx}
\usepackage{booktabs}      
\SingleSpacedXI            % single spacing, as in 
\begin{document}

\TITLE{Dynamic Matching with Abandonment:\\ Optimal Guarantees via Static Priority Policies}
\RUNTITLE{Dynamic Matching with Abandonment: Optimal Guarantees via Static Priority Policies}
\RUNAUTHOR{Arnosti, Niazadeh, Nuti and Simon}

% NOTE: list authors with NO blank lines between \AUTHOR/\AFF entries.
\ARTICLEAUTHORS{%
\AUTHOR{Nick Arnosti}%
\AFF{Industrial and Systems Engineering, University of Minnesota, \EMAIL{arnosti@umn.edu}}%
\AUTHOR{Rad Niazadeh}%
\AFF{The University of Chicago Booth School of Business, \EMAIL{rad.niazadeh@chicagobooth.edu}}%
\AUTHOR{Pranav Nuti}%
\AFF{The University of Chicago Booth School of Business, \EMAIL{pranav.nuti@chicagobooth.edu}}%
\AUTHOR{Felipe Simon}%
\AFF{School of Economics and Business, University of Chile, \EMAIL{fsimon@fen.uchile.cl}}%
}

\ifdraftversion
\todo{List of to-dos:}
\begin{itemize}

% \item Write an abstract
\item Finalize notation/names for $\plpm$ and $\plblp$. Maybe change $\plblp$ to $\textup{LP}^\pi$?
\end{itemize}

\na{Further comments on notation and terminology}
\begin{itemize}

\item allowed vs. accepted (and maybe we should just call $M$ restriction everywhere?)
\end{itemize}

\newpage
\fi

\ABSTRACT{
We study stochastic dynamic matching with abandonment, a fundamental model in which agents of different types arrive over time, may abandon the market while waiting, and can be matched in pairs for rewards that depend on their types. No online policy can guarantee more than 0.5 of the omniscient benchmark, which knows all arrival and departure times. Whether this guarantee is achievable has been a central open problem, even in special cases. We resolve this problem by showing that, on every instance, a simple static priority policy, which is computable in polynomial time, earns at least 0.5 of the omniscient benchmark. This establishes that the optimal competitive ratio is exactly $0.5$, improving the previous best guarantee of 0.192 for the general model~\citep{li2023fully,collina2020dynamic}.

Our analysis rests on a key stochastic result, which is a new lower bound on the probability that some agent from a given set of types is available. 
Building on this bound, we develop two approaches: a primal approach that reads the priorities from an optimal solution of a novel linear program, and a dual approach that assigns a ``value'' to each type, and prioritizes matches with high ``surplus'' (reward minus value of involved agents). These values are the unique solution of a fixed-point equation derived from linear programming duality. We also show that the two approaches are closely connected.%

}
\pagenumbering{gobble}

\maketitle

\newpage
{\hypersetup{linkcolor=royalazure}
\setcounter{tocdepth}{2}
\tableofcontents}

\newpage
\pagenumbering{arabic}  
\section{Introduction}
\label{sec:intro}
%Rad: final paragraph
Many allocation problems in practice are, at their core, instances of dynamic matching with \emph{abandonment}: participants (or goods) arrive over time and may leave the system before being matched. Consider, for example, kidney exchange programs in which patients with incompatible donors enter a pool over time. Such a program repeatedly searches the pool for pairs that can exchange donors, so that both patients receive compatible kidneys~\citep{roth2004kidney,roth2005pairwise,agarwal2019market,ashlagi2021kidney}.\footnote{Two-way swaps capture one form of kidney exchange; in practice, programs also arrange longer cycles and chains.}
% \na{This is not correct in the US, where most transplants through exchanges happen through chains. I believe that in some countries in Europe, only pairwise exchange is allowed.}}
A central challenge is to decide which matches to make, as the urgency of each patient's condition limits how long they can wait in the pool. Timing is even more critical in the allocation of kidneys from deceased donors: a recovered kidney remains viable for at most a day or two and must be discarded if not placed in time \citep{bertsimas2013fairness,agarwal2021equilibrium}. The same basic problem arises over longer horizons in municipal homeless-response systems. In most major U.S. cities, ``Continuum of Care'' programs refer people experiencing homelessness to housing opportunities with various providers. These programs must decide which opportunities to offer to which applicants, accounting for compatibility and the possibility that both applicants and housing opportunities might become unavailable over time \citep{azizi2018designing,arnosti2026target}; similar problems arise in the allocation of public and affordable housing~\citep{arnosti2020design,waldinger2021targeting,leshno2022dynamic}.
%\footnote{A Continuum of Care is a regional planning body, mandated by the U.S.\ Department of Housing and Urban Development, that coordinates housing and services for people experiencing homelessness. Participants may exit the system unmatched---often because they lose contact with the system---and opportunities that remain unfilled for too long may be withdrawn from the coordinated entry system. In the city of Chicago, All Chicago, the lead agency of the Chicago Continuum of Care, operates such a system; see \url{https://allchicago.org/continuum-of-care/}.}
Beyond these high-stakes settings, online chess and video-game platforms solve a similar problem at a massive scale: within moments, they pair each arriving player, who may log off at any time, with an opponent of comparable strength who is waiting in the lobby~\citep{herbrich2006trueskill,li2023fully}. Ride-sharing platforms likewise make real-time matches, pairing riders with nearby drivers before they close the app or choose an alternative platform~\citep{ozkan2020dynamic,hu2022dynamic,aveklouris2025matching}.

%Rad: final paragraph
Despite their differences, all of the above applications share a common structure as centralized matching systems: participants of different types arrive over time, the value of a match depends on the types of both parties and can vary substantially  (patients in need of an organ transplant have different expected outcomes from different compatible donors; homeless households might find certain housing opportunities much more attractive than others; and game players generally prefer playing someone of similar skill), and participants may abandon the system if they wait too long (patient health deteriorates, housing applicants disengage, and players log off). Given such a structure, each platform must decide, dynamically and in real time, how to make matching decisions in order to form the best matches before opportunities disappear.

% \todo{Return to this paragraph.}
%\rn{The central trade-off in dynamic matching is that matching two available agents collects $\rew{i}{j}$ at once, but removes both from the pool, possibly losing more valuable matches each agent could form later. Waiting preserves those options but risks losing agents to abandonment before any matches are made. The planner must balance these considerations without knowledge of future arrival or abandonment times.}

%Rad: final paragraph
These matching decisions involve two intertwined tradeoffs. The first tradeoff concerns \emph{whom to match with whom}: the platform must weigh the immediate reward generated from each match against its opportunity cost. Committing an agent to a partner is irreversible and forecloses every other match that either agent might otherwise have formed, potentially with agents that have not yet arrived. Matching well therefore requires correctly valuing scarce agents under uncertainty---that is, balancing the reward collected today against the future value sacrificed \citep{baccara2020optimal,collina2020dynamic}. The second tradeoff concerns \emph{when to match}: delaying available matches allows the market to ``thicken'' as more agents arrive~\citep{akbarpour2020thickness,ashlagi2023matching}, but an agent held back for a better match may abandon before such a partner arrives~\citep{kohlenberg_2025}.\footnote{Deliberate waiting to thicken a market is common in practice, from the periodic match runs of many kidney exchange programs to the batching windows of ride-sharing platforms~\citep{ashlagi2018matchrun,ashlagi2021kidney,yan2020pricing,feng2025batching, ozkan2020dynamic,aveklouris2025matching}.} This risk is significant because departure times are typically not known in advance: the platform is uncertain not only about future arrivals but also about how long current participants will remain in the system.

%Rad: final paragraph
In this paper, we study a fundamental and parsimonious model that captures these tradeoffs: the \emph{stochastic dynamic matching with abandonment} model, introduced independently by \citet{collina2020dynamic} and \citet{aouad2022dynamic}, and studied, often in special cases, by many subsequent papers \citep{kessel2022stationary,patel2024combinatorial,amanihamedani2024improved,li2023fully}. %The model is built from a simple, minimal set of primitives.
In this model, agents of finitely many types $i\in\types$ arrive according to independent Poisson processes---the canonical model for random yet statistically predictable arrivals%, with rates that the platform can estimate from its data
---with rate $\lambda_i$ for type $i$, and wait in a common pool. Each waiting agent abandons the market after an exponentially distributed sojourn time, at a rate $\mu_i$ that depends on its type~$i$.
%This gives the agents a natural, memoryless model of patience (we sketch how we can relax this assumption in \cref{appendix:general-sojourn-time}).
At any moment, the planner may irrevocably match two waiting agents, collecting a reward that depends on their types and on which of the two arrived first: matching an agent of type $i$ with a later-arriving agent of type $j$ yields reward $r_{ij}$. The planner knows the arrival rates, departure rates, and rewards, and observes arrivals and abandonments as they occur, but does not know the future: an agent's arrival time is learned when it joins the system, and its departure time is revealed only upon departure, when it is too late to act. The objective is to maximize the long-run average reward. %; we formalize the model in \Cref{sec:preliminaries}.

%Rad: final paragraph
A key strength of this model is its generality: it captures several important, practically motivated market structures as special cases. In its most general form---in which any two compatible types can match---the model captures kidney exchange and player matchmaking. Splitting the types into two sides, such as supply and demand, and allowing matches only across sides yields \emph{two-sided dynamic matching}, a market structure studied extensively in operations research~\citep{hu2022dynamic,kohlenberg2024cost,aveklouris2025matching}. This captures applications such as housing and ride-sharing, where participants on both sides arrive, wait, and may leave unmatched. If, in addition, one side must be matched or discarded immediately upon arrival, this becomes a natural model of organ allocation from deceased donors, as organs must be matched soon after recovery to a waiting patient. This special case, studied by \citet{aouad2022dynamic}, \citet{patel2024combinatorial}, and \citet{amanihamedani2024improved}, among others, is known as \emph{stationary  bipartite matching}. Restricting the waiting side further to a single type---consider, for instance, unit-demand buyers arriving to purchase units of a single good that are replenished and perish over time---yields the (single-good) \emph{stationary prophet inequality} problem of \citet{kessel2022stationary}, the Markovian continuous-time counterpart of the classical prophet inequality in optimal stopping \citep{krengel-sucheston_1978}. A guarantee for the general model we study therefore applies to all of these special cases.

\subsection{Competitive Analysis and State of the Art}
\label{sec:framework}

%Rad: final paragraph
Following the literature on Bayesian online algorithms and stochastic online optimization, we evaluate a matching policy by its \emph{competitive ratio} against the \emph{omniscient} benchmark. This benchmark is a ``policy'' with hindsight information: it knows the times of all arrivals and departures and matches optimally given this information. The competitive ratio of a policy is the worst-case ratio, over all instances (i.e., all possible types, rewards, and arrival and departure rates), between its long-run average reward and that of the omniscient benchmark.
%This  is a more demanding standard than the \emph{optimal online} policy, which is the policy obtaining the maximum expected long-run average reward obtainable by any feasible online policy.
%\footnote{We note that the optimal online policy is itself given by an intractable dynamic program whose state encodes the entire vector of queue lengths for all types---hence, the state space is exponential in size. While we do not have a formal statement establishing the hardness of computing this policy, similar exponential-sized dynamic programs in stochastic online optimization are known to be PSPACE-hard~\citep{papadimitriou1987complexity,papadimitriou2024online}.}
%Thus, any guarantee against the omniscient benchmark holds, a fortiori, against the optimal online policy.
It thus measures the fraction of the omniscient benchmark's reward that can be secured without information about future arrivals and departures. 
%\todo{wordsmith}

Prior work shows how demanding this standard is: no online policy can achieve a competitive ratio greater than $\half$, even for the highly restricted single-good stationary prophet inequality problem~\citep{kessel2022stationary}. Whether the ratio of $\half$ is \emph{achievable} has remained open in all but the simplest settings, although a sequence of results has made progressive improvements in various special cases. For the single-good stationary prophet inequality problem, \citet{kessel2022stationary} designed an optimal $\half$-competitive policy. For stationary bipartite matching with multiple types on the waiting side, they obtained a ratio of $\tfrac{15}{56} \approx 0.268$; \citet{patel2024combinatorial} improved this ratio to $1-\tfrac{1}{\sqrt{e}} \approx 0.393$, and \citet{amanihamedani2024improved} obtained a further small improvement. For the general model, \citet{collina2020dynamic} gave a $\tfrac{1}{8}$-competitive policy, and the best known guarantee is the $0.192$-competitive policy of \citet{li2023fully}.\footnote{\citet{li2023fully} allow general sojourn-time distributions; with Poisson arrivals, their guarantee of $0.192$ requires the variance of each type's sojourn time to be at most the square of its mean, which holds (with equality) for exponential sojourn times. Their model assumes that the reward of a match does not depend on which agent arrived first.} Thus, a substantial gap remains between the best known guarantees and the upper bound of $\half$; \citet{patel2024combinatorial} posed achieving $\half$ as an open question even for the special case of stationary bipartite matching. We should also note that a parallel line of work evaluates policies against the weaker benchmark of the optimal online policy. For the general model, the $0.192$ guarantee of \citet{li2023fully} against the omniscient benchmark is also the best known guarantee against the optimal online policy, improving the earlier $\tfrac{1}{4}(1-\tfrac{1}{e})\approx 0.158$ guarantee of \citet{aouad2022dynamic}.
%, and the slightly larger guarantee of \citet{amanihamedani2024improved}.
\Cref{fig:positioning} summarizes the established guarantees for each setting and benchmark.

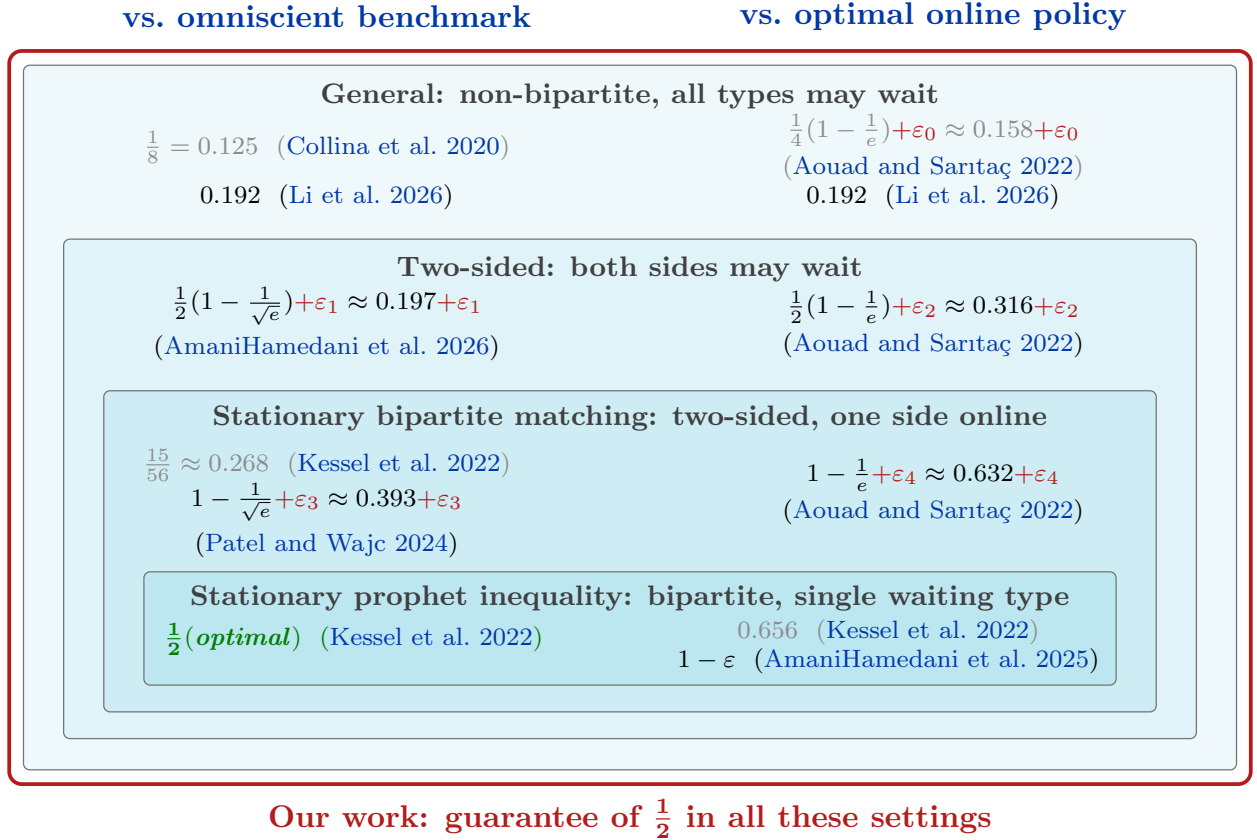
\begin{figure}[htb]
\centering
\hypersetup{pdfborder={0 0 0}}% suppress link borders inside this figure (they show against the tinted fills)
\definecolor{resultgreen}{rgb}{0.0,0.45,0.2}
\resizebox{0.98\textwidth}{!}{%
\begin{tikzpicture}[font=\footnotesize,
  setting/.style={rounded corners=2pt, draw=black!55, line width=0.4pt},
  slabel/.style={font=\footnotesize\bfseries, text=black!75},
  entry/.style={font=\scriptsize, align=center},
  old/.style={font=\scriptsize, align=center, text=black!45}]
  % column headers
  \node[font=\small\bfseries, text=royalazure] at (3.4,8.45) {vs.\ omniscient benchmark};% [Version 4: headers moved up 0.2 to clear the red frame]
  \node[font=\small\bfseries, text=royalazure] at (10.2,8.45) {vs.\ optimal online policy};
  % nested setting boxes (outermost = most general)
  \draw[setting, fill=ballblue!7]  (0.00,0.00) rectangle (13.60,7.90);
  \draw[setting, fill=ballblue!14] (0.45,0.35) rectangle (13.15,5.95);
  \draw[setting, fill=ballblue!22] (0.90,0.65) rectangle (12.70,4.25);
  \draw[setting, fill=ballblue!30] (1.35,0.95) rectangle (12.25,2.22); % (revised: box enlarged slightly to fit the new entry)
  % ---- outermost: general -------------------------------------------------
  \node[slabel] at (6.8,7.55) {General: non-bipartite, all types may wait};
  \node[old]   at (3.4,6.98) {${\tfrac{1}{8}=0.125}$ \ (\citealp{collina2020dynamic})};
  \node[entry] at (3.4,6.42) {${0.192}$ \ (\citealp{li2023fully})};
  \node[old]   at (10.2,6.98) {${\tfrac{1}{4}(1-\tfrac{1}{e}){\color{cornellred}+\varepsilon_0}\approx 0.158{\color{cornellred}+\varepsilon_0}}$\\[2.5pt](\citealp{aouad2022dynamic})};
  \node[entry] at (10.2,6.42) {${0.192}$ \ (\citealp{li2023fully})};
  % ---- two-sided, both sides patient --------------------------------------
  \node[slabel] at (6.8,5.62) {Two-sided: both sides may wait};
  \node[entry] at (3.4,5.00) {${\tfrac{1}{2}(1-\tfrac{1}{\sqrt{e}}){\color{cornellred}+\varepsilon_1}\approx 0.197{\color{cornellred}+\varepsilon_1}}$\\[2.5pt](\citealp{amanihamedani2024improved})};
  \node[entry] at (10.2,5.00) {${\tfrac{1}{2}(1-\tfrac{1}{e}){\color{cornellred}+\varepsilon_2}\approx 0.316{\color{cornellred}+\varepsilon_2}}$\\[2.5pt](\citealp{aouad2022dynamic})};
  % ---- two-sided, one side online ------------------------------------------
  \node[slabel] at (6.8,3.95) {Stationary bipartite  matching: two-sided, one side online};
  \node[old]   at (3.4,3.42) {${\tfrac{15}{56}\approx 0.268}$ \ (\citealp{kessel2022stationary})};
  \node[entry] at (3.4,2.80) {${1-\tfrac{1}{\sqrt{e}}{\color{cornellred}+\varepsilon_3}\approx 0.393{\color{cornellred}+\varepsilon_3}}$\\[2.5pt](\citealp{patel2024combinatorial})};
  \node[entry] at (10.2,3.12) {${1-\tfrac{1}{e}{\color{cornellred}+\varepsilon_4}\approx 0.632{\color{cornellred}+\varepsilon_4}}$\\[2.5pt](\citealp{aouad2022dynamic})};
  % ---- stationary prophet inequality ---------------------------------------
  \node[slabel] at (6.8,1.93) {{Stationary prophet inequality: bipartite, single waiting type}};
  \node[entry] at (3.7,1.47) {\color{dgreen}{$\mathbf{\half (\textbf{\emph{optimal}})}$} \ (\citealp{kessel2022stationary})};
  % (revised: 0.656 is now shown in gray, as superseded by the approximation scheme of AmaniHamedani et al. 2025)
  \node[old]   at (9.7,1.55) {${0.656}$ \ (\citealp{kessel2022stationary})};
  \node[entry] at (9.7,1.22) {$1-\varepsilon$ \ (\citealp{amanihamedani2025adaptive})};
  % ---- our result: red frame around all settings, with a red label below it -----
  % [Version 4: added at the authors' request; the label replaces the last sentence of the caption.]
  \draw[draw=cornellred, line width=1.2pt, rounded corners=4pt] (-0.16,-0.16) rectangle (13.76,8.06);
  \node[font=\small\bfseries\boldmath, text=cornellred] at (6.8,-0.55) {Our work: guarantee of $\tfrac{1}{2}$ in all these settings};
  % % ---- banner ---------------------------------------------------------------
  % \node[font=\small\bfseries, text=resultgreen, align=center] at (6.8,-0.60)
  %   {This paper: a static priority policy guaranteeing $\half$ of the omniscient benchmark\\ on every instance---optimal in all of the settings above.};
\end{tikzpicture}}
\caption{Competitive ratios known before this work, against the omniscient benchmark (left column) and against the optimal online policy (right column). The nested boxes range from the general model (outermost) to the stationary prophet inequality (innermost). Red terms $+\varepsilon_k$ mark the small improvements of \citet{amanihamedani2024improved}, and gray entries have been superseded by later work. The guarantee against the omniscient benchmark when both sides may wait is implied by the results of \citet{amanihamedani2024improved} through a standard reduction to the case in which one side is online. In the innermost setting, the adaptive approximation scheme of \citet{amanihamedani2025adaptive} supersedes the $0.656$ guarantee of a posted-price policy \citep{kessel2022stationary}.}
\label{fig:positioning}
\end{figure}

Obtaining performance guarantees is challenging, especially without the kind of large-market assumptions used by \citet{ozkan2020dynamic} and \citet{aveklouris2025matching}. Greedy policies that make every feasible match may match scarce agents with frequent, low-value arrivals. Batching policies, which are common in large-market regimes, may lose impatient agents while the market thickens; in a cost-minimization version of the problem, even optimally tuned batching policies can perform arbitrarily poorly~\citep{aouad2022dynamic}. Moreover, analyzing the reward obtained by a policy is difficult because the availabilities of different types of agents in the pool are correlated through their shared matching history. 

To address these challenges, existing approaches obtain the guarantees above through a common strategy: they first solve a linear programming relaxation (or a variant of it) that prescribes target match rates, and then use these rates to design randomized online policies, such as policies that match with carefully tuned probabilities \citep{collina2020dynamic,aouad2022dynamic,kessel2022stationary,patel2024combinatorial,li2023fully}. In these papers, randomization plays a key role: it simplifies the analysis and ensures that some matching capacity remains available for favorable future arrivals.

These randomized policies, however, differ sharply from those used in practice, where platforms typically rely on \emph{priority rules}: coordinated entry systems refer compatible housing applicants in priority order, deceased-donor allocation follows point-based rankings, and matchmaking systems rank feasible opponents~\citep{fajardo2026coordinated,israni2014new,graepel2006ranking}. Such rules are deterministic, transparent, and fixed in advance. This gap between theory and practice leads to our central research question: can a simple priority-based policy of the kind practitioners already use attain an optimal performance guarantee? More precisely, we ask the following algorithmic question:

\vspace{-2mm}
\begin{displayquote}
\emph{Can a simple priority-based policy, computable in polynomial time, achieve the best possible competitive ratio of $\half$ for stochastic dynamic matching with abandonment---in full generality?} %computable in polynomial time,
\end{displayquote}

\subsection{Main Result: Static Priority Policies Offer Optimal Guarantees}
\label{sec:main-result}
%Rad: final paragraph
We answer this question in the affirmative. We study the simple class of \emph{static priority} policies.\footnote{We borrow this term from the literature on dynamic matching without abandonment, where \citet{kerimov2023optimality} show that well-designed greedy static priority policies are hindsight optimal in certain matching networks.} Such a policy is defined by a ranked list of acceptable partners for each type, and it acts greedily with respect to these lists: an arriving agent is matched to a waiting agent of the highest-ranked acceptable type that is available, and joins the pool if no such agent is available. 
%priority list for each type: a ranked subset of the types it deems acceptable. An arriving agent is matched to a waiting agent of the highest-ranked available type on its list, and joins the waiting pool if no listed type is available.
Here, ``static'' refers to the fact that priority lists are fixed in advance and do not depend on time, queue lengths, or any other aspect of history; the policy still matches based on which types are available (see \Cref{fig:priority}). Our main result is the following.

\vspace{-2mm}
\begin{displayquote}
%\begin{informal_theorem}[See \Cref{thm:main-result}]
%\label{informal:main}
{\bf Main Result.} {\it There is a polynomial-time algorithm that, given any instance of dynamic matching with abandonment, computes a static priority policy whose long-run average reward is at least $\half$ of the omniscient benchmark. No online policy can guarantee more.}

%\end{informal_theorem}
\end{displayquote}

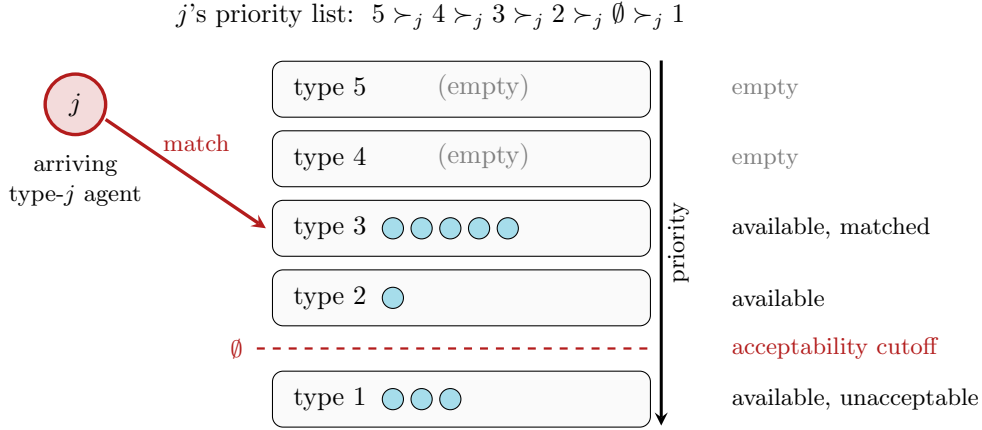
\begin{figure}[htb]
\centering
\begin{tikzpicture}[font=\small,
  agent/.style={circle, draw, fill=ballblue!40, minimum size=2.8mm, inner sep=0pt},
  jagent/.style={circle, draw=cornellred, line width=0.45mm, fill=cornellred!16,
                 minimum size=8mm, inner sep=0pt},
  row/.style={draw, rounded corners, minimum width=5.0cm, minimum height=0.74cm,
              anchor=west, fill=black!2}]

  % j's priority order (header), kept left of the priority arrow
  \node[anchor=south] at (5.0,4.05)
    {$j$'s priority list:\ \ $5 \porder{j} 4 \porder{j} 3 \porder{j} 2 \porder{j} \emptyset \porder{j} 1$};

  % priority list: one row per type, highest priority at the top
  \node[row] (r5) at (2.9, 3.40) {};
  \node[row] (r4) at (2.9, 2.48) {};
  \node[row] (r3) at (2.9, 1.56) {};
  \node[row] (r2) at (2.9, 0.64) {};
  \node[row] (r1) at (2.9,-0.70) {};

  % type labels (inside, left)
  \foreach \y/\t in {3.40/5, 2.48/4, 1.56/3, 0.64/2, -0.70/1}
    \node[anchor=west] at (3.05,\y) {type $\t$};

  % queues: types 5 and 4 are empty; types 3, 2, 1 hold waiting agents
  \node[anchor=base, gray] at (5.7,3.34) {(empty)};
  \node[anchor=base, gray] at (5.7,2.42) {(empty)};
  % type 3 holds five waiting agents, type 2 a single one, type 1 three
  \foreach \x in {4.5,4.88,5.26,5.64,6.02} { \node[agent] at (\x,1.56){}; }
  \node[agent] at (4.5,0.64){};
  \foreach \x in {4.5,4.88,5.26} { \node[agent] at (\x,-0.70){}; }

  % annotations (outside, right)
  \node[anchor=west, font=\footnotesize, gray] at (8.85,3.40) {empty};
  \node[anchor=west, font=\footnotesize, gray] at (8.85,2.48) {empty};
  \node[anchor=west, font=\footnotesize] at (8.85,1.56) {available, matched};
  \node[anchor=west, font=\footnotesize] at (8.85,0.64) {available};
  \node[anchor=west, font=\footnotesize] at (8.85,-0.70) {available, unacceptable};

  % acceptability cutoff (red) between type 2 and type 1
  \draw[dashed, cornellred, line width=0.3mm] (2.7,-0.03) -- (7.9,-0.03);
  \node[anchor=east, font=\footnotesize, cornellred] at (2.65,-0.03) {$\emptyset$};
  \node[anchor=west, font=\footnotesize, cornellred] at (8.85,-0.03) {acceptability cutoff};

  % priority direction (top = highest)
  \draw[->, line width=0.4mm] (8.05,3.78) -- (8.05,-1.05);
  \node[rotate=90, anchor=south, font=\footnotesize] at (8.55,1.37) {priority};

  % arriving type-j agent (left), aligned with the matched row
  \node[jagent] (j) at (0.3,3.20) {$j$};
  \node[anchor=north, align=center, font=\footnotesize] at (0.3,2.66)
    {arriving\\type-$j$ agent};

  % the match: j to the highest-ranked available type (type 3)
  \draw[->, line width=0.45mm, cornellred] (0.7,3.0) -- (2.82,1.56);
  \node[anchor=south, font=\footnotesize, cornellred] at (1.9,2.45) {match};
\end{tikzpicture}
\caption{Illustration of the matching process of a static priority policy. The blue dots indicate the number of agents of each type that are available to match when a new agent of type $j$ arrives. Since the highest-priority types ($5$ and $4$) are not available, the arriving agent matches to an agent of type $3$. If no acceptable agents were available, the arriving agent would wait. %\rn{I like this figure, but I can also imagine dropping it.}
}
\label{fig:priority}
\end{figure}

%Rad: final paragraph
This establishes that the optimal competitive ratio for general stochastic dynamic matching with abandonment is \emph{exactly} $\half$, and that a simple deterministic policy attains it. Our result improves the best known guarantee for the general problem from $0.192$ to the optimal $\half$, and extends the optimal single-type guarantee of \citet{kessel2022stationary} to a much larger class of markets---with arbitrarily many types on both sides, non-bipartite compatibility structures, and order-dependent rewards. It achieves the best possible guarantee in every special case considered so far, and closes the central open problem in this line of work, resolving the question of \citet{patel2024combinatorial} beyond the setting in which it was posed. Since the omniscient benchmark dominates the optimal online policy, our policy also guarantees $\half$ of the reward of the optimal online policy, improving upon the best prior guarantees of $0.192$ for the general model and about $0.316$ for two-sided markets in which both sides may wait.

\medskip
\noindent\textbf{Managerial insights.}
%Rad: final paragraph
Our main result shows that the value of knowing the future in these markets, measured in the worst case over instances, is precisely a factor of two. It is no greater, regardless of the complexity of the instance, and no smaller, regardless of the sophistication of the policy. Furthermore, properly designed priority rules of the kind favored in practice provide optimal worst-case guarantees, which cannot be improved by clever policy design or unbounded computational power.
%are worst-case optimal, and no policy, however adaptive or computationally powerful, can guarantee more than what a well-chosen static priority list already secures in the worst-case.
In other words, simplicity entails no loss in worst-case performance.

In fact, we show that the guarantee of $\half$ can be achieved by a particularly simple class of static priority policies, which we call \emph{value-based} policies. Such a policy is described by a single value $v_i$ for each type $i$, capturing the opportunity cost of using a waiting agent of that type. (By contrast, a generic static priority policy has a description whose size is quadratic in the number of types.)
%compute actuall even stronger sense: it has a one-number-per-type description. Our algorithm assigns each type $i$ a \emph{value} $\dvar{i}$, interpreted as the opportunity cost of consuming a waiting type-$i$ agent. These values are computed once by solving a polynomial-time linear program (as we discuss later).
A match between an agent of type $i$ and a later-arriving agent of type $j$ is given the score $\score{i}{j}$. This score is used both to identify acceptable matches (all matches with positive scores, and possibly some with zero score) and to rank them. Our results show that assigning a single number to each type can address the complex tradeoffs identified earlier with no loss to worst-case performance.

\subsection{Summary of the Techniques}
\label{sec:techniques}
%Establishing a worst-case guarantee for a static priority policy runs into an immediate obstacle. Once the priority lists are fixed, 

Given a static priority policy, the numbers of waiting agents of different types evolve as a continuous-time Markov chain. An exact analysis of the policy's performance would require computing the stationary distribution of this chain. This is difficult because the counts of different types are correlated through their shared history of matches, and their joint distribution depends on the entire collection of priority lists. Our analysis sidesteps this difficulty: it never computes a stationary distribution and instead indirectly certifies the performance of a policy, as discussed below.

\medskip
\noindent\textbf{An upper bound on the omniscient benchmark.} Our starting point is a linear programming relaxation~$\offrbrel$ of the omniscient benchmark (\Cref{sec:benchmark-guarantee}), whose variables are the long-run rates of matches between types. Its \emph{balance} constraints ensure that no type is matched more often than it arrives. Its \emph{availability} constraints limit how often arrivals of each type can be matched with agents from a given set of types. Such a match requires an agent from the set to be waiting, and since matches only remove agents, an agent from the set is waiting at most as often as in a market where no one is ever matched. In that market, the numbers of waiting agents of different types are independent Poisson random variables, which gives an explicit bound. Similar linear programs appear in prior work. Ours is tighter than those proposed by \citet{collina2020dynamic} and \citet{kessel2022stationary}, which only constrain matches with single types, and it covers a more general setting than the relaxation of \citet{amanihamedani2024improved} for stationary bipartite matching.

\medskip
\noindent\textbf{Properties of static priority policies and a lower bound on agent availability.} These linear programs build on a few observations about static priority policies. Suppose an agent of a given type arrives, and consider a \emph{prefix} of its ranked list of types. Because these policies match the arriving agent to some type within this prefix whenever at least one type in the prefix is available, the analysis of their reward reduces to a question about agent availability: how often does a given set of types have at least one agent waiting? Calculating this probability exactly is difficult in general, but \Cref{sec:priority-policies} (\Cref{lem:gammabound}) provides a simple lower bound:
the {\em probability} that some agent from a set $S$ is waiting is at least the {\em expected} number of waiting agents from $S$, times an explicit constant $\gamma_S$ that depends on the arrival and departure rates but not on the policy. Interestingly, this bound can be viewed as a converse, up to the factor $\gamma_S$, to Markov's inequality.
%If $N_S$ denotes the number of waiting agents with types in the set $S$, it establishes that $\mathbb{P}(N_S > 0)$ is at least $\gamma_S \mathbb{E}[N_S]$, where $\gamma_S$ is an explicit constant (which depends on arrival and departure rates, but not the policy).

This is the key stochastic result of the paper. Our proof relies on a property satisfied by static priority policies, which we call \emph{independence of irrelevant alternatives}: removing a waiting agent from the history does not change what happens to any other agent.
Combining this property with the Slivnyak--Mecke formula for Poisson processes, we show that if $N_S$ is the number of waiting agents with types in $S$ in steady state, then the conditional distribution of $N_S$ given $N_S > 0$ is stochastically dominated by the same conditional distribution under the policy that makes no matches, in which $N_S$ is simply a Poisson random variable; comparing conditional means then yields the bound.

% With this lower bound in its availability constraints, $\plpm$ can certify the reward of a policy, but it does not say which policy to run. We give two answers, primal and dual, and then show that they are closely connected.

\medskip
\noindent\textbf{A new family of linear programs to identify a policy and bound its performance.} From here, our approach departs from prior work.
%\rn{Our optimal guarantee requires none of the machinery used in previous approaches---randomized thinning, LP-prescribed matching probabilities, or adaptation to queue lengths.}
Rather than using the solution of such a relaxation to define a randomized policy, we use a new family of linear programs, denoted by $\plpm$ and indexed by a set $\m{M}$ of allowed matches, to identify static priority policies. Unlike $\offrbrel$, these programs are not relaxations. Their values are \emph{lower} bounds on the rewards of the policies they identify, so each such policy and the omniscient benchmark are squeezed into a ``linear programming sandwich,'' with $\plpm$ below and $\offrbrel$ above. We further show that by properly selecting the set $\m{M}$ of allowed matches (including the case $\m{M}=\TT$, when all matches are allowed), the optimal objective values of these linear programs never differ by a factor greater than two.

\medskip
\noindent\textbf{The primal approach: priorities from tight constraints.}  To design a policy via the above ingredients, our first approach (\Cref{sec:primal}) uses the primal $\plpm$ directly and reads the priorities from its optimal solution. When the availability constraint of an arriving type and a set of types is tight, the program matches these arrivals with the set as often as the constraint allows, just as a static priority policy does when the set is a prefix of the arriving type's priority list. This suggests a simple recipe: turn tight sets into prefixes of priority lists. We call a policy \emph{consistent} with the solution if it follows this recipe and accepts only the matches that the solution uses. We show that every consistent policy earns at least the value of $\plpm$, even though it may make some matches less often than the program prescribes. Geometrically, the optimal solution lies on a line segment between the policy's long-run outcome, which need not be feasible, and a feasible solution of $\plpm$---see \Cref{fig:piercing}; since the latter earns no more than the optimum, the policy earns no less.

The catch is that a consistent policy may not exist. Using a diminishing-returns (submodularity) property of the availability constraints, we show that the tight sets of each arriving type are nested, so they always fit into a single list. The recipe therefore fails exactly when a tight set contains a type that the solution never matches with the arriving type, since this match would then have to be both accepted and rejected. We call solutions for which the recipe works \emph{suitable}, and we find one by \emph{pruning}: starting with all matches allowed, we repeatedly solve $\plpm$ and remove the matches that its solution does not use. We show that pruning never lowers the value of $\plpm$ and ends, after polynomially many rounds, with a suitable optimal solution; the same diminishing-returns property lets us solve each program in polynomial time, despite its exponentially many constraints.

For the set $\m{M}$ produced by pruning and any policy $\pi$ that the primal approach constructs from it, our results yield the following chain, where $\offopt$ denotes the omniscient benchmark and $\val{\pi}$ the long-run average reward of $\pi$:
\[
\offrbrel\;\ge\;\offopt\;\ge\;\val{\pi}\;\ge\;\plpm\;\ge\;\plpt\;\ge\;\half\offrbrel.
\]
The first inequality holds because $\offrbrel$ is a relaxation, and the second because the omniscient benchmark is at least as good as any policy. The last three are the main steps of the primal approach: a consistent policy earns at least the value of its program, pruning never lowers this value, and halving turns any solution of $\offrbrel$ into a solution of $\plpt$. Each inequality may be loose on a given instance; however, their combined loss is at most a factor of two. The primal approach, therefore, yields a $\half$-competitive policy in polynomial-time.

\medskip
\noindent\textbf{The dual approach: priorities from values.} Our second approach (\Cref{sec:values}) works directly with the value-based policies of \Cref{sec:main-result}, so the question becomes how to set the values. If they are too low, the policy uses up its agents on poor matches; if they are too high, it holds agents back until they abandon. To find the right balance, for every static priority policy $\pi$, we introduce a linear program $\plblp(\pi)$ that resembles $\plpm$ but is tied to $\pi$: it is a minimization program, and its availability constraints are flipped and imposed only on the prefixes of the priority lists of $\pi$. With these changes, it is easy to see that $\plblp(\pi)$ is a lower bound on the long-run average reward $\val{\pi}$ of $\pi$. We then use linear programming duality: the dual of $\plblp(\pi)$ gives each type a variable that plays the role of its value. Its constraints also identify a set of inequalities that tell us when a vector of values certifies that $\pi$ earns at least the dual objective $\sum_i\lambda_iv_i$ by weak duality. Interestingly, reversing these inequalities---including for the matches that $\pi$ never makes---gives a certificate in the other direction: we show that values that satisfy the reversed inequalities guarantee that the omniscient benchmark earns at most twice $\sum_i\lambda_iv_i$. A \emph{tight} certificate, which satisfies both sets of inequalities, therefore shows that $\pi$ earns at least half of the omniscient benchmark.

Which values admit a tight certificate? We show that they are exactly the solutions of a fixed-point equation $\mathbf{v}=\mathbf{F}(\mathbf{v})$ for a particular piecewise affine map $\mathbf{F}$. In economic terms, the equation balances, for each type, the rate at which a waiting agent loses value to abandonment against a pessimistic estimate of the rate at which its availability creates surplus (positive scores) for later arrivals, in the spirit of the balanced thresholds and prices used for prophet inequalities~\citep{kleinberg2012matroid,dutting2020prophet}. The fixed point is also unique and efficiently computable. The map $\mathbf{F}$ is not always a contraction, but we show that a damped version of it brings any two vectors of values closer together by a fixed factor under an appropriately weighted norm. Hence, the fixed-point equation has exactly one solution. Moreover, because $\mathbf{F}$ is piecewise affine, Ellipsoid-based algorithms for contractions~\citep{sikorski1993ellipsoid,daskalakis2026polynomial} find this solution in polynomial time. The dual approach, therefore, yields a $\half$-competitive value-based policy in polynomial time.

\medskip
\noindent\textbf{Primal meets dual.} One might wonder whether the two approaches above are connected in any way. The primal approach reads priorities from the tight constraints of $\plpm$, while the dual approach computes them from a fixed point, without solving $\plpm$ at all. Despite these differences, it turns out that the two approaches are mathematically linked through the dual of $\plpm$ (\Cref{sec:connections}). This dual assigns a shadow price to each type, and its constraints have the same form as the reversed inequalities of the dual approach, restricted to the matches in $\m{M}$. We show that the \emph{restricted} fixed point $\mathbf{v}^{\m{M}}$, which solves the fixed-point equation when only the matches in $\m{M}$ are allowed, is always feasible for this dual. Remarkably, suitability, which the primal approach critically needs to read a policy from an optimal solution, is \emph{exactly} what turns the shadow prices into fixed-point values: we show that if $\plpm$ has a suitable optimal solution, the optimal shadow prices are unique and equal $\mathbf{v}^{\m{M}}$, so solving the dual recovers the restricted fixed point. We also show by an example that, without suitability, they need not solve this equation.

For the set $\m{M}$ produced by pruning, the connection extends to policies. Through complementary slackness, we show that every policy that the primal approach constructs is indeed a value-based policy restricted to $\m{M}$: it ranks the matches in $\m{M}$ by their scores under $\mathbf{v}^{\m{M}}$. The primal approach is thus a value-based approach \emph{in disguise}, even though it never computes a value. Notably, this restricted value-based policy may differ from the unrestricted value-based policy of the dual approach, but it is still a $\half$-competitive policy, as explained earlier.

As a byproduct of the above investigation, we identify a broader class of (restricted) value-based policies that includes the policies of both approaches as special cases. Specifically, for every set of allowed matches $\m{M}$ with $\plpm\ge\plpt$, we show that every value-based policy restricted to $\m{M}$ whose values are the restricted fixed point $\mathbf{v}^{\m{M}}$ earns at least half of the omniscient benchmark. The policies of the dual approach are exactly the members of this class with $\m{M}=\TT$. The set produced by pruning also satisfies $\plpm\ge\plpt$, so, by the value-based interpretation above, every policy of the primal approach is also a member of this class. We view this class as a unifying lens on our two approaches and believe that it may be of independent interest.

\subsection{Organization}
%\rn{this could be dropped!}
The rest of the paper is organized as follows. \Cref{sec:preliminaries} includes the model, preliminaries, and backgrounds. \Cref{sec:priority-policies} defines static priority policies and establishes their key properties, including our lower bound on agent availability (\Cref{lem:gammabound}). \Cref{sec:primal} presents the primal approach, \Cref{sec:values} presents the dual approach, and \Cref{sec:connections} connects the two approaches through the primal-dual lens. \Cref{sec:conclusion} concludes. The appendices contain a discussion of further related work in operations research, economics, and computer science (\Cref{sec:related-work}), the definitions and proofs for Poisson processes (\Cref{sec:Poisson process}), the omitted proofs (\Cref{sec:omitted-proofs}), and additional examples (\Cref{sec:additional-examples}). %,  and a sketch of how our results extend to a large class of sojourn-time distributions (\Cref{sec:general-sojourn-time}).

\section{Preliminaries}
\label{sec:preliminaries}
In this section, we define our model (\Cref{sec:model-notation}), state some key facts about Poisson processes that our analysis uses (\Cref{sec:poisson-background}), and introduce the omniscient benchmark and a linear programming relaxation that upper-bounds it (\Cref{sec:benchmark-guarantee}).

%---------------------------------------------------------------------
\subsection{Model and Notation}
\label{sec:model-notation}

We study the continuous-time model of stochastic dynamic matching with abandonment
introduced by \citet{collina2020dynamic} and
\citet{aouad2022dynamic}. This model has Poisson arrivals, exponential departures, and pairwise matching. Formally, an
\emph{instance} is a tuple $\inst = (\types, \boldsymbol{\lambda}, \boldsymbol{\mu}, \mathbf{r})$, where
\begin{itemize}
    \item $\types$ is a finite set of agent types.
    \item $\boldsymbol{\lambda} \in \mathbb{R}_{>0}^\types$ specifies agent arrival rates. Agents of each type arrive according to Poisson processes that are independent across types, with $\lambda_i > 0$ denoting the arrival rate of type $i \in \types$.
    \item $\boldsymbol{\mu} \in \mathbb{R}_{>0}^\types$ specifies departure rates. Each agent has a {\em maximum sojourn time} (or simply \emph{sojourn time}), after which the agent will abandon the system if not matched. These times are independent of each other and of the arrival processes, and exponentially distributed. The expected maximum sojourn time for an agent of type $i \in \types$ is $1/\mu_i$.
    \item $\mathbf{r} \in \Rnn^{\types \times \types}$ specifies matching rewards. Each match involves a pair of agents. For any $i,j \in \types$, $r_{ij} \ge 0$ specifies the planner's reward for the match $(i,j)$, that is, from matching an agent of type $i$ to an agent of type $j$ who arrived later. We allow $\rew{i}{j} \neq \rew{j}{i}$, so the reward may depend on which agent has waited longer. (Since arrivals follow independent Poisson processes, almost surely no two agents arrive at the same time.) %\footnote{Prior work assumes symmetric rewards $r_{ij} = r_{ji}$. We allow our rewards to depend on which type arrived first.}
\end{itemize}
%Agents arrive over time, wait in a common pool, and eventually leave the market---either matched by the planner to another agent or, after waiting too long, abandoning the market unmatched. The planner wishes to maximize the long-run average reward. Throughout, $i$ and $j$ denote types and $a$ and $b$ denote individual agents.

%\paragraph{Arrivals, abandonment, and histories.}
%Agents of type $i \in \types$ arrive according to a Poisson process with rate $\arate{i} > 0$ (see \Cref{sec:Poisson process} for a formal definition). Arrivals are independent across types. 

For $t \ge 0$ and $i \in \types$, let
$\agentset_i(t)$ be the set of type-$i$ agents that have arrived before time $t$,
and $\agentset(t) = \bigcup_{i \in \types}\agentset_i(t)$ the set of all
arrivals. For an agent $a$, let $\atype{a}$ denote its type, $\atime{a}$ denote its arrival time, $\soj{a}$ denote its maximum sojourn time, and $\dtime{a} = \atime{a} + \soj{a}$
denote the time at which $a$ will depart if not matched. We next define the notions of \emph{present} and \emph{available} agents.
%The planner does not observe sojourn times and cannot anticipate when an unmatched agent will abandon.
%Each agent $a$ of type $i$ has a \emph{sojourn time} $\soj{a}$: if it is not matched within $\soj{a}$ time units of its arrival, it \emph{abandons} the market unmatched. Type-$i$ agents abandon at rate $\drate{i} > 0$, i.e., the sojourn times $\{\soj{a}\}$ are independent across agents and exponentially distributed with mean $1/\drate{i}$. We let $\dtime{a} = \atime{a} + \soj{a}$ denote the time at which $a$ would depart if not matched. 
\begin{definition}[Present/available agents]\label{def:present-available}
An agent $a$ is \textbf{present} at time $t$ if $\atime{a} \le t < \dtime{a}$, that is, if it has arrived and has not reached its departure time. Under an online policy, which is a policy that makes irrevocable matches over time (see below), an agent is \textbf{available} at time $t$ if it is present and has not yet been matched.
\end{definition}
These definitions of {\em present} and {\em available} are consistent with the use of these terms in past work---see, for example \citet{kessel2022stationary,patel2024combinatorial}. 
Presence depends only on the arrival and departure times, and not on the matching policy. Availability, in contrast, depends on the policy: a present agent stops being available as soon as the policy matches it.

We now define matching policies formally. The \emph{history} up to time $t$,
\[
  \hist(t) = \{(\atype{a}, \atime{a}, \dtime{a})\}_{a \in \agentset(t)},
\]
records the type, arrival time, and departure time of every agent that has arrived before time $t$; for notational simplicity, it includes departure times that occur after~$t$. From the history $\hist(t)$, we define the \emph{compatibility graph} $\graph(t)$ with vertices $\agentset(t)$, and an edge between agents $a$ and $b$ if there is a time at which they are both present. Each edge has a weight, which is simply the reward earned from matching the two agents involved. A \emph{matching} of $\graph(t)$ is a subset of edges with no two edges sharing a vertex, and its weight is the sum of the rewards on those edges.

Given the above primitives, a  policy \emph{with hindsight information} is a  function $\pi$ that maps each history $\hist(t)$ to a matching of $\graph(t)$. This definition allows the policy to make decisions for early arrivals based on future arrival and departure information, and does not require the policy to be consistent across
horizons.\footnote{For instance, if only two agents have arrived by time $t$, a
policy may match them to each other when optimizing for horizon $t$, yet leave them for other partners when optimizing for a later horizon. This gives great power to our omniscient benchmark, while static priority policies do not use this power.  %Such a policy may be even more powerful than an online clairvoyant policy that has look-ahead access to the entire future arrivals and departures.
} An \emph{online}  policy, in contrast, makes irrevocable matches between available agents (\Cref{def:present-available}) using only the information revealed so far: the types and arrival times of the agents that have arrived, and the departures that have already occurred. In particular, it learns the departure time of an agent only when the agent abandons. The matches that an online policy---possibly randomized---makes during $[0,t]$ form a matching of $\graph(t)$ in each realization.\footnote{
Models that instead
reveal when an agent is about to abandon \citep{akbarpour2020thickness,
huang2018match}, or that disallow abandonment altogether
\citep{kerimov2023optimality}, differ substantially from ours; see \Cref{sec:related-work}. 
}

Fix a policy $\pol$, types $i,j \in \types$ and a horizon $t > 0$. Let
$\Xmatch{i}{j}^{\pol}(t)$ denote the number of $(i,j)$ matches that $\pi$
makes during $[0,t]$ (which is a random variable). Let $\xrate{i}{j}^{\pol}(t)$ be the corresponding
\emph{expected match rate} over the period $[0,t]$, where the expectation is over the arrival and abandonment processes:
\begin{align}
  \Xmatch{i}{j}^{\pol}(t) &= \big|\{ (a,b) \in \pol(\hist(t)) :
      \atype{a}=i,\ \atype{b}=j,\ \atime{a} < \atime{b} \}\big|,
  \nonumber\\
  \xrate{i}{j}^{\pol}(t) &= \tfrac{1}{t} \E{\Xmatch{i}{j}^{\pol}(t)}.
  \label{eq:xijt}
\end{align}
%where the expectation is over the arrival and abandonment processes and any randomness in $\pol$. 
From this, we can compute the
\emph{expected reward rate} of $\pol$ on $[0,t]$,
\begin{equation}
  \valt{\pol} = \sum_{i,j \in \types} \rew{i}{j} \xrate{i}{j}^{\pol}(t).
\end{equation}
We evaluate a policy by its long-run average reward $\val{\pol}=\liminf_{t\to\infty}\valt{\pol}$; for the static priority policies of \Cref{sec:priority-policies}, this limit exists (\Cref{lemma:stationaryctmc}).

\subsection{Background on Poisson Processes}
\label{sec:poisson-background}

Our analysis relies on three results about Poisson processes, which we collect here. \Cref{sec:Poisson process} contains the formal definition of a Poisson process and some associated proofs. Informally, a Poisson process on a set $\mathcal{X}\subseteq\reals^n$ with rate function $\lambda(\cdot)$ is a random collection of points in $\mathcal{X}$ such that the numbers of points in disjoint regions are independent Poisson random variables, and the expected number of points in a region is the integral of $\lambda(\cdot)$ over that region (\Cref{def:Poisson process}). In our model, the arrival times of type-$i$ agents form a Poisson process on $[0,\infty)$ with constant rate $\lambda_i$.

% [Version 5, second round: the marking fact (formerly Fact 2.1) moved to Appendix A, and the occupancy lemma (formerly Lemma A.1) moved here in its place; its proof remains in Appendix A.]
The first result concerns the agents that are present at a given time. It follows from the marking property of Poisson processes (\Cref{fact:productppp} in \Cref{sec:Poisson process}), which, applied to the arrivals of type $i$ with their sojourn times as marks, shows that the pairs of arrival times and sojourn times of type-$i$ agents form a Poisson process on $[0,\infty)^2$ with rate function $(s,\ell)\mapsto\lambda_i\mu_ie^{-\mu_i\ell}$.

\begin{restatable}{lemma}{occupancy}
\label{fact:poisson-occupancy}
The number of type-$i$ agents present at time $t$ is Poisson with mean
\[
  \tfrac{\arate{i}}{\drate{i}}\big(1-e^{-\drate{i}t}\big).
\]
These counts are independent across types.
\end{restatable}

Since presence does not depend on the policy, and an agent can be available only if it is present, this lemma limits the availability of agents under every policy. We prove it in \Cref{sec:Poisson process}.

The second result is a formula for the expected value of a sum over the points of a Poisson process~\citep{slivnyak1962some,mecke1967stationare}; see also \citet[Theorem~4.1]{lastpenrose2017lectures}. Intuitively, it says that, conditional on a Poisson process having a point at a given location, the other points still form a Poisson process with the same rate function. Below, $\mathcal{N}(\mathcal{X})$ denotes the collection of locally finite subsets of $\mathcal{X}$.

% [Version 2: the following fact was moved here from Appendix A (formerly Fact A.2).]
\begin{fact}[Campbell--Mecke formula, Mecke equation, or Slivnyak--Mecke formula]
\label{fact:poisson-arrival-calculus}
Let $\Phi$ be a Poisson process on a Borel set
$\mathcal{X}\subseteq\mathbb R^n$ with rate function $\lambda(x)$, and let
$g:\mathcal{X}\times\mathcal{N}(\mathcal{X})\to\Rnn$ be a measurable function of a distinguished
point and a point set. Then
\[
  \E{\sum_{x\in\Phi} g(x,\Phi \setminus \{x\})}
  =
  \int_{\mathcal{X}} \E{g(x,\Phi)}\lambda(x)dx.
\]
\end{fact}

\Cref{fact:poisson-arrival-calculus} is the main tool in the proof of our key stochastic result, \Cref{lem:gammabound} in \Cref{sec:priority-policies}. Together with the marking property (\Cref{fact:productppp}), it also implies the third result, which is known as PASTA, for ``Poisson arrivals see time averages''~\citep{wolff1982poisson}: in the long run, the fraction of arrivals of a given type that find the system in a given set of states equals the fraction of time that the system spends in these states. To state it, for an online policy $\pol$, let $N^{\pol}_i(t)$ be the number of type-$i$ agents available just before time $t$, and let $\mathbf{N}^{\pol}(t)=(N^{\pol}_i(t))_{i\in\types}$ be the \emph{state} of the system.

% [Version 2: the following corollary was moved here from Appendix A (formerly Corollary A.3); its proof is in Appendix A.]
\begin{restatable}[PASTA]{corollary}{pasta}
\label{cor:pasta}
Let $\pol$ be an online policy, and fix a type $j$ and a set of states $C$. Then for every $t\ge 0$,
\[
  \E{\sum_{a\in\agentset_j(t)} \indic{\mathbf{N}^{\pol}(t_a)\in C }}
  =
  \arate{j}\int_0^t \PP{\mathbf{N}^{\pol}(s) \in C}ds.
\] 
In particular, if $\PP{\mathbf{N}^{\pol}(s) \in C} \to p$ as $s\to\infty$, then
\[
  \lim_{t \to \infty}\frac{\E{\sum_{a\in\agentset_j(t)} \indic{\mathbf{N}^{\pol}(t_a)\in C }}}{t}
  = \lambda_jp.
\]
\end{restatable}

\subsection{Benchmark and Guarantee}
\label{sec:benchmark-guarantee}

\paragraph{The omniscient benchmark.}
We evaluate policies by comparing their long-run average reward against an \emph{omniscient}
benchmark: the offline optimum that observes the entire realized sequence of
arrivals and abandonments and matches optimally in hindsight. For each horizon $t$, the largest reward attainable by time $t$ is the weight of a maximum-weight matching of the compatibility graph $\graph(t)$; see \Cref{fig:omn-benchmark} for an illustration.

\begin{definition}[Omniscient benchmark]
\label{def:omniscient}
For an instance $\inst$ and time $t > 0$, let $\offoptt$ be the expected weight
of a maximum-weight matching of $\graph(t)$, divided by $t$. The
\textbf{omniscient benchmark} is $\offopt = \limsup_{t \to \infty} \offoptt$.
\end{definition}

% ---------- colours ----------
\definecolor{typefill}{RGB}{191,224,231}
\definecolor{typedraw}{RGB}{60,96,110}
\definecolor{edgecol}{RGB}{125,178,193}
\definecolor{matchcol}{RGB}{176,30,36}
\definecolor{panelframe}{RGB}{190,198,204}
\definecolor{paneltitle}{RGB}{238,243,246}

% ---------- inline type symbols ----------
% Each symbol is centered 0.7ex above the baseline, so the circle, diamond and
% triangle line up with each other and with the digits in each row.
\newcommand{\typecircle}{%
  \tikz[baseline={([yshift=-0.7ex]current bounding box.center)}]\node[circle,draw=typedraw,line width=0.5pt,fill=typefill,minimum size=1.2em,inner sep=0pt]{};%
}
\newcommand{\typediamond}{%
  \tikz[baseline={([yshift=-0.7ex]current bounding box.center)}]\node[diamond,draw=typedraw,line width=0.5pt,fill=typefill,aspect=1,minimum size=1.3em,inner sep=0pt]{};%
}
\newcommand{\typetriangle}{%
  \tikz[baseline={([yshift=-0.7ex]current bounding box.center)}]\node[regular polygon,regular polygon sides=3,draw=typedraw,line width=0.5pt,fill=typefill,minimum size=1.4em,inner sep=0pt]{};%
}
\newcommand{\smalltypecircle}{%
  \tikz[baseline=-0.45ex]\node[circle,draw=typedraw,line width=0.45pt,fill=typefill,minimum size=0.9em,inner sep=0pt]{};%
}
\newcommand{\smalltypediamond}{%
  \tikz[baseline=-0.5ex]\node[diamond,draw=typedraw,line width=0.45pt,fill=typefill,aspect=1,minimum size=1.0em,inner sep=0pt]{};%
}
\newcommand{\smalltypetriangle}{%
  \tikz[baseline=-0.62ex]\node[regular polygon,regular polygon sides=3,draw=typedraw,line width=0.45pt,fill=typefill,minimum size=1.05em,inner sep=0pt]{};%
}
% priority symbol with even spacing between type symbols
\newcommand{\prioarrow}{\,$\succ$\,}

% ---------- graph node styles ----------
\tikzset{
  circnode/.style={circle,draw=typedraw,line width=0.7pt,fill=typefill,minimum size=7.2mm,inner sep=0pt},
  diamondnode/.style={diamond,draw=typedraw,line width=0.7pt,fill=typefill,aspect=1,minimum size=8.4mm,inner sep=0pt},
  trianglenode/.style={regular polygon,regular polygon sides=3,draw=typedraw,line width=0.7pt,fill=typefill,minimum size=10.2mm,inner sep=0pt},
  compatedge/.style={draw=edgecol,line width=1.3pt},
  matchedge/.style={draw=matchcol,line width=2.2pt},
  weight/.style={font=\normalsize,fill=white,inner sep=1.2pt,rounded corners=1pt}
}

% ---------- panel wrapper ----------
% \panelbox[<tcolorbox options>]{<width>}{<title>}{<content>}
\newcommand{\panelbox}[4][]{%
  \begin{tcolorbox}[enhanced,nobeforeafter,box align=top,valign=center,width=#2,
    colback=white,colframe=panelframe,boxrule=0.6pt,arc=1.5mm,
    left=5pt,right=5pt,top=5pt,bottom=4pt,
    colbacktitle=paneltitle,coltitle=black,fonttitle=\small\bfseries,
    halign title=center,toptitle=2pt,bottomtitle=2pt,lefttitle=3pt,righttitle=3pt,titlerule=0.4pt,
    title={\strut #3},#1]
  #4
  \end{tcolorbox}%
}

% ---------- panel contents ----------
\newcommand{\AgentTablePanel}{%
\centering\small
\setlength{\tabcolsep}{2.4pt}%
\renewcommand{\arraystretch}{1.12}%
\begin{tabular}{cccc}
\toprule
Agent & Type & Arrival & Departure\\
\midrule
$A$ & \smalltypecircle & $0.3$ & $2.9$\\
$B$ & \smalltypediamond & $1.7$ & $4.8$\\
$C$ & \smalltypetriangle & $2.5$ & $4.5$\\
$D$ & \smalltypediamond & $4.2$ & $5.8$\\
$E$ & \smalltypecircle & $5.3$ & $5.6$\\
\bottomrule
\end{tabular}%
}

\newcommand{\RewardTablePanel}{%
\centering\small
\setlength{\tabcolsep}{5pt}%
\renewcommand{\arraystretch}{1.3}%
\begin{tabular}{c|ccc}
 & \typecircle & \typediamond & \typetriangle\\[1pt]
\hline
\typecircle & $0$ & $5$ & $4$\\
\typediamond & $5$ & $8$ & $10$\\
\typetriangle & $4$ & $10$ & $0$\\
\end{tabular}%
}

% #1: style of the B--C edge (compatedge in Figure 2, matchedge in Figure 3)
\newcommand{\GraphPanel}[1][compatedge]{%
\centering
\resizebox{\linewidth}{!}{%
\begin{tikzpicture}
  \node[circnode] (A) at (0.3,0) {$A$};
  \node[diamondnode] (B) at (1.7,1.75) {$B$};
  \node[trianglenode] (C) at (2.6,0) {$C$};
  \node[diamondnode] (D) at (4.2,1.75) {$D$};
  \node[circnode] (E) at (5.4,0) {$E$};
  \draw[compatedge] (A) -- node[weight,above left=-1pt and 0pt,fill=none] {$5$} (B);
  \draw[compatedge] (A) -- node[weight,below=0.5pt,fill=none] {$4$} (C);
  \draw[#1] (B) -- node[weight,pos=0.48] {$10$} (C);
  \draw[compatedge] (B) -- node[weight,above=0.5pt,fill=none] {$8$} (D);
  \draw[compatedge] (C) -- node[weight,pos=0.52] {$10$} (D);
  \draw[compatedge] (D) -- node[weight,above right=-1pt and 0pt,fill=none] {$5$} (E);
\end{tikzpicture}}%
}

\newcommand{\PolicyTablePanel}{%
\centering\small
\setlength{\tabcolsep}{2.4pt}%
\renewcommand{\arraystretch}{1.6}%
\begin{tabular}{cl}
\toprule
Type $j$ & Priority list $\succ_j$\\
\midrule
\smalltypecircle & \smalltypetriangle\prioarrow$\varnothing$\\
\smalltypetriangle & \smalltypediamond\prioarrow\smalltypecircle\prioarrow$\varnothing$\\
\smalltypediamond & \smalltypetriangle\prioarrow\smalltypediamond\prioarrow$\varnothing$\\
\bottomrule
\end{tabular}%
}

\newcommand{\WalkthroughPanel}{%
\raggedright\footnotesize
\begin{enumerate}[leftmargin=*,itemsep=1.5pt,topsep=0pt,parsep=0pt,partopsep=0pt,label=\arabic*.]
\item \textbf{$A$ arrives}. No available \smalltypetriangle, so $A$ waits.
\item \textbf{$B$ arrives}. No available \smalltypetriangle{} or \smalltypediamond, so $B$ waits.
\item \textbf{$C$ arrives}. It prioritizes \smalltypediamond{} over \smalltypecircle, so $C$ matches with $B$.
\item \textbf{$A$ abandons}.
\item \textbf{$D$ arrives}. No available agents, so $D$ waits.
\item \textbf{$E$ arrives}. Only $D$ is available, and \smalltypediamond{} is not acceptable, so $E$ waits.
\end{enumerate}%
}

% ----------------------------------------------------------------
% Figure 1: OMN / omniscient policy
% Reorder the three panel macros below if you want a different left-to-right order.
% ----------------------------------------------------------------
\begin{figure}[!t]
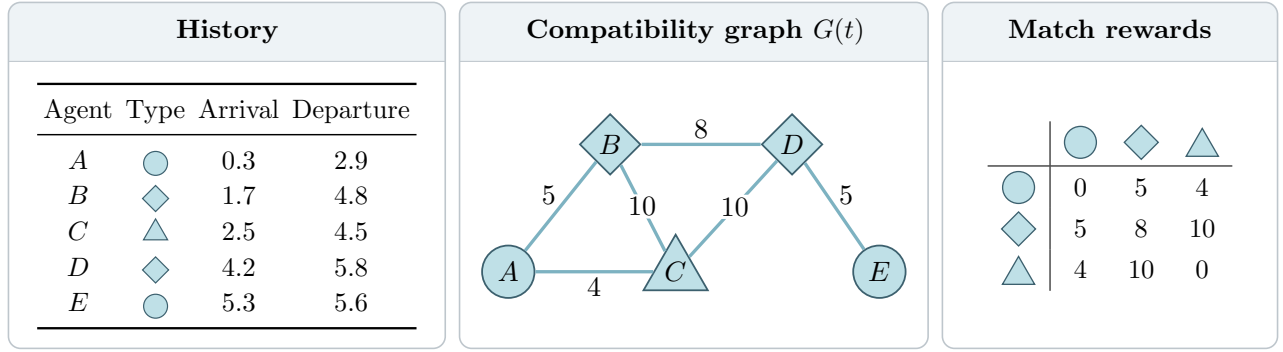

\centering
\panelbox[equal height group=omnexample]{0.345\textwidth}{History}{\AgentTablePanel}\hfill
\panelbox[equal height group=omnexample]{0.37\textwidth}{Compatibility graph $G(t)$}{\GraphPanel}\hfill
\panelbox[equal height group=omnexample]{0.265\textwidth}{Match rewards}{\RewardTablePanel}
\caption{An example with three types and five arriving agents. Given the history shown on the left, the corresponding graph $G(t)$ is shown in the middle, with edge weights determined by the match rewards. There are two maximum-weight matchings: $\{(A,B),(C,D)\}$ and $\{(B,C),(D,E)\}$, both with weight $15$; $\textup{OMN}(t)$ is the expected weight of a maximum-weight matching of $G(t)$, divided by $t$.}
\label{fig:omn-benchmark}
\end{figure}

Because any matching policy must choose a matching of $G(t)$, the omniscient benchmark bounds the performance of any policy, including the optimal online policy.

\begin{definition}[Competitive ratio]
\label{def:competitive-ratio}
A policy $\pol$ is \textbf{$c$-competitive}, for $c \in [0,1]$, if for every instance $\inst$, $\val{\pol} \ge c \cdot \offopt$. Its
\textbf{competitive ratio} is the largest $c$ for which this guarantee holds, 
$\inf_{\inst} \val{\pol}/\offopt$.
\end{definition}
%\na{This definition is fine, but we don't use the phrase ``competitive ratio'' much in the rest of the paper.}

Our goal is to find an online policy whose competitive ratio is as large as possible.
\paragraph{A linear programming relaxation of the benchmark.}
% The benchmark $\offopt$ is itself difficult to analyze ---it is the expected weight of a maximum-weight matching of a random graph---so we work instead with the linear
% program below, which upper-bounds it. The program has one variable $\xrate{i}{j} \ge 0$ per ordered pair of types, interpreted as the long-run
% rate of $(i,j)$ matches, and one variable $n_i \ge 0$ for each type, interpreted as the long-run average number of type $i$ agents waiting in the pool.
The benchmark $\offopt$ is itself difficult to analyze---it is the expected weight of a maximum-weight matching of a random graph---so we work instead with the linear
program below, which upper-bounds it. The program has one variable $\xrate{i}{j} \ge 0$ per ordered pair of types, interpreted as the long-run
rate of $(i,j)$ matches. %, and one variable $\alpha_i \ge 0$ for each type, interpreted as the long-run rate at which type-$i$ agents leave unmatched. 

%\na{Why do we introduce $\alpha_i$?}
\begin{lpbox}{Upper-bound LP --- $\offrbrel$}
\begin{align}
  \offrbrel \;\triangleq\; \max_{\bx \ge 0}\quad
   & \sum_{i,j \in \types} \rew{i}{j}\xrate{i}{j}
   \nonumber \\[2pt]
  \textrm{s.t.}\quad
   & \sum_{i \in \types} \xrate{i}{j} + \sum_{i \in \types} \xrate{j}{i}
     \;\le\; \arate{j}
   & &\forall j \in \types, \label{eq:omn-balance}\tag{\textsc{OMN-Supply}}\\
   & \sum_{i \in S} \xrate{i}{j}
     \;\le\; \arate{j}\Big( 1 - e^{-\sum_{i \in S} \arate{i}/\drate{i}} \Big)
   & &\forall j \in \types,\ S \subseteq \types. \label{eq:omn-match}\tag{\textsc{OMN-Presence}}
\end{align}
\end{lpbox}

% \begin{lpbox}{Upper-bound LP --- $\offrbrel$}
% \begin{align}
%   \offrbrel \;\triangleq\; \max_{\bn,\bx \ge 0}\quad
%    & \sum_{i,j \in \types} \rew{i}{j}\xrate{i}{j}
%    \nonumber \\[2pt]
%   \textrm{s.t.}\quad
%    &  n_i \mu_i +\sum_{i \in \types} \xrate{i}{j} + \sum_{i \in \types} \xrate{j}{i}
%      \;=\; \arate{j}
%    & &\forall j \in \types, \label{eq:omn-balance}\tag{\textsc{Balance-I}}\\
%    & \sum_{i \in S} \xrate{i}{j}
%      \;\le\; \arate{j}\Big( 1 - e^{-\sum_{i \in S} \arate{i}/\drate{i}} \Big)
%    & &\forall j \in \types,\ S \subseteq \types. \label{eq:omn-match}\tag{\textsc{Availability-I}}
% \end{align}
% \end{lpbox}

\noindent
%\na{I propose that we introduce language to describe these constraints. For example, we could call \eqref{eq:omn-balance} "supply" constraints and \eqref{eq:omn-match} "availability" constraints. I am open to other names as well. But these names suggest the meaning of the constraint, and allow us to say things like "take $v_i$ to be the dual variables associated with the supply constraints", or "look at the tight availability constraints".}
%Constraint~\eqref{eq:omn-balance} is a flow-balance bound: the total rate of matches involving a type-$j$ agent cannot exceed its arrival rate $\arate{j}$. Constraint~\eqref{eq:omn-match} bounds the rate at which arriving type-$j$ agents match a previously available agent of a type in $S$. Such a match requires that some agent of a type in $S$ be present, an event whose probability we can calculate from \cref{fact:poisson-occupancy} which showed that distribution of the number of such agents is Poisson with mean $\tfrac{\arate{i}}{\drate{i}}\big(1-e^{-\drate{i}t}\big)$. The complete proof of \Cref{prop:offline-bound} can be found in \Cref{apx:offline-bound}. We present a proof sketch that captures the main idea behind this result. 

With these variables, Constraint~\eqref{eq:omn-balance} is a flow-balance constraint; in particular, the total rate of matches involving type-$j$ agents is no more than the rate $\arate{j}$ for arrivals of type-$j$ agents.\footnote{As we will see in \Cref{sec:priority-policies}, this constraint in the equality from---which can be equivalently derived by introducing new variables $\boldsymbol{\alpha}$ for its non-negative slack, where $\alpha_j$ is interpreted as the total rate at which type-$j$ agents leave unmatched---parallels the balance equation of static priority policies.} Constraint~\eqref{eq:omn-match}
bounds the rate at which arriving type-$j$ agents are matched with earlier arrivals whose types are in $S$. The following proposition shows that $\offrbrel$ upper-bounds the omniscient benchmark. We sketch its proof below and postpone the full proof to \Cref{apx:omittedproofs-prelim}. 

% Constraint~\eqref{eq:omn-balance} is a flow-balance bound: the total rate of
% matches involving a type-$j$ agent plus the total rate at which type-$j$ agents leave unmatched cannot exceed its arrival rate $\arate{j}$. \na{This description matches the version of the constraint without $\alpha_i$; with $\alpha_i$, we might want to say total matches + abandonment = arrivals. } Constraint~\eqref{eq:omn-match}
% bounds the rate at which arriving type-$j$ agents match a previously available
% agent of a type in $S$. Such a match requires that some agent of a type in $S$ be present, an event whose probability can be calculated exactly, and shown to be at most $1 - e^{-\sum_{i \in S} \lambda_i/\mu_i}$. This is established in \Cref{prop:offline-bound}. We present a proof sketch below, and the full proof is postponed to \Cref{apx:omittedproofs-prelim}. 

% \na{The paragraph above and proof sketch below are saying many of the same things, and maybe could be combined.}

\begin{restatable}{proposition}{offlinebound}
\label{prop:offline-bound}
For every instance $\inst$, $\offopt \le \offrbrel$.
\end{restatable}
%\todo{Return and modify, explicitly addressing what Poisson Facts we are using.}
\begin{proof}[Proof sketch]
It suffices to show that for every policy $\pol$ (even with hindsight information) and every horizon $t$, the match rates $\xrate{i}{j}^{\pol}(t)$ satisfy \eqref{eq:omn-balance} and \eqref{eq:omn-match}. Constraint~\eqref{eq:omn-balance} holds because each agent is matched at most once, so the rate of matches involving type-$j$ agents is at most $\arate{j}$; the variable $\alpha_j$ fills the slack. Turning to \eqref{eq:omn-match}, an arriving type-$j$ agent can contribute to the left side of \eqref{eq:omn-match} only if some earlier arrival with a type in $S$ is still present upon its arrival. By \cref{fact:poisson-occupancy} (see \Cref{sec:poisson-background}), the number of type-$i$ agents present at time $s$ is Poisson with mean $\tfrac{\arate{i}}{\drate{i}}(1 - e^{-\drate{i} s}) \le \arate{i}/\drate{i}$, independently across types. Since a sum of independent Poisson random variables is Poisson, the probability that no earlier arrival with a type in $S$ is present at time $s$ is at least $e^{-\sum_{i \in S}\arate{i}/\drate{i}}$. By \Cref{cor:poisson-arrival-count}, a version of PASTA property (\Cref{cor:pasta}) for events that depend on the history up to any arbitrary time $t$, the expected number of type-$j$ agents that arrive during $[0,t]$ and find some earlier arrival with a type in $S$ still available (and hence present) is therefore at most $\arate{j}t\big(1 - e^{-\sum_{i \in S}\arate{i}/\drate{i}}\big)$. Dividing by $t$ gives~\eqref{eq:omn-match}.
\end{proof}

The relaxation $\offrbrel$ is tighter than those used in prior work, which bound matches with only one type at a time: \citet{kessel2022stationary} use the singleton cases $|S|=1$ of \eqref{eq:omn-match}, and \citet{collina2020dynamic} use a weaker form of these singleton constraints. Constraints for larger sets matter: an arriving agent can be matched with at most one partner, and the probability that some type in $S$ is present can be much smaller than the sum of the probabilities that each type in $S$ is present. 
% [Version 2: the sentence above addresses the following comment.]
% \na{Could say a little more.} 
The work of  \citet{amanihamedani2024improved} allows for all subsets $S$, but assumes a restricted setting where types are partitioned into two sides, all valid matches involve one type from each side, and one side must be matched immediately upon arrival.

%\na{Could say a little more.} 

\section{Static Priority Policies}
\label{sec:priority-policies}
We now introduce the class of simple deterministic policies that are the focus of this paper. Formally, a \emph{static priority policy} $\pol$ specifies, for each type $j \in \types$, a strict order $\porder{j}$ over $\types \cup \{\emptyset\}$, called the \emph{priority list} of type $j$. We say that type $i$ is {\em acceptable} to type $j$ if $i \succ_j \emptyset$, and let 
\[\m{M}(\pi) = \{(i,j) : i \succ_j \emptyset\}\]
be the set of acceptable matches. The policy will never make an unacceptable match.

Recall from \Cref{def:present-available} that an agent is available at time $t$ if it is present and has not yet been matched. We say that a {\em type} is  available if there is an available agent of that type. When an agent of type $j$ arrives, the policy matches this agent to the acceptable and available type with the highest priority (according to $\succ_j$). If  no acceptable type is available, the arriving agent joins the pool and waits. Therefore, every match made by a static priority policy pairs a newly arrived agent with a waiting agent (two waiting agents are never matched to each other). Ties among available agents of the same type are broken by an arbitrary but fixed rule.\footnote{\label{footnote:tie-break}For example, among agents of a given type, we might match the earliest arrival. The choice of tie-breaking rule does not affect any of our analysis or results, which depend only on the number of available agents of each type $\mathbf{N}^{\pol}$.} \Cref{fig:static-priority} illustrates how a static priority policy operates on a simple example.

% ----------------------------------------------------------------
% Figure 2: static priority policy
% Reorder the three panel macros below if you want a different left-to-right order.
% ----------------------------------------------------------------

\begin{figure}[!t]
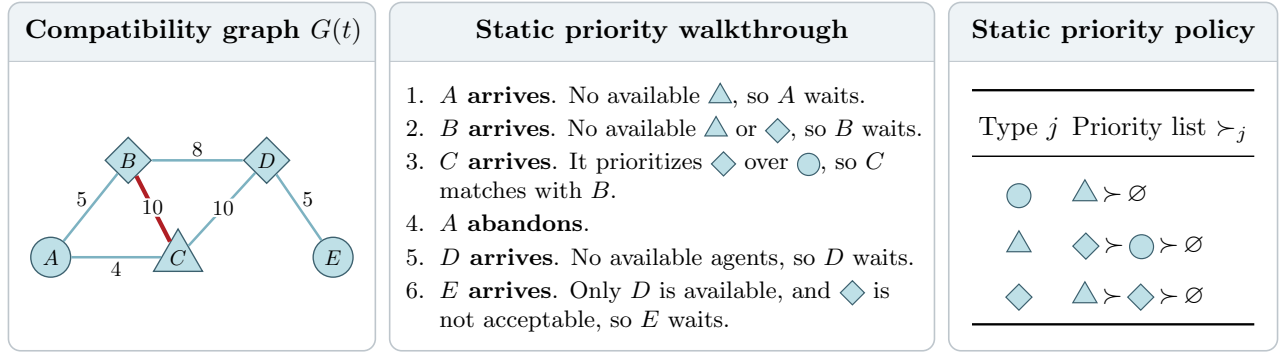

\centering
\panelbox[equal height group=spexample]{0.29\textwidth}{Compatibility graph $G(t)$}{\GraphPanel[matchedge]}\hfill
\panelbox[equal height group=spexample,left=3pt,right=3pt]{0.43\textwidth}{Static priority walkthrough}{\WalkthroughPanel}\hfill
\panelbox[equal height group=spexample]{0.26\textwidth}{Static priority policy}{\PolicyTablePanel}
\caption{Illustration of a static priority policy on the same example from \Cref{fig:omn-benchmark}, with five agents belonging to three types. The policy is given in the box on the right. The middle panel explains how the policy operates on this example; the policy's only match, between $B$ and $C$, is highlighted in red in the graph on the left. The policy earns a total reward of 10, whereas the maximum possible reward is 15. The best choice of policy depends on the arrival rates, departure rates, and rewards. We show how to identify a policy that earns at least  half of the omniscient benchmark.}
\label{fig:static-priority}
\end{figure}

\paragraph{Steady state.}
Under a static priority policy, whether a match occurs depends only on which
types are currently available, so the count of available agents of each type  evolves as a
Markov process. Let $\Npres{i}^{\pol}(t)$ be the number of type-$i$ agents available just before time $t$, as in \Cref{sec:poisson-background}. We refer to 
$\mathbf{N}^{\pol}(t) = (\Npres{i}^{\pol}(t))_{i \in \types}$ as the \emph{state}. We show that this state converges to a stationary distribution.

\begin{restatable}{lemma}{stationaryctmc}
\label{lemma:stationaryctmc}
For any static priority policy $\pol$, the state $\mathbf{N}^{\pol}(t)$ is a continuous-time Markov chain. This chain is ergodic, and has a unique stationary distribution. Therefore, the following limits exist:
\begin{equation}\label{eq:nxpi}
\nss{i}^{\pol} \triangleq
  \lim_{t \to \infty} \E{\Npres{i}^{\pol}(t)},
  \quad \xrate{i}{j}^{\pol} \mathrel{\triangleq} \lim_{t \to \infty} \tfrac{1}{t}
      \E{\Xmatch{i}{j}^{\pol}(t)}, \quad \val{\pi} \triangleq\lim_{t \to \infty} \valt{\pi}
\end{equation}
\end{restatable}

These quantities will be our key summary statistics of interest. The quantity $n_i^\pi$ is the expected (or average) number of available agents of type $i$. Meanwhile, the quantity $x_{ij}^\pi$ is the long-run rate at which the policy makes $(i,j)$ matches. The full proofs of the results of this section are in \Cref{apx:omittedproofs-priority}.  

We now identify some key properties of static priority policies. These will be used to motivate two novel families of linear programs: $\plpm$ (in \Cref{sec:primal}) and $\plblp(\pi)$ (in \Cref{sec:values}). We then use these linear programs to identify a good policy and bound its performance.

%The proof is postponed to \cref{apx:stationaryctmc}.

%We now prove two key properties of static priority policies on which the analysis rests. The first is a system of balance equations relating arrival, match, and abandonment rates in steady state.

\begin{definition}[Prefix]
\label{def:prefix}
A set $S \subseteq \types$ is a \textbf{prefix} of an order $\succ$ over
$\types \cup \{\emptyset\}$ if there is an acceptable type $i \succ \emptyset$
with $S = \{ i' \in \types : i' \succeq i \}$. We write $\prefixes{\succ}$
for the collection of all prefixes of $\succ$, and for a static priority policy $\pol=\{\porder{j}\}_{j\in\types}$ we write $\mathcal{P}_j(\pol)=\prefixes{\porder{j}}$.
\end{definition}

\begin{restatable}{lemma}{balancemultiple}
\label{lemma:balance:multiple}
Let $\pol = \{\porder{j}\}_{j \in \types}$ be a static priority policy and let
$\mathbf{N}^{\pol}$ be drawn from its stationary distribution. Then %$V(\pi) = \sum_{i,j\in \types} r_{ij} x_{ij}^\pi$, and 
the following hold:
\begin{alignat}{2}
 & V(\pi)
   = \sum_{i,j\in\types} r_{ij} x_{ij}^\pi
 && \label{eq:Vpi}\\[2pt]
 & \nss{j}^{\pol}\drate{j}
   + \sum_{i\in\types}
     \xrate{i}{j}^{\pol}+ \sum_{i\in\types}\xrate{j}{i}^{\pol}
   = \arate{j}
 &\qquad & \forall j\in\types,
 \label{eq:balance}\tag{\textsc{SP-Balance}}\\[2pt]
 & \sum_{i\in S}\xrate{i}{j}^{\pol}
   \le \arate{j}
      \PP{\textstyle\sum_{i\in S}\Npres{i}^{\pol}>0}
 &\qquad & \forall j\in\types,\ S\subseteq\types.
 \label{eq:matchrate}\tag{\textsc{SP-Availability}}
\end{alignat}
Moreover, \eqref{eq:matchrate} holds with equality whenever $S \in \prefixes{\porder{j}}$.

%\na{Might be nice to note that $x_{ij} > 0 \iff i \succ_j \emptyset$, but not necessary.}
%, and if \eqref{eq:matchrate} holds with equality, then $\xrate{i}{j}^{\pol} > 0$ for every $i \in S$.
\end{restatable}

%\pn{The “only if” direction of showing the inequality is strict if S is not a prefix is actually false. Consider the following set up. Suppose there are 3 types. Both 1 and 2, when they arrive, match with each other and nobody else. Type 3, when it arrives, chooses 1 first, and 2 next. Then I think actually $x_{23} = \lambda_3 Pr(N_2 > 0)$ even though $\{2\}$ is not a prefix. This is because 1 and 2 never coexist in the system.}

% \na{Maybe name these constraints?}

%We present a proof sketch capturing the main ideas behind this result, and postpone the full proof to \Cref{apx:balance-multiple}. 

%We present a proof sketch capturing the main ideas behind this result, and postpone the full proof to \Cref{apx:balance-multiple}. 

\begin{proof}[Proof sketch]
Most of the proof follows the same lines as in the proof of \Cref{prop:offline-bound}. For constraint~\eqref{eq:balance}, note that every arriving agent of type $i$ eventually matches or abandons the system. Because departures follow an exponential distribution and are not known in advance, the long-run rate of total departure from the system is the average number of available agents $\nss{i}^{\pol}$, multiplied by the departure rate $\mu_i$. The rate of matches is $\sum_{j \in \types} \big( \xrate{i}{j}^{\pol} + \xrate{j}{i}^{\pol} \big)$. Adding these together, therefore, gives the arrival rate of type-$i$ agents. For~\eqref{eq:matchrate}, an arriving type-$j$ agent can match into $S$ and contribute to the left side of the constraint only if some agent of a type in $S$ is available upon its arrival. This happens with probability $\PP{\textstyle\sum_{i \in S} \Npres{i}^{\pol} > 0}$, by PASTA (\Cref{cor:pasta}). Multiplying by the arrival rate of type-$j$ agents gives \eqref{eq:matchrate}. If $S$ is a prefix, then if some agent of a type in $S$ is available, $j$ \textit{must} match with such an agent. This gives the equality.
\end{proof}

\Cref{lemma:balance:multiple}---and \eqref{eq:matchrate} in particular---says that a key to analyzing the performance of a static priority policy is understanding the probability with which various sets of types are available. This can be quite challenging, as availability is correlated across types. However, our next result establishes a lower bound on the probability that a type from $S$ is available, expressed in terms of the {\em expected} number of available agents from this set; notably, the expected number of available type-$j$ agents is connected to the total departure rate of type $j$ (through a linear function in the case of exponential departures), which itself is connected to matching rates $\mathbf{x}^\pol$ via \eqref{eq:balance}.

To state this lower bound, for $S \subseteq \types$ write $\rho(S) = \sum_{i \in S}\arate{i}/\drate{i}$ and define
\begin{equation}
  \gam{S} \triangleq \frac{1 - e^{-\rho(S)}}{\rho(S)}
          = \gammaform{S}.
  \label{eq:gammadef}
\end{equation}
For $S=\emptyset$, we set $\gamma_\emptyset=1$, the limit of this expression as $\rho(S)\to0$.

\begin{restatable}{lemma}{gammabound}
\label{lem:gammabound}
Let $\pol = \{\porder{j}\}_{j \in \types}$ be any static priority policy and let $\mathbf{N}^{\pol}$ be drawn from its
stationary distribution. Then for every $S \subseteq \types$,
\begin{equation}
  \PP{\textstyle\sum_{i \in S} \Npres{i}^{\pol} > 0}
    \ge \gam{S} \sum_{i \in S} \nss{i}^{\pol}.
  \label{eq:gammabound}
\end{equation}
\end{restatable}

This is the key stochastic result of our work. There are several ways to interpret it. One interpretation is as a sort of converse to Markov's inequality, which states that the probability that a random variable on $\mathbb{N}$ is positive is at most its expectation. Inequality~\eqref{eq:gammabound} shows that this expectation also serves as a lower bound after multiplying by the constant $\gamma_S$. Another observation is that if the policy chooses never to match types in $S$, then \eqref{eq:gammabound} holds with equality. Thus, the conditional expectation of the number of available agents from $S$, given that this number is at least one, is maximized by the policy that makes no matches. There is also a conceptual connection to the literature on prophet inequalities, which often uses arguments of the flavor ``either item $i$ matches with high probability, or it is likely to be available to match with other items'' (see, for example, \citealp{lucier2017economic}). \Cref{lem:gammabound} gives an analogous conclusion for the stationary setting that we consider: if agents from $S$ are unlikely to be available (the left side of \eqref{eq:gammabound} is small), then it must be that they are matching at high rates (each $n_i^\pi$ is small). 

We defer the full proof of \Cref{lem:gammabound} to \Cref{apx:omittedproofs-priority}, and sketch main ideas next.

\begin{proof}[Proof sketch of \Cref{lem:gammabound}]
The proof relies solely on a property of static priority policies that we call \emph{independence of irrelevant alternatives} (IIA) (formally in \Cref{def:iia}): if an agent $a$ is available at time $t$, then removing $a$ from the history does not change what happens to any other agent, so the only effect on the state at time $t$ is that the number of available agents of $a$'s type decreases by one. Let $\mathbf{N}^\pol$ be the ``stationary state'' of the policy $\pi$ satisfying IIA. The main step is showing the following inequality, which holds for every type $i$ and every set $C$ of states:
\[
  \E{\Npres{i}^{\pol}\,\indic{\mathbf{N}^{\pol}-\mathbf{e}_i\in C}}
  \;\le\; \frac{\arate{i}}{\drate{i}}\,\PP{\mathbf{N}^{\pol}\in C},
\]
where $\mathbf{e}_i$ is the $i$-th unit vector. We first prove this inequality at a fixed time $t$, with $\mathbf{N}^{\pol}(t)$ in place of $\mathbf{N}^{\pol}$, and then let $t\to\infty$. Intuitively, the left side is the steady-state (or limit) value of the expected number of available type-$i$ agents whose removal from the history $\hist(t)$ would leave the system in a state in $C$. By IIA, removing one agent does not change other agents in distribution. Then, the Slivnyak--Mecke formula (\Cref{fact:poisson-arrival-calculus}) finds another way of counting this expected number by turning the expectation into an integral over the arrival and sojourn times $(s,\ell)$ of an extra type-$i$ agent added to the system, with respect to the rate function $\arate{i}\drate{i}e^{-\drate{i}\ell}$ of these times (\Cref{sec:poisson-background}). The integrand is the probability that this agent is available at time $t$ while the system without it is in a state in $C$. Since the extra agent can be available only if it is present, this integral is at most $\PP{\mathbf{N}^{\pol}(t)\in C}$ times the expected number of type-$i$ agents present at time $t$, which is at most $\arate{i}/\drate{i}$ by \Cref{fact:poisson-occupancy}.

Now fix a nonempty set $S\subseteq\types$, let $N_S=\sum_{i\in S}\Npres{i}^{\pol}$ and $p_m=\PP{N_S=m}$, and apply the inequality with $C$ the set of states in which exactly $m$ available agents have types in $S$. By summing over $i\in S$,
\[
  (m+1)\,p_{m+1}\;\le\;\rho(S)\,p_m \qquad\text{for all } m\ge 0, 
\]
simply because removing a type-$i$ agent leaves the system in a state in $C$ if and only if $N_S=m+1$. A Poisson random variable $L_S$ with mean $\rho(S)$ satisfies this recurrence with equality; $N_S$ has this distribution in steady state when agents with types in $S$ are never matched, because such agents are then available whenever they are present (\Cref{fact:poisson-occupancy}). Comparing the two recurrences shows that the ratio $p_m/\PP{L_S=m}$ is nonincreasing in $m$, so that the conditional distribution of $N_S$ given $N_S>0$ is stochastically dominated by that of $L_S$ given $L_S>0$. In particular,
\[
  \frac{\E{N_S}}{\PP{N_S>0}}=\E{N_S\mid N_S>0}\;\le\;\E{L_S\mid L_S>0}=\frac{\rho(S)}{1-e^{-\rho(S)}}=\frac{1}{\gam{S}}.
\]
Since $\E{N_S}=\sum_{i\in S}\nss{i}^{\pol}$, rearranging gives \eqref{eq:gammabound}.
\end{proof}

Together, \eqref{eq:balance}, \eqref{eq:matchrate}, and \Cref{lem:gammabound} are the only properties of static priority policies that we use. In the next sections, they let us construct families of linear programs whose solutions both identify a static priority policy and lower-bound its reward. In fact, from here we give two approaches for identifying a static priority policy that earns at least $\half \offopt$:

\begin{enumerate}[label=(\roman*)]
\item In \Cref{sec:primal}, we introduce a new family of linear programs indexed by a set $\m{M}\subseteq \TT$ of \emph{allowed} matches. Their constraints mirror \eqref{eq:balance} and \eqref{eq:matchrate} (restricted to $\m{M}$), with the probability in \eqref{eq:matchrate} replaced by its lower bound from \Cref{lem:gammabound}. Given a solution to an LP in this family, we draw inspiration from \Cref{lemma:balance:multiple}, designing a static priority policy such that tight constraints all correspond to prefixes of the priority list $\succ_j$. (Intuitively, if the constraint corresponding to $S$ and $j$ is tight, then the linear program is trying to make as many matches as possible between arriving type $j$ agents and waiting agents from the set $S$, so types in $S$ should appear at the ``top'' of $\succ_j$.) 

We show that any policy consistent with an optimal solution of an LP in this family earns a long-run average reward of at least the LP value. We then provide a polynomial-time algorithm to find a set $\m{M}$ such that a consistent
policy with respect to an optimal solution of the LP corresponding to $\m{M}$ exists, and this LP value is at least $\half \offopt$.

\item In \Cref{sec:values}, we explore an especially simple class of static priority policies, called \emph{value-based} policies, which assign a value $v_i$ to each type $i$, interpreted as the ``opportunity cost'' of matching that type. These policies  prioritize matches based on their ``score'' $r_{ij} - v_i - v_j$. Matches with negative scores are unacceptable, and matches with positive scores are acceptable. 

To find an appropriate choice of $\mathbf{v}$, we introduce a linear program that lower-bounds the reward of a given static priority policy $\pi$. Taking the dual of this program establishes conditions on $\mathbf{v}$ that ensure that the long-run average reward of $\pi$ is at least $\sum_{i \in \types} \lambda_i v_i$. Meanwhile, taking the dual of $\offrbrel$ establishes complementary conditions on $\mathbf{v}$ that ensure that $\sum_{i \in \types} \lambda_i v_i$ is at least $\half \offopt$. Finding a $\mathbf{v}$ that satisfies both sets of conditions simultaneously comes down to solving a particular fixed-point equation, which we show has a unique solution that is computable in polynomial time.
%Determine $\succ_j$ based on the ``value'' or the ``opportunity cost of matching'' of each type, as determined by the dual of $\plpm$. Let $v_i$ be the dual variable associated with the $i^{th}$ constraint in \eqref{eq:balance-lp-glb}. Construct $\succ_j$ by ordering agents $i$ according to their ``score'' $r_{ij} - v_i - v_j$, with the acceptable matches being those for which this score is positive.
\end{enumerate}
Although the two approaches above seem quite different, as we show later in \Cref{sec:connections}, there are deep connections between the two, and they are mathematically linked via linear programming duality. 

\section{Primal Approach}
\label{sec:primal}
In this section, we introduce a family of linear programs $\plpm$, indexed by a set $\m{M} \subseteq \m{T}\times \m{T}$ of allowed matches. We use these linear programs both to identify a static priority policy and to lower-bound its performance.

This is how this section is organized: \Cref{sec:lpAlg} introduces $\plpm$ and shows that the value of $\plpt$ is at least half of $\offrbrel$, and hence of the omniscient benchmark (\Cref{prop:factor2}). \Cref{sec:primal-alg} describes how to identify a static priority policy from an optimal solution of $\plpm$: \Cref{prop:plplowerbound} establishes that any static priority policy that is {\em consistent} with such a solution earns a long-run average reward of at least $\plpm$. Unfortunately, an optimal solution of $\plpt$ might not be {\em suitable}, meaning that no static priority policy is consistent with it. \Cref{sec:suitable} addresses this concern with a pruning algorithm that finds a set $\m{M}$ such that $\plpm$ has a suitable optimal solution and $\plpm\ge\plpt$. It also shows that $\plpm$ can be solved in polynomial time, despite its exponentially many constraints.

\subsection{A Novel Family of Linear Programs} \label{sec:lpAlg}
Our linear program is inspired by the objective \eqref{eq:Vpi} and by the constraints \eqref{eq:balance} and \eqref{eq:matchrate} of static priority policies (\Cref{lemma:balance:multiple}). The decision variables are $\bn = (n_i)_{i \in \m{T}}$ (with $n_i$ interpreted as the expected number of available type-$i$ agents) and $\bx = (x_{ij})_{(i,j) \in \m{M}}$ (with $x_{ij}$ interpreted as the rate of $(i,j)$ matches). 
%We also introduce slack variables $\bm{\psi} = \{\psi_{Sj}\}_{j \in \m{T}, S \subseteq \m{M}_j }$, resulting in the following linear program.
%\na{this might be easier to parse if you first say what are the variables -- i.e. move the sentence right after this program to before}
%\na{We should explicitly state the variables involved in this LP: $\{n_i\}_{i \in \m{T}}$, $\{x_{ij}\}_{i,j \in \m{T}}$ and $\{\psi_{Sj}\}_{j \in \m{T}, S \subseteq S(j,\m{M})}$. We could also change this to $\{x_{ij}\}_{j \in \m{T}, i \in S(j,\m{M})}$} \fs{This last one is not right, it should be $\{x_{ij}\}_{(i,j) \in \m{M}}$}
Given any set $\m{M}$ of allowed matches, define $\m{M}_j = \{i : (i,j) \in \m{M}\}$ to be the set of types that an arriving type-$j$ agent may be matched with, and define $\plpm$---referred to as the \emph{policy LP}---as follows.
\begin{lpbox}{Policy LP --- $\plpm$}
\begin{align}
    \plpm \triangleq \max\limits_{{\bn},{\bx} \ge 0} \quad &  \sum_{(i,j)\in \m{M}}r_{ij}x_{ij} \nonumber \\
\textrm{s.t.}\quad & 
    n_j \mu_j + \sum_{i \in \m{M}_j } x_{ij} + \sum_{i: j \in \m{M}_i } x_{ji} = \lambda_j & \quad \forall j \in \mathcal{T} \label{eq:balance-lp-glb}\tag{\textsc{P-Balance}}\\
    & \sum_{i \in S} x_{ij}  \le \lambda_j \gamma_{S}\sum_{i\in S} n_i & \quad \forall j \in \mathcal{T},   S \subseteq \m{M}_j \label{eq:availability-alg}\tag{\textsc{P-Availability}}
    %\\ x_{ij} &\ge 0 & \quad \forall (i,j) \in \mathcal{M} \label{eq:nonnegative-lp-glb}
\end{align}
\end{lpbox}
%\na{Note that I removed the $\Psi$, so that we only have $|\m{T}|^2 +|\m{T}|$ variables.}

%Note that \eqref{eq:balance-lp-glb} is analogous to \eqref{eq:balance}, and says that the arrival rate of a type $i$ must equal its match rate plus its abandonment rate $n_i \mu_i$. Meanwhile, \eqref{eq:availability-alg} is analogous to \eqref{eq:matchrate}. However, \eqref{eq:matchrate} includes the complex quantity $\PP{\sum_{i \in S} N_i^\pi > 0}$, which cannot be expressed as a linear function of ${\bn}$ and ${\bx}$. In \eqref{eq:availability-alg}, we have replaced this quantity with $\gamma_S \sum_{i \in S} n_i^\pi$. Recall that $\gamma_S$ is defined in \eqref{eq:gammadef}, and that \Cref{lem:gammabound} says that  for static priority policies,  $\PP{\sum_{i \in S} N_i^\pi > 0} \ge \gamma_S \sum_{i \in S} n_i^\pi$. Therefore, we can think of $\gamma_{S}\sum_{i\in S} n_i $ as a {\em pessimistic} approximation of the probability that an agent of some type $i \in S$ is available. This suggests the possibility of showing that $\plpm$ provides a {\em lower bound} on the performance of an appropriately chosen static priority policy. 

We now give some intuition for $\plpm$, comparing it first to the constraints \eqref{eq:balance} and \eqref{eq:matchrate} satisfied by static priority policies, and then to the constraints \eqref{eq:omn-balance} and \eqref{eq:omn-match} that define our upper bound $\offrbrel$.

Comparing \eqref{eq:balance-lp-glb} and \eqref{eq:availability-alg} with their counterparts \eqref{eq:balance} and \eqref{eq:matchrate}, the only difference is that the probability $\mathbb{P}(\sum_{i \in S} N_i^\pi > 0)$---that some type in $S$ is available---which appears in \eqref{eq:matchrate}, has been replaced by the linear term $\gamma_S \sum_{i \in S} n_i$ in \eqref{eq:availability-alg}. This replacement is motivated by \Cref{lem:gammabound}, which says that, for static priority policies, $\mathbb{P}(\sum_{i \in S} N_i^\pi > 0)$ is always at least $\gamma_S \sum_{i \in S} n_i^\pi$. Thus, \eqref{eq:availability-alg} is a {\em tightening} of the corresponding constraint satisfied by static priority policies. In \Cref{sec:primal-alg}, this will be the key fact used to show that the value of $\plpm$ is a lower bound on the reward of any static priority policy consistent with an optimal solution of $\plpm$.

For now, we turn to a different goal: relating the value of $\plpt$ to $\offrbrel$. There is a clear correspondence between \eqref{eq:omn-balance} and \eqref{eq:balance-lp-glb}: given any $\bx$ that satisfies \eqref{eq:omn-balance}, simply define $\bn$ so that \eqref{eq:balance-lp-glb} holds, that is, $n_j\mu_j=\alpha_j$. Thus, $n_j\mu_j$ is the difference between the arrival rate and the match rate of type $j$, that is, the rate at which type-$j$ agents abandon the system. The meaningful difference between $\offrbrel$ and $\plpt$ is the replacement of \eqref{eq:omn-match} with \eqref{eq:availability-alg}. Applying the definition of $\gamma_S$ in \eqref{eq:gammadef}, the right side of \eqref{eq:availability-alg} can be written as 
\[\lambda_j (1-e^{-\rho(S)}) \frac{\sum_{i \in S} n_i}{\rho(S)}.\]
This is simply the right side of \eqref{eq:omn-match}, times a fraction that is at most one. (To see that the fraction is at most one, note that \eqref{eq:balance-lp-glb} ensures that each $n_i \le \lambda_i/\mu_i$; conceptually this makes sense, as $n_i$ represents the expected number of {\em available} agents, while $\lambda_i/\mu_i$ is the expected number of {\em present} agents.) As a result, for any feasible $\bn$, \eqref{eq:availability-alg} is a tighter constraint than \eqref{eq:omn-match}, and thus $\plpm \le \offrbrel$.

The constraint \eqref{eq:availability-alg} also captures the following intuition. Matched agents are no longer available, so the more often a policy matches type-$i$ agents, the fewer type-$i$ agents are available, on average, when other agents arrive. The right side of \eqref{eq:omn-match} ignores this effect, since it does not depend on how often agents are matched. In \eqref{eq:availability-alg} with $S=\{i\}$, by contrast, the fraction that multiplies the right side of \eqref{eq:omn-match} is $n_i\mu_i/\lambda_i$, which by \eqref{eq:balance-lp-glb} is the fraction of type-$i$ agents that abandon rather than match. Thus, the higher the match rate of type-$i$ agents, the smaller $n_i$ becomes, resulting in a tighter availability constraint \eqref{eq:availability-alg}.
% \rn{the previous version was  difficult to read. What do you think about the new version?}
%Rad: this was  difficult to read. What do you think about the new version?
% The form of \eqref{eq:omn-balance} also captures the following intuition. For an online policy $\pi$ and a type $i\in \m{T}$, it should not be possible for \eqref{eq:omn-balance} and \eqref{eq:omn-match} with $S = \{i\}$ to both be tight, or even nearly tight. If \eqref{eq:omn-balance} is (nearly) tight for type $i$, this says that (almost) all agents of type $i$ are matching. If \eqref{eq:omn-balance} is (nearly) tight for $S = \{i\}$, this says the probability that type $i$ is available is (almost) as large as the probability that it is present, which should only be possible if type $i$ rarely matches. The constraint \eqref{eq:availability-alg} captures exactly this intuition: the higher the match rate for agents of type $i$, the smaller $n_i$ becomes, resulting in a tighter availability constraint \eqref{eq:availability-alg}. 

Although $\plpt \le \offrbrel$, it cannot be too much smaller, as the following result shows.

\begin{restatable}{proposition}{factortwo}
$ \plpt \geq \frac{1}{2}\offrbrel$. \label{prop:factor2}
\end{restatable}
%The proof is direct: it shows that if $\bar{\bx}$ is any feasible solution of $\offrbrel$, then taking $\hat{\bx} =  {\bx}/2$ and $\hat{\bn}$ defined by \eqref{eq:balance-lp-glb} provides a solution to $\plpt$, with an objective that is half as large.
The proof is simple: it shows that for any feasible solution $\bar\bx$ to $\offrbrel$, $\bar\bx/2$ is feasible for $\plpt$ (with $\bn$ defined so that \eqref{eq:balance-lp-glb} holds).
\begin{proof}
Let $\bar\bx$ be feasible for $\offrbrel$. Define $\bx = \bar\bx/2$, and define $\bn$ to make \eqref{eq:balance-lp-glb} hold. Because $\bar \bx$ satisfies \eqref{eq:omn-balance}, we must have $n_i \ge \lambda_i/(2\mu_i)$. Apply \eqref{eq:gammadef} to see that the right side of \eqref{eq:availability-alg} is at least $\lambda_j (1-e^{-\rho(S)})/2$. Because $\bar \bx$ satisfies \eqref{eq:omn-match}, the left side of \eqref{eq:availability-alg} is at most $\lambda_j (1-e^{-\rho(S)})/2$, so \eqref{eq:availability-alg} holds. Thus, $(\bn, \bx)$ is feasible for $\plpt$, with objective value $\frac{1}{2} \sum r_{ij} \bar{x}_{ij}$. 
\end{proof}

Therefore, if we can find a static priority policy whose long-run average reward is at least $\plpt$, this policy also earns at least half of the omniscient benchmark. Identifying such a policy is the goal of \Cref{sec:primal-alg,sec:suitable}.

%\na{Our next task is to identify a static priority policy that outperforms $\plpm$. (Maybe could reference \Cref{lem:gammabound} to motivate why this ``should'' be true.)}

\subsection{Identifying a Policy and Bounding Its Performance} \label{sec:primal-alg}

In this subsection, we define what it means for a static priority policy $\pi$ to be {\em consistent} with a feasible solution $(\bn, \bx)$ of $\plpm$. The first condition is that the policy should only make matches in $\m{M}$, and should not make $(i,j)$ matches if $x_{ij} = 0$. The second condition comes from examining the constraints \eqref{eq:availability-alg}. If \eqref{eq:availability-alg} is tight for a particular $S$ and $j$, then intuitively, the linear program matches arriving type-$j$ agents with types in $S$ as frequently as possible, so $S$ should be a prefix of the priority list $\succ_j$. This condition is also inspired by \Cref{lemma:balance:multiple}, which says that for any static priority policy and any prefix $S$ of $\succ_j$, the constraint \eqref{eq:matchrate} --- the analog of \eqref{eq:availability-alg} --- is tight.

To formalize our definition, for any feasible solution $({\bn}, {\bx})$ to $\plpm$, we define
    \begin{equation}
        \m{P}_j(\bn, \bx) =  \{\emptyset\ne{} S \subseteq \m{M}_j\ : \eqref{eq:availability-alg} \text{ holds with equality}\} \label{eq:Pj}.
    \end{equation}

\begin{definition}[Primal consistency] \label{def:associated-greedy}
A static priority policy $\pol = \{\porder{j}\}_{j \in \types}$ is \textbf{consistent} with a feasible solution $({\bn}, {\bx})$ of $\plpm$ if $\pi$ only makes matches in $\mathcal{M} \cap \{(i,j) : x_{ij} > 0\}$, and $\m{P}_j(\bn,\bx) \subseteq \prefixes{\succ_j}$ for all $j \in \m{T}$.
\end{definition}
%\na{As phrased, if all $x_{ij}$ are positive but no constraints are tight, then any $\pi$ is consistent (i.e. we do not require that $\pi$ makes all matches in $\mathcal{M} \cap \{x_{ij} > 0\}$). Do we want to change this? If we do, the statement about suitability would change.}

There may not be a unique $\pi$ consistent with $(\bn, \bx)$. For example, suppose that for type $j$, we have $x_{1j} = x_{4j} = 0$, and there are two tight constraints in \eqref{eq:availability-alg}, corresponding to the sets $\{2,3\}$ and $\{2,3,5\}$. Then a consistent static priority policy could choose either $2 \succ_j 3 \succ_j 5 \succ_j \emptyset$ or $3 \succ_j 2 \succ_j 5 \succ_j \emptyset$. A separate concern is that there might be {\em no} static priority policy consistent with $(\bn,\bx)$. We address this concern in \Cref{sec:suitable}; for now, we proceed assuming such a policy exists. With this assumption, we prove \Cref{prop:plplowerbound} in \Cref{apx:omittedproofs-primal} and explain the main idea here.

\begin{restatable}{proposition}{plplowerbound}
%\begin{proposition}
    \label{prop:plplowerbound}
    If $({\bn^*}, {\bx^*})$ is an optimal solution of $\plpm$ and $\pi$ is a static priority policy consistent with $({\bn^*}, {\bx^*})$, then $V(\pi) \ge \plpm$.
%\end{proposition}
\end{restatable}
\begin{proof}[Proof sketch.]The key fact used in the proof is that, by \Cref{lem:gammabound}, \eqref{eq:availability-alg} is a tightening of the constraint \eqref{eq:matchrate} satisfied by static priority policies. As a result, the steady-state quantities $(\bn^\pi,\bx^\pi)$ of $\pi$, defined in \eqref{eq:nxpi}, are typically not feasible for $\plpm$. However, this alone does not imply that $\bx^\pi$ yields a higher objective value than $\bx^*$. %Furthermore, it may not be true that $x_{ij}^\pi \ge x_{ij}^*$ for each $(i,j)$. 
Our proof shows that the ``arrow'' from $\bx^\pi$ to $\bx^*$ ``pierces'' the feasible region (see \Cref{fig:piercing}). This implies that $(\bn^*,\bx^*)$ can be expressed as a convex combination of $(\bn^\pi,\bx^\pi)$ and another feasible point $(\hat{\bn},\hat{\bx})$. If we let $\mathbf{r}\circ\bx=\sum_{(i,j)\in\m{M}}r_{ij}x_{ij}$ denote the objective value of $\bx$, then we have $\mathbf{r} \circ \hat{\bx} \le \mathbf{r} \circ \bx^*$ by optimality of $(\bn^*, \bx^*)$, so it follows that $\val{\pi}=\mathbf{r} \circ \bx^\pi \ge \mathbf{r} \circ \bx^*=\plpm$.
\end{proof}

\sidecaptionvpos{figure}{c}% [Version 3: center the side caption vertically next to the larger drawing]
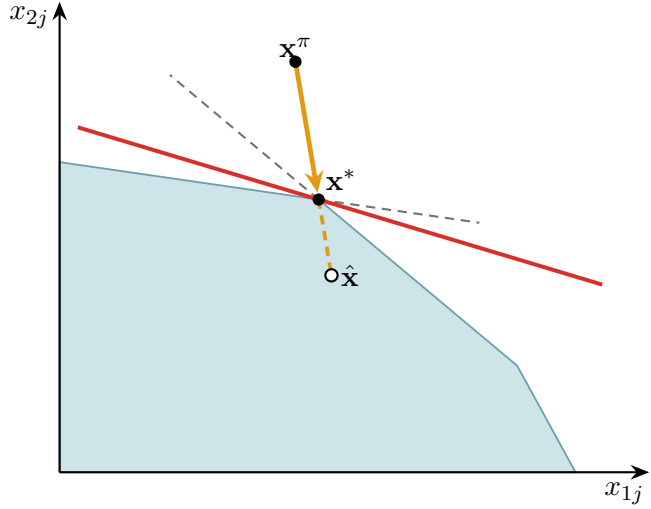
\begin{SCfigure}
%=====================================================================
%  figures/piercinglemma-tikz.tex
%  TikZ drawing for Figure \ref{fig:piercing} (proof of Proposition 4.2).
%
%  Usage (inside the SCfigure environment of tex/primal-approach.tex):
%      \input{figures/piercinglemma-tikz}
%  Requires: tikz with the libraries calc and arrows.meta (both loaded
%            in tex/macro.tex).
%  Size:     change x= and y= below to rescale the drawing.
%=====================================================================
\begin{tikzpicture}[
    x=0.75cm, y=0.75cm,
    region/.style    = {fill=pierceRegion, draw=pierceEdge, line width=0.7pt, line join=round},
    axis/.style      = {line width=0.8pt, -{Stealth[length=7pt, width=5.5pt]}},
    objective/.style = {line width=1.5pt, draw=pierceRed},
    facet/.style     = {line width=0.8pt, draw=black!55, dash pattern=on 3.5pt off 2.5pt},
    pierce/.style    = {line width=1.8pt, draw=pierceGold, -{Stealth[length=8pt, width=7pt]}, shorten >=2.6pt},
    extend/.style    = {line width=1.4pt, draw=pierceGold, dash pattern=on 3.5pt off 2.5pt},
    dot/.style       = {circle, fill=black, inner sep=0pt, minimum size=4.6pt},
    hollow/.style    = {circle, draw=black, fill=white, line width=0.8pt, inner sep=0pt, minimum size=4.6pt},
    lbl/.style       = {font=\normalsize, inner sep=1.5pt},
  ]
  % ---- Colours (local to the picture) --------------------------------
  \definecolor{pierceRegion}{RGB}{205,229,233}   % feasible region
  \definecolor{pierceEdge}{RGB}{110,160,172}     % its boundary
  \definecolor{pierceRed}{RGB}{214,48,39}        % iso-objective line
  \definecolor{pierceGold}{RGB}{226,154,20}      % arrow from x^pi to x^*

  % ---- Points --------------------------------------------------------
  % O, B, xs, C, D: vertices of the projected feasible region, where xs is
  % the vertex x^*; xpi: the point x^pi (outside the region); xh: the
  % feasible point beyond x^* on the line through x^pi and x^*.
  \coordinate (O)   at (0,0);
  \coordinate (B)   at (0,5.47);
  \coordinate (xs)  at (4.57,4.81);
  \coordinate (C)   at (8.07,1.88);
  \coordinate (D)   at (9.10,0);
  \coordinate (xpi) at (4.16,7.24);
  \coordinate (xh)  at ($(xs)!-0.55!(xpi)$);

  % ---- Feasible region -----------------------------------------------
  \draw[region] (O) -- (B) -- (xs) -- (C) -- (D) -- cycle;

  % ---- Axes ----------------------------------------------------------
  \draw[axis] (O) -- (10.4,0) node[lbl, anchor=north east, yshift=-2pt] {$x_{1j}$};
  \draw[axis] (O) -- (0,8.3) node[lbl, anchor=north east, xshift=-2pt] {$x_{2j}$};

  % ---- Dashed lines: the two constraints that are tight at x^*, ------
  % extended beyond x^*.
  \draw[facet] (xs) -- ($(xs)!-0.75!(C)$);
  \draw[facet] (xs) -- ($(xs)!-0.62!(B)$);

  % ---- Red line: points with the same objective value as x^* ---------
  \draw[objective] ($(xs)+(-4.25,1.275)$) -- ($(xs)+(5.0,-1.5)$);

  % ---- Arrow from x^pi to x^*, and its continuation into the region ---
  \draw[extend] (xs) -- (xh);
  \draw[pierce] (xpi) -- (xs);

  % ---- Points and labels ---------------------------------------------
  \node[dot] at (xpi) {};
  \node[dot] at (xs) {};
  \node[hollow] at (xh) {};
  \node[lbl, anchor=south] at (xpi) {$\mathbf{x}^{\pi}$};
  \node[lbl, anchor=south west, xshift=1pt, yshift=2pt] at (xs) {$\mathbf{x}^{*}$};
  \node[lbl, anchor=west, xshift=2pt] at (xh) {$\hat{\mathbf{x}}$};
\end{tikzpicture}%
\caption{Visualization of the proof of \Cref{prop:plplowerbound}. The shaded region represents the projection of the polytope of feasible solutions to $\plpm$ onto $\bx$-space, and the red line contains the points with the same objective value as the optimal vertex $\bx^*$. The policy's match rates $\bx^\pi$ lie beyond both constraints that are tight at $\bx^*$ (dashed), so the ``arrow'' from $\bx^\pi$ to $\bx^*$ ``pierces'' the feasible region: continued past $\bx^*$, it reaches a feasible point $\hat{\bx}$. Since $\hat{\bx}$ cannot lie above the red line, $\bx^\pi$ cannot lie below it. %\na{Technically, the polytope is in $(\bn,\bx)$ space.}
} \label{fig:piercing}
\end{SCfigure}

\subsection{Finding a Suitable Solution} \label{sec:suitable}
\Cref{sec:lpAlg,sec:primal-alg} establish two main results. First, $\plpt \ge \frac{1}{2}\offrbrel$. Second, if we take an optimal solution of $\plpm$ and let $\pi$ be any static priority policy consistent with this solution, then $V(\pi) \ge \plpm$. Combining these, any static priority policy consistent with an optimal solution $(\bn^*,\bx^*)$ of $\plpt$ has long-run average reward at least $\frac{1}{2}\offrbrel \ge \half\offopt$. 

Two barriers remain. First, and most importantly, we have not guaranteed that there exists a static priority policy consistent with $(\bn^*,\bx^*)$. In fact, such a policy might not exist.\footnote{For example, let there be two types, with $r_{11} = r_{22} = 1$ and $r_{12} = r_{21} = 5$, with $\lambda_1 = \mu_1 = \mu_2 = 1$ and $\lambda_2 = 2$. Then when solving $\plpt$, the constraint for $j = 1$, $S = \{1, 2\}$ is tight, but $x_{11} = 0$. It is impossible to simultaneously have $\{1,2\}$ be a prefix of $\succ_1$ and to have $(1,1)$ matches be unacceptable.} Second, if we care not only about the existence of a good policy but also about our ability to find it, then one might be concerned about solving the linear program, given that \eqref{eq:availability-alg} consists of exponentially many constraints.

In this subsection, we show how to find a set of matches $\m{M}$ and an optimal solution $({\bn^*}, {\bx^*})$ to $\plpm$ such that the following hold:
\begin{enumerate}[(i)]
\item there exists a static priority policy consistent with $({\bn^*}, {\bx^*})$, and
\item $\plpm \ge \plpt$. 
\end{enumerate}
From \Cref{prop:factor2,prop:plplowerbound}, it follows that such a policy earns at least half of the omniscient benchmark. Also, we show that such a policy can be found in time polynomial in the input size.\footnote{\label{fn:arithmetic}For simplicity, we assume that $\boldsymbol{\lambda}$, $\boldsymbol{\mu}$, and $\mathbf{r}$ are rational, and we call $|\types|$ plus their encoding length the \emph{input size}. We also treat the constants $\gamma_S$, which involve exponentials, as rational numbers of polynomial encoding length that can be computed exactly. We do not dwell on the details of the arithmetic and the  bookkeeping of encoding lengths---both could be handled using standard methods. The polynomial-time claims of \Cref{sec:values,sec:connections} rely on the same convention.}

We start by giving a name to solutions with which some static priority policy is consistent.
\begin{definition}[Suitable solution]
A feasible solution $({\bn}, {\bx})$ of $\plpm$ is {\bf suitable} if there is a static priority policy $\pol$ that is consistent with it.
\end{definition}

Examining \Cref{def:associated-greedy}, there are two things that might prevent a consistent policy from existing. One is that it might be impossible to make each $S \in \m{P}_j(\bn,\bx)$ a prefix of $\succ_j$. For example, if for a given $j$, \eqref{eq:availability-alg} is tight for both $\{1\}$ and $\{2\}$, then no consistent policy exists: these sets cannot both be prefixes of $\succ_j$ ($1$ and $2$ cannot both be the highest-priority partner for $j$). Similar issues arise if the tight sets are, for example, $\{1, 3\}$ and $\{1, 2, 4\}$. Our next lemma states that we do not need to worry about this concern: given any feasible solution $(\bn,\bx)$ of $\plpm$, the sets in $\m{P}_j(\bn,\bx)$ are ``nested'', meaning that for any two sets in this family, one contains the other.

A second concern that might prevent a consistent policy from existing is that there might be a $j \in \types$ and $S \in \m{P}_j(\bn,\bx)$, that is, a set whose constraint \eqref{eq:availability-alg} is tight, and yet $x_{ij} = 0$ for some $i \in S$. This causes problems, as one condition of consistency is that $S$ should be a prefix of $\succ_j$, and a second condition is that $i$ should be unacceptable to $j$. \Cref{lemma:suitablestructure} establishes that this possibility is the only barrier that might prevent $(\bn,\bx)$ from being suitable.

\begin{restatable}{lemma}{suitablestructure}
    \label{lemma:suitablestructure}
Let $({\bn}, {\bx})$ be any feasible solution to $\plpm$. Then for each $j \in \m{T}$, the set function 
\[\Psi_j(S) \triangleq \frac{1}{\gamma_S}\left(\lambda_j \gamma_S \sum_{i \in S} n_i - \sum_{i \in S} x_{ij}\right)\]
is submodular. Moreover, if $S,S' \in \m{P}_j(\bn, \bx)$, then either $S \subseteq S'$ or $S' \subseteq S$. Consequently, a feasible solution $({\bn}, {\bx})$ of $\plpm$ is suitable if and only if 
\begin{equation} x_{ij} > 0 \quad \text{for all } j\in\m{T},\ S \in \m{P}_j(\bn, \bx), \text{ and } i \in S. \label{eq:suitability} \end{equation}  
\end{restatable}

%This possibility is exactly what is ruled out by \eqref{eq:suitability}.

%there might be $i,j \in \types$ and $S \subseteq \m{M}_j$ such that $x_{ij} = 0$ (so the policy should treat $i$ as unacceptable to $j$), and yet $i \in S$ and $S \in \m{P}_j$ (so the policy should make $S$ a prefix of $\succ_j$, implying that $i$ is acceptable). 

\begin{algorithm}[t]
	\caption{Suitability finder (pruning)} 
	\begin{algorithmic}[1]
        \State {\bf Input:} instance $\m{I} = (\m{T}, \boldsymbol{\lambda}, \boldsymbol{\mu}, \mathbf{r})$.
        \State $\m{M}^{(0)} \leftarrow \m{T} \times \m{T}$.
        \State $R_0 \leftarrow 0$.
        \State $k \leftarrow 0$.
        \State   $exists\_zero \leftarrow \ True$.
        \While {$exists\_zero$}
        \State $k \leftarrow k+1$
        \State $R_k, (\bn^k, \bx^k) \leftarrow$ solve $\plp(\m{M}^{(k-1)})$ \Comment{optimal value and an optimal solution}
        \State   $exists\_zero \leftarrow$ [$x^k_{ij} = 0$ for some $(i,j)\in\m{M}^{(k-1)}$].
        \State $\m{M}^{(k)} \leftarrow \m{M}^{(k-1)} \setminus \{(i,j) : x^k_{ij} = 0\}$.
        \EndWhile
        \State {\bf Return:} $\m{M}^{(k)}$, $R_k$, and $(\bn^k, \bx^k)$.
	\end{algorithmic} 
    \label{algo:suitability}
\end{algorithm}

%\begin{restatable}{lemma}{suitability}
% \label{lemma:suitability}
%A feasible solution $({\bn}, {\bx})$ of $\plpm$ is suitable if and only if 
%\begin{equation} S \in \m{P}_j(\bn, \bx) \implies x_{ij} > 0 \text{ for } i \in S. %\label{eq:suitability} \end{equation}   
%\end{restatable}

%\na{Maybe we could combine \Cref{lemma:suitability} and \Cref{lemma:suitablestructure} into one lemma Once you have \Cref{lemma:suitablestructure}, the proof of \Cref{lemma:suitability} is pretty much automatic.}

\Cref{lemma:suitablestructure}, which we prove in \Cref{apx:omittedproofs-primal}, establishes that \eqref{eq:suitability} is necessary and sufficient for a consistent policy to exist. We now show how to find a choice of $\m{M}$ such that an optimal solution of $\plpm$ is suitable.

Recall that $\m{P}_j(\bn,\bx)$, defined in \eqref{eq:Pj}, consists of the sets for which \eqref{eq:availability-alg} is tight. Thus, \eqref{eq:suitability} says that if $x_{ij}$ is involved in a tight constraint, then it should be positive. We can ensure that this condition is satisfied using the following algorithm. First solve $\plpt$. Any $(i,j)$ such that $x_{ij} = 0$ could potentially cause  unsuitability, so remove these pairs from $\m{M}$ and re-solve $\plpm$. Repeat until $x_{ij} > 0$ for all $(i,j) \in \m{M}$, ensuring suitability of the solution by \Cref{lemma:suitablestructure}. Because the solution from each iteration is feasible for the next iteration, the optimal value never decreases throughout the procedure. We call this procedure \emph{pruning} and formalize it in \Cref{algo:suitability} and \Cref{lemma:algfindssuitable}.

Each iteration of this procedure solves a linear program of the form $\plpm$, which brings us back to the second barrier: $\plpm$ has exponentially many availability constraints. These constraints, however, admit an efficient separation oracle. The key is the submodularity established in \Cref{lemma:suitablestructure}. For a nonnegative candidate solution $(\hat{\bn},\hat{\bx})$, the availability constraints \eqref{eq:availability-alg} for arrivals of type $j$ hold if and only if $\Psi_j(S)\ge0$ for every nonempty subset $S\subseteq\m{M}_j$. Submodular minimization can therefore either certify these constraints or identify a violated one~\citep{schrijver2000combinatorial}. Together with directly checking the other constraints, this gives an exact polynomial-time separation oracle for $\plpm$. We can then use the standard equivalence of separation and optimization~\citep{GrotschelLovaszSchrijver1988}, which yields a polynomial-time algorithm based on the ellipsoid method for computing an exact optimal solution of $\plpm$. We do not dwell on the details of the arithmetic (see footnote~\ref{fn:arithmetic}). Formally, we have the following proposition.

\begin{restatable}{proposition}{polytime}
\label{prop:polytime}
For every $\m{M}\subseteq\TT$, an optimal solution of $\plpm$ can be computed in time polynomial in the input size.

% \rn{do we feel comfortable with this? rational vs. irrational is not really the topic of this paper...}
% \footnote{The constants $\gamma_S$ are not necessarily rational numbers, so the linear program does not necessarily have a short exact representation. However, we can approximate these constants to any given precision, and show that this approximation does not cause a significant multiplicative loss in the value of the linear program or in the guarantee of the resulting policy. We omit the details.}
\end{restatable}

The proof is deferred to \Cref{apx:omittedproofs-primal}.

\begin{proposition}
    \label{lemma:algfindssuitable}
    \Cref{algo:suitability} returns a set $\m{M}$ such that $\plpm \geq \plpt$, as well as a suitable optimal solution of $\plpm$. Furthermore,  \Cref{algo:suitability} runs in time polynomial in the input size.% [
% \todo{Say what ``polynomial time'' means.}
\end{proposition}

\begin{proof}
Every iteration except the last removes at least one of the $|\types|^2$ pairs, so the process stops after at most $|\types|^2+1$ iterations, and each iteration solves one linear program, which takes polynomial time by \Cref{prop:polytime}.
% \todo{Explain why each iteration runs in poly time.}

When this process stops, by definition $x^k_{ij} > 0$ for each $(i,j) \in \m{M}^{(k)}$, so by \Cref{lemma:suitablestructure}, the solution must be suitable. To see that the $R_k$ form a nondecreasing sequence, notice that the solution $(\bn^k, \bx^k)$ of $\plp(\m{M}^{(k-1)})$ is also feasible for $\plp(\m{M}^{(k)})$, since $\m{M}^{(k)}_j\subseteq\m{M}^{(k-1)}_j$ for every $j$ and therefore \eqref{eq:availability-alg} encodes fewer constraints. (Technically, $\plp(\m{M}^{(k)})$ involves fewer decision variables; we simply ignore the values $x_{ij}^k$ for which $(i,j) \notin \m{M}^{(k)}$.) Furthermore, this feasible solution of $\plp(\m{M}^{(k)})$ still has value $R_k$, since by definition $x_{ij}^k = 0$ for each $(i,j) \in \m{M}^{(k-1)} \setminus \m{M}^{(k)}$. Hence $R_{k+1}\ge R_k$, and the value of the returned program is at least $R_1=\plpt$.
%moving from problem $k$ to problem $k+1$ is equivalent to adding constraints that certain $x_{ij} = 0$ (which by definition does not change the objective value, since this is satisfied by our optimal solution), and then removing some of the constraints in \eqref{eq:availability-alg} (since $\m{M}_j$ has shrunk). The removal of these constraints weakly increases the objective. 
\end{proof}

Let $(\bn,\bx)$ be the solution returned by \Cref{algo:suitability}. A static priority policy consistent with $(\bn,\bx)$ can also be computed in polynomial time. Fix $j\in\m{T}$, and note that $\Psi_j(\emptyset)=0$ and, by feasibility, $\Psi_j(S)\ge0$ for every nonempty $S\subseteq\m{M}_j$. Thus, a type $i\in\m{M}_j$ belongs to some set in $\m{P}_j(\bn,\bx)$ if and only if the minimum of $\Psi_j$ over the sets $S\subseteq\m{M}_j$ that contain $i$ is zero, and some set in $\m{P}_j(\bn,\bx)$ contains $i$ but not $i'$ if and only if the minimum of $\Psi_j$ over the sets $S\subseteq\m{M}_j$ that contain $i$ but not $i'$ is zero. Each of these minimizations is a submodular minimization, which takes polynomial time~\citep{schrijver2000combinatorial}. Since the sets in $\m{P}_j(\bn,\bx)$ are nested, we can then rank as acceptable exactly the types that belong to some set in $\m{P}_j(\bn,\bx)$, with $i\succ_j i'$ whenever some set in $\m{P}_j(\bn,\bx)$ contains $i$ but not $i'$ and with the remaining ties broken arbitrarily. As in the proof of \Cref{lemma:suitablestructure}, since $(\bn,\bx)$ satisfies \eqref{eq:suitability}, the resulting static priority policy is consistent with $(\bn,\bx)$.

\section{Dual Approach}
\label{sec:values}
The primal approach of \Cref{sec:primal} identifies a static priority policy by solving a linear program and reading the priorities off its tight constraints. We now describe a second, dual-based way to identify such policies and certify their reward. This section focuses on the specific class of value-based static priority policies that were introduced earlier. As a reminder, such a policy gives each type a \emph{value} that captures the opportunity cost of matching an agent of that type, and prioritizes matches by their \emph{score}: the reward minus the values of the two matched types. Since the values will serve as dual variables (as we formalize shortly), the scores can be interpreted as dual-adjusted rewards.

As in \Cref{sec:primal}, it is convenient to fix a set of matches $\m{M}\subseteq\TT$ outside of which the policy never matches. Formally, we define such restricted value-based policies, which are described by one number $v_i$ per type $i$ together with the set $\m{M}$, as follows.

\begin{definition}[Dual consistency]
    \label{defn:dual-consistent-policy}
    Given  a vector of values $\mathbf{v} \in \mathbb{R}^{\m{T}}$ and a set of matches $\m{M}\subseteq\TT$, we say a static priority policy $\pol = \{\porder{j}\}_{j \in \types}$ is \textbf{consistent} with $\m{M}$ and $\mathbf{v}$ if it satisfies the following:
\begin{itemize}
    \item The set of acceptable matches $\m{M}(\pi)=\{(i,j): i \porder{j} \emptyset \}$ is a subset of $\m{M}\cap{}\{(i, j) : r_{ij} - v_i - v_j \geq 0\} $.
    \item If $(i,j)\in\m{M}$ and $r_{ij} - v_i - v_j > 0$, then $i \porder{j} \emptyset$.
    \item If $(i,j),(i',j)\in\m{M}$ and $r_{ij} -v_i - v_j > r_{i'j}-v_{i'}-v_j\geq 0$, then $i \porder{j} i'$.
    \end{itemize}
When $\m{M}=\TT$ we simply say that $\pi$ is consistent with $\mathbf{v}$.
\end{definition}
In words, a value-based policy consistent with $\m{M}$ and $\mathbf{v}$ makes only matches in $\m{M}$ with nonnegative score $r_{ij}-v_i-v_j$, makes every match in $\m{M}$ with positive score acceptable, and ranks the acceptable partners of each arriving type in decreasing order of score (i.e., of dual-adjusted reward). Such a policy always exists, although it might not be unique: ties between equal scores may be broken arbitrarily, and a match with score exactly zero may be either acceptable or unacceptable.

In fact, the unrestricted case $\m{M}=\TT$ suffices for the guarantee of this section. Therefore, our main analysis in this section sets $\m{M}=\TT$, and shows how to find values $\mathbf{v}$ that simultaneously (i) certify a lower bound on the long-run average reward of every policy consistent with them, and (ii) make this lower bound at least half of the omniscient benchmark. A general set $\m{M}$ is needed only in \Cref{sec:restricted-values} and later in \Cref{sec:connections}, where we connect the policies of this section with those of \Cref{sec:primal}, which make matches only within the set $\m{M}$ returned by \Cref{algo:suitability}.

% \begin{definition}
%     \label{defn:dual-consistent-policy}
%     Given a vector of values $v \in \mathbb{R}^{\m{T}}$, we say a static priority policy $\pol = \{\porder{j}\}_{j \in \types}$ is \textbf{consistent} with $v$ if it satisfies all of the following:
% \begin{itemize}
%     \item The set of acceptable matches $\{i: i \porder{j} \emptyset \}$ is a subset of $\{(i, j) : r_{ij} - v_i - v_j \geq 0\} $.
%     \item If  $r_{ij} - v_i - v_j > 0$, then $i \porder{j} \emptyset$.
%     \item If  $r_{ij} -v_i - v_j > r_{i'j}-v_{i'}-v_j\geq 0$, then $i \porder{j} i'$.
%     \end{itemize}
% \end{definition}

% \begin{definition}
%     \label{defn:dual-consistent-policy}
%     Given a vector of values $v \in \mathbb{R}^{\m{T}}$, we say a static priority policy $\pol = \{\porder{j}\}_{j \in \types}$ is \textbf{consistent} with $v$ if it satisfies all of the following:
% \begin{itemize}
%     \item The set of acceptable matches $\{i: i \porder{j} \emptyset \}$ is a subset of $\{(i, j) : r_{ij} - v_i - v_j \geq 0\} $.
%     \item If  $r_{ij} - v_i - v_j > 0$, then $i \porder{j} \emptyset$.
%     \item If  $r_{ij} -v_i - v_j > r_{i'j}-v_{i'}-v_j\geq 0$, then $i \porder{j} i'$.
%     \end{itemize}
% \end{definition}

\subsection{Value-Based Certificates} 
\label{sec:value-certificate-1}
For the first goal, we begin with a linear program, inspired by \Cref{lem:gammabound}, that provides a lower bound on the long-run average  reward of an \emph{arbitrary} static priority policy $\pi$; in particular, $\pi$ need not be value-based. As before, $\mathcal M(\pi)=\{(i,j):i\succ_j\emptyset\}$ denotes its set of acceptable matches, and $\mathcal P_j(\pi)$ denotes the collection of prefixes of $\succ_j$ (\Cref{def:prefix}). We then introduce the following LP:

\begin{lpbox}{Lower-bound LP --- $\plblp(\boldsymbol{\pi})$}
\begin{align}
  \plblp(\pi)
  \;\triangleq\; \min_{\bn,\bx \ge 0}\quad
   & \sum_{(i,j)\in\mathcal M(\pi)} r_{ij}x_{ij}
  \nonumber\\[2pt]
  \textrm{s.t.}\quad
   & n_j\mu_j
     + \sum_{i:(i,j)\in\mathcal M(\pi)} x_{ij}
     + \sum_{i:(j,i)\in\mathcal M(\pi)} x_{ji}
     \;\ge\; \lambda_j
   && \forall j\in\mathcal T, \label{eq:balance-lb}\tag{\textsc{$\pi$-Balance}}\\
   & \sum_{i\in S} x_{ij}
     \;\ge\; \lambda_j\gamma_S \sum_{i\in S} n_i
   && \forall j\in\mathcal T,\ S\in\mathcal P_j(\pi).
   \label{eq:availability-lb}\tag{\textsc{$\pi$-Availability}}
\end{align}
\end{lpbox}

% \begin{equation}\label{eq:replacement-primal}
%  \textup{LP}^{\textup{LB}}(\mathcal I,\pi)
%  =\min_{n,x\ge0}\quad \sum_{(i,j)\in\mathcal M(\pi)}r_{ij}x_{ij}
% \end{equation}
% \begin{align}
%  \textup{subject to}\quad
%  n_i\mu_i+\sum_{j:(i,j)\in\m{M}(\pi)}x_{ij} +\sum_{j:(j,i)\in\m{M}(\pi)}x_{ji}\ge\lambda_i 
%        &&\forall i\in\mathcal T, \label{eq:balance-lb}\\
%  \sum_{i\in S}x_{ij}\ge\lambda_j\gamma_S\sum_{i\in S}n_i
%        &&\forall j\in\mathcal T,\ S\in\mathcal P_j(\pi). \label{eq:availability-lb}
% \end{align}

This LP is similar in spirit to $\plpm$, with two important differences. First, its objective and constraints depend on $\pi$, so \eqref{eq:availability-lb} consists of at most $|\types|^2$ constraints, rather than one constraint for each arriving type $j$ and nonempty set $S\subseteq\m{M}_j$. Second, it minimizes rather than maximizes its objective, its availability constraints are reversed, and its balance constraints are inequalities ($\ge$) rather than equalities. As a result, it provides a lower bound on $\val{\pi}$: the steady-state quantities $(\bn^\pi,\bx^\pi)$ defined in \eqref{eq:nxpi} form a feasible solution of $\plblp(\pi)$ with objective value $\sum_{(i,j)\in\mathcal M(\pi)} r_{ij}x_{ij}^{\pi}=\val{\pi}$. Indeed, \Cref{lemma:balance:multiple} shows that these quantities satisfy \eqref{eq:balance-lb} with equality, while \Cref{lemma:balance:multiple,lem:gammabound} together imply that, for every $j\in\mathcal T$ and prefix $S\in\mathcal P_j(\pi)$,
\[
 \sum_{i\in S}x_{ij}^\pi
 =\lambda_j\mathbb P\!\left(\sum_{i\in S}N_i^\pi>0\right)
 \ge\lambda_j\gamma_S\sum_{i\in S}n_i^\pi,
\]
which establishes \eqref{eq:availability-lb}.

We now take the dual of $\plblp(\pi)$, associating a dual variable $v_i$ with the constraint \eqref{eq:balance-lb} of type $i$ and a dual variable $z_{Sj}$ with the constraint \eqref{eq:availability-lb} of type $j$ and prefix $S$. This yields the following maximization linear program:
\begin{lpbox}{Dual of $\plblp(\boldsymbol{\pi})$ }
\begin{align}
 \dlblp(\pi)
  \;\triangleq\; \max_{\mathbf{v},\mathbf{z} \ge 0}\quad
   & \sum_{i\in\mathcal T} \lambda_i v_i
   \nonumber\\[2pt]
  \textrm{s.t.}\quad
   & v_i+v_j
     + \sum_{S\in\mathcal P_j(\pi)} z_{Sj}\mathbf{1}(i\in S)
     \;\le\; r_{ij}
   && \forall (i,j)\in\mathcal M(\pi),
   \label{eq:replacement-reward}\\
   & \mu_i v_i
     \;\le\; \sum_{j\in\mathcal T}\sum_{S\in\mathcal P_j(\pi)}
       \lambda_j\gamma_S z_{Sj}\mathbf{1}(i\in S)
   && \forall i\in\mathcal T.
   \label{eq:replacement-balance}
\end{align}
\end{lpbox}
The notation $v_i$ for the dual variable of the constraint \eqref{eq:balance-lb} of type $i$ is deliberate: we will use a particular (not necessarily optimal) choice of these dual variables as the values of a value-based policy. With this viewpoint, each constraint \eqref{eq:replacement-reward} states that the values of the two agents consumed by an acceptable match, together with the  weights $z_{Sj}$ of the prefixes containing the waiting partner, do not exceed the reward of the match. Also, each constraint \eqref{eq:replacement-balance} states that the rate $\mu_i v_i$ at which value is lost to abandonment per waiting type-$i$ agent does not exceed a pessimistic estimate of the rate at which the agent's availability generates \emph{surplus} for arriving agents, as we explain next.

Fix an arriving type $j$, list its acceptable partners as $i_1\succ_j\cdots\succ_j i_m$, and let $P_k=\{i_1,\ldots,i_k\}$. When the constraints \eqref{eq:replacement-reward} of type $j$ hold with equality, as they will for the values we construct, the score $r_{i_kj}-v_{i_k}-v_j$ equals $\sum_{\ell\ge k}z_{P_\ell j}$, the total weight of the prefixes that contain $i_k$. As a result, this total weight accounts for the surplus of the match, that is, the part of its reward that exceeds the values of the two matched agents, which is its score $r_{i_kj}-v_{i_k}-v_j$.  Equivalently, $z_{P_kj}$ is the drop in  score from the $k$-th to the $(k+1)$-th type on the list, where the score after the last type is taken to be zero; as a result, $z_{P_\ell j}$ can be interpreted as the contribution of the prefix $P_\ell$ to the surplus generated by each match of an arriving type-$j$ agent with a partner in $P_\ell$. By \eqref{eq:availability-lb}, such matches occur at a rate of at least $\lambda_j\gamma_{P_\ell}\sum_{i\in P_\ell}n^\pi_i$. Intuitively speaking, spreading this contribution evenly over the waiting agents in $P_\ell$ credits each of them with surplus at rate $\lambda_j\gamma_{P_\ell}z_{P_\ell j}$. Summing over arriving types and prefixes that contain type $i$, each waiting type-$i$ agent is credited with surplus at rate $\sum_{j}\sum_{S\in\mathcal P_j(\pi)}\lambda_j\gamma_S z_{Sj}\mathbf{1}(i\in S)$, which is the right side of \eqref{eq:replacement-balance}.

% \na{Choice to call this $v_i$ is intentional (connect to value-based policy).}
% \na{I find the preceding paragraph hard to read because I have very little intuition for the meaning of the $z_{Sj}$. Should we do more to interpret/explain ``the weights of the prefixes'' and why the right side corresponds to ``the rate at which availability generates excess rewards''?}
\smallskip
Weak duality now provides conditions under which a vector of values serves as a lower-bound \emph{certificate} for the long-run average reward of a given static priority policy.

\begin{proposition}\label{prop:replacement-certificate}
Let \(\pi\) be a static priority policy and let \((\mathbf{v},\mathbf{z})\) be feasible for
\(\textup{DLP}^{\textup{LB}}(\pi)\). Then 
\[
 V(\pi)\ge\sum_{i\in\mathcal T}\lambda_i v_i.
\]
\end{proposition}
\begin{proof}
Since the primal program $\plblp(\pi)$ lower-bounds the long-run average reward  $V(\pi)$ of policy \(\pi\), weak duality gives $
 V(\pi)\ge\textup{LP}^{\textup{LB}}(\pi)
 \ge\sum_{i\in\mathcal T}\lambda_i v_i$, as desired.
\end{proof}

In summary, \Cref{prop:replacement-certificate} states that nonnegative values $\mathbf{v}$ provide the desired certificate for a policy $\pi$ if there exist nonnegative prefix weights $\mathbf{z}$ such that the pair $(\mathbf{v},\mathbf{z})$ satisfies \eqref{eq:replacement-reward} and \eqref{eq:replacement-balance}. We next look for such a certificate, together with its corresponding static priority policy, that also provides a bound on the omniscient benchmark.

\subsection{Bounding the Omniscient Benchmark via Values}
\label{sec:value-certificate}
%\na{Goal: use $(v,z)$ to relate to optimum.}
The second objective of this section is to find values $\mathbf{v}$, together with prefix weights $\mathbf{z}$, that ``cover'' the omniscient benchmark, meaning that $\sum_i\lambda_iv_i$ bounds it from above (up to a multiplicative factor). Since $\offrbrel$ upper-bounds the omniscient benchmark (\Cref{prop:offline-bound}), it is natural to look for such a bound through the dual of $\offrbrel$. As it turns out, the key is to reverse the inequalities in \eqref{eq:replacement-reward} and \eqref{eq:replacement-balance}. We formalize this claim in the following proposition, whose proof relates these reversed inequalities to the objective and constraints of the dual of $\offrbrel$. We sketch the proof below and defer the full proof to \Cref{apx:value-certificate-proofs}.
\begin{restatable}{proposition}{replacementhalf}
\label{prop:replacement-half}
Fix a static priority policy $\pi$. Suppose that
$(v_i)_{i\in\mathcal T}$ and
$(z_{Sj})_{j\in\mathcal T,\,S\in\mathcal P_j(\pi)}$
satisfy $v_i,z_{Sj}\ge0$ and
\begin{alignat}{2}
 & v_i+v_j
   +\sum_{S\in\mathcal P_j(\pi)}z_{Sj}\mathbf{1}(i\in S)
   \ge r_{ij}
 &\qquad & \forall i,j\in\mathcal T,
 \label{eq:replacement-reverse-reward}\\[2pt]
 & \mu_i v_i
   \ge\sum_{j\in\mathcal T}\sum_{S\in\mathcal P_j(\pi)}
       \lambda_j\gamma_S z_{Sj}\mathbf{1}(i\in S)
 &\qquad & \forall i\in\mathcal T.
 \label{eq:replacement-reverse-balance}
\end{alignat}
Then $\sum_i\lambda_i v_i\ge\tfrac12\offrbrel$.
\end{restatable}

\begin{proof}[Proof sketch]
Extending the prefix weights by zero to all subsets of $\mathcal{T}$, the reversed constraints \eqref{eq:replacement-reverse-reward} say precisely that $(\mathbf{v},\mathbf{z})$ is feasible for the dual of $\offrbrel$, whose objective is $$\sum_{i \in \types} \lambda_iv_i+\sum_{j \in \types} \sum_{S \subseteq \types} \lambda_j(1-e^{-\rho(S)})z_{Sj}.$$
Multiplying the constraints \eqref{eq:replacement-reverse-balance} by $\lambda_i/\mu_i$ and summing over $i$ shows that this objective is at most $2\sum_i \lambda_i v_i$. Weak duality then gives $\offrbrel\le2\sum_i\lambda_iv_i$.\footnote{This is a dual form of \Cref{prop:factor2}, which establishes an analogous bound between $\offrbrel$ and $\plpt$. See also \Cref{sec:suitable-dual}, which relates the dual of $\plpm$ to these conditions.}
\end{proof}

% \na{Proof sketch: show that $v$ and $z$ naturally correspond to a feasible solution of the dual of $\textup{LP}^\textup{OMN}$ with objective value at most $2 \sum_{i \in \m{T}} \lambda_i v_i$.}

\subsection{Satisfying Both Requirements via a Fixed Point}
\label{sec:fixed-point-dual}

\Cref{prop:replacement-certificate,prop:replacement-half} 
show that a static priority policy $\pi$ earns at least $\half\offrbrel$, and hence at least half of the omniscient benchmark, whenever the policy \(\pi\), the values \(\mathbf{v}\), and the prefix weights \(\mathbf{z}\) satisfy \eqref{eq:replacement-reward}, \eqref{eq:replacement-balance}, \eqref{eq:replacement-reverse-reward}, and \eqref{eq:replacement-reverse-balance}. These four sets of conditions are equivalent to the following:
% \na{make sure I understand where we use/need \eqref{eq:replacement-joint-rejected}.}
\begin{align}
 v_i+v_j+\sum_{S\in\mathcal P_j(\pi)}z_{Sj}\mathbf1(i\in S)
 &=r_{ij}
 &&\forall(i,j)\in\mathcal M(\pi),\label{eq:replacement-joint-reward}\\
 \mu_i v_i
 &=\sum_j\sum_{S\in\mathcal P_j(\pi)}
       \lambda_j\gamma_S z_{Sj}\mathbf1(i\in S)
 &&\forall i\in\mathcal T,\label{eq:replacement-joint-balance}\\
v_i+v_j&\ge r_{ij}
 &&\forall(i,j)\notin\mathcal M(\pi),\label{eq:replacement-joint-rejected}\\
 %v_i&\ge0 &&\forall i\in\mathcal T,\label{eq:replacement-joint-v}\\
 z_{Sj}&\ge0
 &&\forall j\in\mathcal T,\ S\in\mathcal P_j(\pi).
 \label{eq:replacement-joint-z}
\end{align}
Specifically, \eqref{eq:replacement-reward} and \eqref{eq:replacement-reverse-reward} are equivalent to \eqref{eq:replacement-joint-reward} and \eqref{eq:replacement-joint-rejected}, while \eqref{eq:replacement-balance} and \eqref{eq:replacement-reverse-balance} are equivalent to \eqref{eq:replacement-joint-balance}. The nonnegativity of $\mathbf{v}$ is implied by \eqref{eq:replacement-joint-balance} and \eqref{eq:replacement-joint-z}.
%Indeed, \eqref{eq:replacement-joint-reward} and \eqref{eq:replacement-joint-balance} make $(\mathbf{v},\mathbf{z})$ feasible for $\textup{DLP}^{\textup{LB}}(\m{I},\pi)$, so that \Cref{prop:replacement-certificate} applies, while \eqref{eq:replacement-joint-reward}--\eqref{eq:replacement-joint-balance} are the conditions from \Cref{prop:replacement-half}: on acceptable pairs the reversed reward constraints hold with equality, and on unacceptable pairs, which belong to no prefix, they reduce to \eqref{eq:replacement-joint-rejected}. 

How can we find a policy \(\pi\), values \(\mathbf{v}\), and prefix weights \(\mathbf{z}\) that satisfy these conditions? The conditions themselves point the way. Fix an arriving type $j$. Recall from an earlier discussion in  \Cref{sec:value-certificate-1} that by \eqref{eq:replacement-joint-reward} and \eqref{eq:replacement-joint-z}, the scores $r_{ij}-v_i-v_j$ of its acceptable partners are nonnegative and nonincreasing along $\succ_j$, and the weight of each prefix is the drop in score from one type on the list to the next. By \eqref{eq:replacement-joint-rejected}, every match with a positive score is acceptable. Thus, once the values are fixed, the prefix weights are determined, and so are the priority lists, up to ties and zero-score matches. The conditions \eqref{eq:replacement-joint-balance} then become a system of equations in the values alone.

We now write this system of equations explicitly for a general set of allowed matches $\m{M}$---as we need it later in  \Cref{sec:restricted-values,sec:connections}. Fix a set of matches $\m{M}\subseteq\TT$, recall that $\m{M}_j=\{i:(i,j)\in\m{M}\}$, and for a vector of values $\mathbf{v}$ write $u_{ij}(\mathbf{v})=(r_{ij}-v_i-v_j)^+$ for the \emph{residual score} of an $(i,j)$ match, that is, the positive part of its score. For an arriving type $j$ and a nonempty set $S\subseteq\m{M}_j$, let
\begin{equation}\label{eq:replacement-Z}
 Z^{\m{M}}_{Sj}(\mathbf{v})\triangleq
 \left(\min_{i\in S}u_{ij}(\mathbf{v})
       -\max_{i\in\m{M}_j\setminus S}u_{ij}(\mathbf{v})\right)^+,
\end{equation}
where the maximum or minimum over the empty set is taken to be zero, and for every $i \in \m{T} $ let
\begin{equation}\label{eq:replacement-F}
 F^{\m{M}}_i(\mathbf{v})\triangleq\frac1{\mu_i}\sum_{j\in\mathcal T}\lambda_j
       \sum_{S\subseteq\m{M}_j}
         \gamma_S Z^{\m{M}}_{Sj}(\mathbf{v})\mathbf1(i\in S).
\end{equation}

The weight $Z^{\m{M}}_{Sj}(\mathbf{v})$ is positive only if every type in $S$ has a positive score that is strictly larger than the score of every type in $\m{M}_j\setminus S$; it is then the gap between the smallest score inside $S$ and the largest score outside $S$, where negative scores are replaced by zero. The sets carrying positive weight for type $j$ are therefore nested: each is a prefix of the list of types in $\m{M}_j$ ranked in decreasing order of the residual scores $u_{ij}(\mathbf{v})$. Moreover, the weights telescope along this chain:
\begin{equation}\label{eq:telescope-M}
 \sum_{S\subseteq\m{M}_j}Z^{\m{M}}_{Sj}(\mathbf{v})\,\mathbf1(i\in S)=u_{ij}(\mathbf{v})\qquad\forall (i,j)\in\m{M}.
\end{equation}
We write $\mathbf{Z}^{\m{M}}(\mathbf{v})=(Z^{\m{M}}_{Sj}(\mathbf{v}))_{j,S}$ for the vector of these weights. When $\m{M}=\TT$ we drop the superscript and write $Z_{Sj}(\mathbf{v})$ and $F_i(\mathbf{v})$.

Returning to the unrestricted case $\m{M}=\TT$, the following proposition, whose proof is deferred to \Cref{apx:tight-value-proofs}, confirms that finding a policy \(\pi\), values \(\mathbf{v}\), and prefix weights \(\mathbf{z}\) that satisfy \eqref{eq:replacement-joint-reward}--\eqref{eq:replacement-joint-z} reduces to finding a \emph{fixed point} of the map $\mathbf{F}$.

% \na{Right now, the definition of $Z_{Sj}^{\m{M}}$ sort of comes from nowhere. My instinct is to include this definition as part of \Cref{lem:replacement-equivalence}. The advantage of this is that then we at least know ``why we care'' about the function $Z$ (we want to know which $v$ admit a good $(v,z,\pi)$ triple). The disadvantage is that \Cref{lem:replacement-equivalence} might be a little more cumbersome to state, and the general $\m{M}$ version is still not motivated.}
\begin{restatable}{proposition}{equivalence}
\label{lem:replacement-equivalence}
A vector \(\mathbf{v}\ge0\) admits a static priority policy \(\pi\) and  weights
\(\mathbf{z}\) satisfying
\eqref{eq:replacement-joint-reward}--\eqref{eq:replacement-joint-z}
only if \(\mathbf{F}(\mathbf{v})=\mathbf{v}\). Conversely, if $\mathbf{F}(\mathbf{v})=\mathbf{v}$, then every static priority policy $\pi$ consistent with $\mathbf{v}$ (\Cref{defn:dual-consistent-policy}), together with the weights $z_{Sj}=Z_{Sj}(\mathbf{v})$ for $j\in\types$ and $S\in\mathcal{P}_j(\pi)$, satisfies \eqref{eq:replacement-joint-reward}--\eqref{eq:replacement-joint-z}.
\end{restatable}

The fixed-point equation $\mathbf{F}(\mathbf{v})=\mathbf{v}$ has an intuitive economic interpretation. For each type $i$, we can write it as $\mu_i v_i=\mu_i F_i(\mathbf{v})$. Its left side is the rate at which value is lost to abandonment per waiting type-$i$ agent, and its right side is the pessimistic estimate, discussed earlier, of the rate at which that agent's availability generates surplus for arrivals of all types. A fixed point makes these two rates equal and, in this sense, yields a ``balanced'' vector of values---cf.\ the balanced thresholds and prices used in prophet inequalities~\citep{kleinberg2012matroid,dutting2020prophet}.

Taken together, these results show that the existence of a fixed point of $\mathbf{F}$ suffices for the guarantee of $\half$ in this section. In fact, it is straightforward to use Brouwer's theorem~\citep{brouwer1911abbildung} to show the existence. We go a step further: for every set of matches $\m{M}$, the fixed point of the restricted map $\mathbf{F}^{\m{M}}$ is unique and can be \emph{computed} in polynomial time. (See footnote~\ref{fn:arithmetic} for our convention about arithmetic representations). We sketch the proof below and defer the full proof to \Cref{apx:tight-value-proofs}. Notably, uniqueness is used in \Cref{sec:restricted-values,sec:connections}. Moreover, \Cref{sec:connections} also gives an alternative, LP-based way to compute the restricted fixed point for a broad class of sets $\m{M}$ and shows that it is sufficient for our purpose of obtaining $\half$-competitive value-based policies.

% \begin{restatable}{lemma}{equivalence}
% \label{lem:replacement-equivalence} For a vector of values $v$, write \(u_{ij}(v)=(r_{ij}-v_i-v_j)^+\) and 
% \begin{equation}\label{eq:replacement-Z}
%  Z_{Sj}(v):=
%  \left(\min_{i\in S}u_{ij}(v)
%        -\max_{i\in\mathcal T\setminus S}u_{ij}(v)\right)^+,
% \end{equation}
% where the maximum over the empty set is zero. For every $i \in \m{T} $ let
% \begin{equation}\label{eq:replacement-F}
%  F_i(v)=\frac1{\mu_i}\sum_{j\in\mathcal T}\lambda_j
%        \sum_{\emptyset\ne S\subseteq\mathcal T}
%          \gamma_S Z_{Sj}(v)\mathbf1(i\in S).
% \end{equation}
% Then, a vector \(v\ge0\) admits a static priority policy \(\pi\) and prefix weights
% \(z\) satisfying
% \eqref{eq:replacement-joint-reward}--\eqref{eq:replacement-joint-z}
% if and only if \(F(v)=v\).
% In any such triple, \(\pi\) is consistent with $v$. \fs{I am tempted to write: The triple is completed by choosing any policy \(\pi\) consistent with $v$ and by setting $z_{Sj}=Z_{Sj}(v)$. }
% %Moreover, extending the prefix weights by zero gives \(z_{Sj}=Z_{Sj}(v)\) for
% %all nonempty \(S\subseteq\mathcal T\).
% %\na{Do we need to talk about extension, or do we just need $z_{Sj} = Z_{Sj}(v)$ for prefixes of $\pi$?}
% \end{restatable}

\begin{restatable}{proposition}{fixedpointunique}
\label{prop:fixed-point}
For every set of matches $\m{M}\subseteq\TT$, the map $\mathbf{F}^{\m{M}}$ has a unique fixed point in $\reals^{\types}$, which we call the \textbf{restricted fixed point} and denote by $\mathbf{v}^{\m{M}}$. Moreover, $\mathbf{v}^{\m{M}}\ge0$, and $\mathbf{v}^{\m{M}}$ can be computed in time polynomial in the input size.
\end{restatable}

% \na{My original instinct was to make \Cref{lemma:fixedpoint-uniqueness} a proposition and \Cref{lem:replacement-equivalence} a lemma. Might be worth a discussion (maybe I will change my mind), but also not a big deal for a first draft.}

\begin{proof}[Proof sketch]
Existence, uniqueness, and computation all rely on a one-sided Lipschitz bound, which we prove in \Cref{lem:F-one-sided}: with $\langle\mathbf{a},\mathbf{b}\rangle_{\boldsymbol{\lambda}}=\sum_{i}\lambda_ia_ib_i$ and $\|\mathbf{a}\|_{\boldsymbol{\lambda}}=\langle\mathbf{a},\mathbf{a}\rangle_{\boldsymbol{\lambda}}^{1/2}$,
\begin{equation}\label{eq:one-sided}
 \big\langle\mathbf{F}^{\m{M}}(\mathbf{v})-\mathbf{F}^{\m{M}}(\mathbf{w}),\,\mathbf{v}-\mathbf{w}\big\rangle_{\boldsymbol{\lambda}}\;\le\;\tfrac12\,\|\mathbf{v}-\mathbf{w}\|_{\boldsymbol{\lambda}}^2
 \qquad\text{for all }\mathbf{v},\mathbf{w}\in\reals^{\types}.
\end{equation}
Intuitively, raising the value of a type lowers the surplus credited to its availability, and in the $\lambda$-weighted inner product this effect outweighs the surplus that shifts to other types. The map $\mathbf{F}^{\m{M}}$ itself need not be a contraction, but \eqref{eq:one-sided} and a Lipschitz bound imply that the damped map $\mathbf{v}\mapsto(1-\theta)\mathbf{v}+\theta\mathbf{F}^{\m{M}}(\mathbf{v})$, which has the same fixed points, is a contraction with respect to $\|\cdot\|_{\boldsymbol{\lambda}}$ for small enough $\theta>0$. Banach's fixed-point theorem then gives existence and uniqueness, and the fixed point is nonnegative because $\mathbf{F}^{\m{M}}$ takes nonnegative values. For the computation, ellipsoid methods for fixed points of contractions~\citep{sikorski1993ellipsoid,daskalakis2026polynomial} approximate the fixed point to within any accuracy $\varepsilon$ in time polynomial in the input size and $\log(1/\varepsilon)$, using only evaluations of $\mathbf{F}^{\m{M}}$, each of which amounts to sorting scores. Since $\mathbf{F}^{\m{M}}$ is piecewise affine, a sufficiently accurate approximation determines the fixed point exactly.
% Existence follows from Brouwer's theorem, as $\mathbf{F}^{\m{M}}$ is continuous and maps a box into itself (since each $Z^{\m{M}}_{Sj}(\mathbf{v})$ is bounded by $r_{ij}$). To prove uniqueness we compare two fixed points at a coordinate where they differ the most. At such a coordinate, the larger vector has the smaller scores, so the excess reward that its availability generates is accumulated over a shorter range of score thresholds and, at each threshold, over a larger prefix, whose coefficient $\gamma_S$ is smaller because $\gamma$ decreases under inclusion. The right side of its fixed-point equation is therefore no larger than that of the smaller vector, while its left side is strictly larger, which is impossible. \rn{@Pranav: can you check and revise this sketch if needed?} \pn{I think it's okay, but somehow it is a bit too polished.}
\end{proof}

Putting all the pieces together, we can prove the following main result of this section.
\begin{theorem}
\label{thm:values-half}
There exists a vector of values $\mathbf{v} \ge 0$ such that for any static priority policy $\pi$ consistent with $\mathbf{v}$, $V(\pi) \ge \frac{1}{2}\offrbrel$. Moreover, such a vector, and a static priority policy consistent with it, can be computed in time polynomial in the input size.
\end{theorem}
\begin{proof}
Let $\mathbf{v}=\mathbf{v}^{\TT}$ be the fixed point of $\mathbf{F}$ given by \Cref{prop:fixed-point}, and let $\pi$ be any static priority policy consistent with $\mathbf{v}$. By \Cref{lem:replacement-equivalence}, the weights $z_{Sj}=Z_{Sj}(\mathbf{v})$ on the prefixes of $\pi$ complete a triple $(\pi,\mathbf{v},\mathbf{z})$ satisfying \eqref{eq:replacement-joint-reward}--\eqref{eq:replacement-joint-z}. Hence $(\mathbf{v},\mathbf{z})$ is feasible for $\dlblp(\pi)$, and \Cref{prop:replacement-certificate} gives $V(\pi)\ge\sum_i\lambda_iv_i$, while $(\mathbf{v},\mathbf{z})$ satisfies the conditions of \Cref{prop:replacement-half}, which gives $\sum_i\lambda_iv_i\ge\tfrac12\offrbrel$. Combining the two inequalities proves the desired statement. Finally, by \Cref{prop:fixed-point}, $\mathbf{v}^{\TT}$ can be computed in polynomial time, and a policy consistent with it is obtained by ranking, for each arriving type, its partners with positive scores in decreasing order of score.
\end{proof}

% \begin{remark}[\textbf{polynomial-time computability}]
% \label{rem:polytime-dual}
% While our result so far yields a $\half$-competitive value-based static priority policy, it does not provide a polynomial-time construction unless we can compute the fixed-point of the map $\mathbf{F}$ in polynomial time. Interestingly, as we explain in the next subsection and later in \Cref{sec:connections}, restricting acceptable matches to a suitably chosen set $\mathcal M$---in fact, the same set returned by \Cref{algo:suitability} in \Cref{sec:primal}---allows us to construct a $\half$-competitive value-based policy in polynomial time by solving a polynomial-time solvable linear program.\rn{@Pranav: please check} \pn{Looks okay to me if we're not adding the polytime proof.}
% \end{remark}

\subsection{Discussion: A Family of Optimally Competitive Value-based Policies}
 \label{sec:restricted-values}
 The policy in \Cref{sec:primal} makes matches only within a set $\m{M}$ carefully selected by \Cref{algo:suitability}. In contrast, the value-based policy in \Cref{thm:values-half} is defined on the unrestricted set $\TT$, although \Cref{defn:dual-consistent-policy} and the map $\mathbf{F}^{\m{M}}$ allow for a general set $\m{M}\subseteq\TT$. Moreover, \Cref{prop:fixed-point} establishes the existence of a unique fixed point not only for $\mathbf{F}$, but also for every restricted map $\mathbf{F}^{\m{M}}$. This raises the question of whether we can generalize the machinery developed in this section to design a broader class of value-based policies, each of which is restricted to a properly chosen set $\m{M}$ and remains $\half$-competitive against the omniscient benchmark.\footnote{An arbitrary set $\m{M}$ need not preserve this guarantee; for example, $\m{M}=\emptyset$ precludes all matches.}

As a natural step, consider computing the restricted fixed point $\mathbf{v}^{\m{M}}$ and running a static priority policy consistent with $\m{M}$ and $\mathbf{v}^{\m{M}}$.
% Such a policy accepts every match in $\m{M}$ with positive score $r_{ij}-v^{\m{M}}_i-v^{\m{M}}_j$, may also accept zero-score matches in $\m{M}$, and for each type ranks its acceptable partners by nonincreasing score. 
As we show next, we can use the values $\mathbf{v}^{\m{M}}$ to certify the reward of this policy, in a manner similar to the value lower-bound certificate in \Cref{sec:value-certificate-1} (\Cref{prop:replacement-certificate}).

\begin{restatable}{lemma}{restrictedcertificate}
\label{prop:restricted-certificate}
Let $\m{M}\subseteq\TT$, and let $\pi$ be a static priority policy consistent with $\m{M}$ and $\mathbf{v}^{\m{M}}$. Then:
\[
 V(\pi)\ge\sum_{i\in\m{T}}\lambda_i v^{\m{M}}_i.
\]
\end{restatable}
We prove this lemma in \Cref{apx:proofs-restricted} using weak duality: the restricted fixed point $\mathbf{v}^{\m{M}}$, together with its associated weights $z_{Sj}=Z^{\m{M}}_{Sj}(\mathbf{v}^{\m{M}})~,~\forall j\in\m{T},\ S\in\mathcal P_j(\pi)$, is  feasible for $\textup{DLP}^{\textup{LB}}(\pi)$ and \Cref{prop:replacement-certificate} gives the desired bound.

% \begin{proof}[Proof sketch]
% Since $\pi$ makes every positive-score match in $\m{M}$ acceptable and ranks acceptable partners of each type in nonincreasing order of score, every set $S$ with $Z^{\m{M}}_{Sj}(\mathbf{v}^{\m{M}})>0$ is a prefix of $\porder{j}$. We therefore set

% By \eqref{eq:telescope-M}, the constraints \eqref{eq:replacement-reward} hold with equality for every $(i,j)\in\m{M}(\pi)$, since these matches have nonnegative scores. Likewise, the fixed-point equation $\mathbf{F}^{\m{M}}(\mathbf{v}^{\m{M}})=\mathbf{v}^{\m{M}}$ implies that the constraints \eqref{eq:replacement-balance} hold with equality. Thus, the nonnegative pair $(\mathbf{v}^{\m{M}},\mathbf{z})$ is feasible for $\textup{DLP}^{\textup{LB}}(\pi)$, and \Cref{prop:replacement-certificate} (proved via weak duality) gives the desired bound .
% \end{proof}

Restricting to $\m{M}$ preserves the lower-bound certificate for the policy but does not necessarily provide a comparison with the omniscient benchmark. To obtain this comparison, one could try applying \Cref{prop:replacement-half} directly and impose reversed constraints \eqref{eq:replacement-reverse-reward} and \eqref{eq:replacement-reverse-balance} on $\mathbf{v}^{\m{M}}$. However, as we show in \Cref{apx:proofs-restricted}, this would result in \emph{exactly} the same fixed point $\mathbf{v}^{\TT}$, and the same policy as in \Cref{thm:values-half}. A restriction that changes the fixed point (and the policy)  requires a different argument to compare the policy's lower-bound certificate with the omniscient benchmark.

\smallskip
The natural yardstick for this new comparison is the program $\plpt$ from \Cref{sec:primal}, since $\plpt\ge\half\offrbrel$ by \Cref{prop:factor2}. Our next lemma shows that, for \emph{every} $\m{M}$, the reward certified by the restricted fixed point is at least the value of $\plpm$.

\begin{restatable}{lemma}{fixedpointdualfeasible}
\label{lem:fixed-point-dual-feasible}
For every $\m{M}\subseteq\TT$, we have $\sum_{i\in\m{T}}\lambda_i v^{\m{M}}_i\ge\plpm$.
\end{restatable}
We prove this lemma in \Cref{apx:proofs-restricted} also using weak duality: the restricted fixed point, together with its associated weights, is feasible for the dual of $\plpm$. Given this result, for a policy consistent with $\m{M}$ and $\mathbf{v}^{\m{M}}$ to earn at least half of the omniscient benchmark, it suffices that
\begin{equation}\label{eq:restricted-requirement}
 \plpm\ge\plpt
\end{equation}
% In other words, a well-chosen set $\m{M}$ should ensure that the reward certified by the restricted fixed point is at least the value of $\plpt$.

% Consequently, \eqref{eq:restricted-requirement} holds whenever $\plpm\ge\plpt$, as is the case for the set returned by \Cref{algo:suitability} (\Cref{lemma:algfindssuitable}). 

Putting everything together suggests a recipe: let $\m{M}$ be any set such that \eqref{eq:restricted-requirement} holds, and let $\pi$ be any static priority policy consistent with $\m{M}$ and $\mathbf{v}^{\m{M}}$. \Cref{prop:restricted-certificate} bounds the policy's reward from below by $\sum_{i\in\m{T}}\lambda_i v^{\m{M}}_i$, while \Cref{lem:fixed-point-dual-feasible} shows that this certified value is at least $\plpm$. Combining these inequalities with \eqref{eq:restricted-requirement} and the benchmark bounds in \Cref{prop:factor2,prop:offline-bound} gives
\begin{equation}\label{eq:chain}
 V(\pi)
 \;\ge\;\sum_{i\in\m{T}}\lambda_i v^{\m{M}}_i
 \;\ge\;\plpm
 \;\ge\;\plpt
 \;\ge\;\half\offrbrel
 \;\ge\;\half\offopt.
\end{equation}
Thus, every such policy is $\half$-competitive against the omniscient benchmark. This more general approach, therefore, yields a larger class of $\half$-competitive value-based static priority policies, whose priorities are based on restricted fixed points $\mathbf{v}^{\m{M}}$ for properly chosen sets $\m{M}$, namely, those satisfying \eqref{eq:restricted-requirement}. This class contains the unrestricted value-based policies of \Cref{thm:values-half}, which correspond to $\m{M}=\TT$.
% For this choice, the above argument gives an alternative proof of the guarantee in \Cref{thm:values-half}: \Cref{prop:restricted-certificate} provides the lower-bound certificate, while \Cref{lem:fixed-point-dual-feasible,prop:factor2} establish the comparison with the benchmark through $\plpt$, rather than through \Cref{prop:replacement-half}. 
As another example, note that \eqref{eq:restricted-requirement} is true for the pruned set $\m{M}$ returned by \Cref{algo:suitability} (\Cref{lemma:algfindssuitable}). Interestingly, the resulting value-based policy \emph{coincides} with the primal policy from \Cref{sec:primal}, which we will show in the next section.

\section{Primal Meets Dual: Connections and Computations}
\label{sec:connections}
\Cref{sec:primal,sec:values} develop two seemingly different approaches to constructing a static priority policy that earns at least half of the omniscient benchmark. In \Cref{sec:primal}, the priority lists are read from the tight availability constraints of a suitable optimal solution to $\plpm$, and the policy's reward is certified by the piercing argument of \Cref{prop:plplowerbound}. In \Cref{sec:values}, the priority lists are determined by a vector of values satisfying the fixed-point equation for $\mathbf{F}$, and the reward is certified through the dual programs $\dlblp(\pi)$. The two approaches thus appear to rely on different objects: tight LP constraints in one case and fixed-point values in the other. In this section, we show that the two routes lead to the same place and  that these objects are directly connected: the primal priorities already encode a vector of restricted fixed-point values.

The key to this connection is the linear programming dual of $\plpm$, which we state in \Cref{sec:suitable-dual}. Here is how this section is organized to provide details of this connection:
\begin{enumerate}[label=(\roman*)]
    \item  We first show in \Cref{sec:suitable-dual} that if $\m{M}$ is properly selected so that $\plpm$ admits a suitable optimal solution (as defined in \Cref{sec:primal}), then the restricted fixed point $\mathbf{v}^{\m{M}}$, together with its associated weights, is an optimal solution to the dual of $\plpm$. Moreover, this optimal dual solution is unique, and solving the dual of $\plpm$ recovers  $\mathbf{v}^{\m{M}}$.
    \item We then show in \Cref{sec:primal-is-dual} that the static priority policy of \Cref{sec:primal} is a value-based policy consistent with $\m{M}$ and these same values $\mathbf{v}^{\m{M}}$. In short, the primal-based algorithm in \Cref{sec:primal} is, in fact, a dual-based algorithm in disguise.

\end{enumerate}
We highlight that this connection works in both directions: it provides a value-based interpretation of the primal priorities and, conversely, turns $\plpm$ and its dual into an LP-based alternative to \Cref{prop:fixed-point} for computing the restricted fixed point $\mathbf{v}^{\m{M}}$ whenever $\plpm$ has a suitable optimal solution. For the set returned by \Cref{algo:suitability}, this yields a $\half$-competitive polynomial-time value-based policy that may differ from the unrestricted policy in \Cref{sec:values}. We provide more details on the polynomial-time computation in \Cref{sec:connection-poly}. All proofs in this section, except the short proof of \Cref{thm:fixed-point-polytime}, are deferred to \Cref{apx:omittedproofs-connections}.

\subsection[Suitability and Recovering the Restricted Fixed Point via LP]{Suitability and Recovering the Restricted Fixed Point via LP}
\label{sec:suitable-dual}
By \Cref{prop:fixed-point}, the restricted fixed point $\mathbf{v}^{\m{M}}$ can be computed in polynomial time by solving the fixed-point equation directly. We now present an alternative based on linear programming that computes $\mathbf{v}^{\m{M}}$ for a broad class of sets $\m{M}$, including the set returned by \Cref{algo:suitability}.

Fix a set of matches $\m{M}\subseteq\TT$. Now consider the dual of $\plpm$ by associating a  dual variable $v_i$ with the balance constraint \eqref{eq:balance-lp-glb} for each type $i$, and a dual variable $z_{Sj}$ with the availability constraint \eqref{eq:availability-alg} for each type $j$ and nonempty set $S\subseteq\m{M}_j$:\footnote{Every feasible solution to $\plpm$ satisfies $n_i>0$ for all $i\in\m{T}$; see the proof of \Cref{lemma:suitablestructure}. Complementary slackness therefore implies that the dual  constraints \eqref{eq:dual-alg-balance} hold with equality at every optimal solution to $\dlpm$.}
\begin{lpbox}{Dual of $\plpm$}
\begin{align}
    \dlpm \triangleq \min_{\mathbf{v}, \mathbf{z} \ge 0} \quad &  \sum_{i \in \m{T}} \lambda_i v_i \nonumber \\
\textrm{s.t.}\quad &
    v_i + v_j + \sum_{\emptyset\ne S \subseteq \m{M}_j} z_{Sj}\,\mathbf{1}(i \in S) \;\ge\; r_{ij} & \quad & \forall (i,j) \in \m{M}, \label{eq:dual-alg-reward}\\
    & \mu_i v_i \;\ge\; \sum_{j \in \m{T}} \sum_{\emptyset\ne S \subseteq \m{M}_j} \lambda_j \gamma_S z_{Sj}\, \mathbf{1}(i \in S) & \quad & \forall i \in \m{T}. \label{eq:dual-alg-balance}
\end{align}
\end{lpbox}
Importantly, these constraints have the same form as those of $\textup{DLP}^{\textup{LB}}(\pi)$, but with the inequalities reversed. The constraints \eqref{eq:dual-alg-reward} are imposed on the pairs in $\m{M}$ rather than on the acceptable pairs of a policy, and the weights are assigned to all nonempty subsets of $\m{M}_j$ rather than only to prefixes. The direction of the inequalities reflects the direction of the certificate: by weak duality, every feasible solution to $\dlpm$ satisfies $\sum_i\lambda_i v_i\ge\plpm$, whereas every feasible solution to $\textup{DLP}^{\textup{LB}}(\pi)$ satisfies $\sum_i\lambda_i v_i\le V(\pi)$.

% the constraints in this dual LP have the same form as the conditions \eqref{eq:replacement-reverse-reward} and \eqref{eq:replacement-reverse-balance} in the statement of \Cref{prop:replacement-half}, and 

The restricted fixed point, together with its associated weights, satisfies both sets of constraints at once: it is feasible for $\dlpm$ (see the proof of \Cref{lem:fixed-point-dual-feasible}), and, with its weights restricted to the prefixes of $\pi$, it is feasible for $\textup{DLP}^{\textup{LB}}(\pi)$ for every policy $\pi$ consistent with $\m{M}$ and $\mathbf{v}^{\m{M}}$ (see the proof of \Cref{prop:restricted-certificate}). Moreover, the constraints \eqref{eq:dual-alg-reward} are tight on every pair in $\m{M}$ with nonnegative score. One might therefore hope to recover $\mathbf{v}^{\m{M}}$ by solving $\dlpm$. However, this need not work for an arbitrary $\m{M}$: \Cref{ex:not_suitable_failure} in \Cref{sec:additional-examples} exhibits an instance with $\m{M}=\TT$ in which the optimal dual values are not a fixed point, and the value-based policy constructed from them earns less than $\plpt$.

Despite the above negative result, we show that the restricted fixed point and its associated weights form the unique optimal solution to $\dlpm$ if and only if $\plpm$ has a suitable optimal solution in the sense of \Cref{sec:suitable}. Thus, the exact same condition that allows us to read a policy from the primal solution in \Cref{sec:primal} also ensures that solving the dual recovers $\mathbf{v}^{\m{M}}$ in \Cref{sec:values}.

\begin{restatable}{proposition}{dualsuitability}
\label{thm:dual-suitability}
Let $\m{M}\subseteq\TT$ be a set of matches. The following statements are equivalent:
\begin{enumerate}[(a)]
\item $\plpm$ has a suitable optimal solution.
\item Every optimal solution $(\mathbf{v}^*,\mathbf{z}^*)$ to $\dlpm$ satisfies
$\mathbf{z}^*=\mathbf{Z}^{\m{M}}(\mathbf{v}^*)$.
\end{enumerate}
Moreover, when these conditions hold, $\dlpm$ has a unique optimal solution, consisting of the restricted fixed point and its associated weights:
\[
 (\mathbf{v}^*,\mathbf{z}^*)
 =\big(\mathbf{v}^{\m{M}},\mathbf{Z}^{\m{M}}(\mathbf{v}^{\m{M}})\big).
\]
In particular,
\[
 \plpm=\dlpm=\sum_{i\in\m{T}}\lambda_i v^{\m{M}}_i.
\]
\end{restatable}
At a high level, the proof of this proposition is basically an exercise in complementary slackness and linear programming duality, and is deferred to \Cref{apx:proofs-dual-suitable}.

\subsection{The Primal Policy Is a Value-Based Policy}
\label{sec:primal-is-dual}
We now return to the policy we constructed in \Cref{sec:primal} via the primal approach. Given a suitable optimal solution $(\bn^*,\bx^*)$ to $\plpm$, consider any static priority policy consistent with it in the sense of \Cref{def:associated-greedy}. Such a policy makes only matches in $\m{M}$ with $x^*_{ij}>0$, and every tight set $S\in\m{P}_j(\bn^*,\bx^*)$ is a prefix of $\porder{j}$. The following result shows that this policy is consistent with $\m{M}$ and the dual values $\mathbf{v}^{\m{M}}$.

\begin{restatable}{proposition}{primaldualpolicies}
\label{prop:primal-dual-policies}
Let $\m{M}\subseteq\TT$, and let $(\bn^*,\bx^*)$ be a suitable optimal solution to $\plpm$. Then every static priority policy consistent with $(\bn^*,\bx^*)$ is also consistent with $\m{M}$ and $\mathbf{v}^{\m{M}}$.
\end{restatable}

\begin{proof}[Proof sketch] The full proof of the proposition is deferred to \Cref{apx:proofs-primal-is-dual}, and here we sketch the main ideas. The proposition again follows from complementary slackness, applied to $(\bn^*,\bx^*)$ and the unique optimal solution of $\dlpm$, which consists of $\mathbf{v}^{\m{M}}$ and its associated weights by \Cref{thm:dual-suitability}. A positive match rate $x^*_{ij}>0$ makes the reward constraint for $(i,j)$ tight, so every acceptable match has a nonnegative score. Moreover, every nonempty set obtained by thresholding the positive scores carries a positive dual weight and therefore has a tight primal availability constraint. Each such set must be a prefix of $\porder{j}$, which forces the policy to accept every positive-score match and rank acceptable partners by nonincreasing score.
% \footnote{The inclusion in (i) can be strict. If two types are symmetric in the restricted instance, an optimal vertex of $\plpm$ may break this symmetry by making one of their singleton availability constraints tight. Every policy consistent with this vertex must then rank that type first for the corresponding arriving type. However, the two types have equal scores under $\mathbf{v}^{\m{M}}$, so either order is consistent with $\m{M}$ and $\mathbf{v}^{\m{M}}$. In contrast, the optimal solution in (ii) lies in the relative interior of the optimal face, so an availability constraint is tight there only if it is tight at every primal optimum.}
\end{proof}
\Cref{prop:primal-dual-policies} completes the unifying picture promised at the beginning of this section. The primal approach constructs priorities by solving a linear program, without computing values, while the dual approach computes values as a fixed point, without solving a linear program. Nevertheless, when $\plpm$ has a suitable optimal solution, every policy that the primal approach constructs is also a value-based policy restricted to $\m{M}$: whichever suitable optimal solution of $\plpm$ the primal approach uses, the resulting policy accepts only matches in $\m{M}$ with nonnegative score $r_{ij}-v^{\m{M}}_i-v^{\m{M}}_j$, accepts every match in $\m{M}$ with a positive score, and ranks acceptable partners in decreasing order of score. In other words, the tight availability constraints that the primal approach reads from its solution agree with the ranking by score under the restricted fixed point. For the set returned by \Cref{algo:suitability}, every pair in $\m{M}$ is used at a positive rate, so every pair in $\m{M}$ has a nonnegative score, and the policy of \Cref{sec:primal} accepts every match in $\m{M}$ with a positive score, possibly together with some matches of score zero.

\subsection{Computing Values in Polynomial Time via LP}
\label{sec:connection-poly}
It remains to describe the computational requirements of the procedure: running \Cref{algo:suitability}, solving $\dlpm$, and sorting the acceptable partners of each arriving type by score. By \Cref{lemma:algfindssuitable}, \Cref{algo:suitability} runs in polynomial time, since each program $\plpm$ can be solved in polynomial time (\Cref{prop:polytime}). Although $\dlpm$ has exponentially many variables, the ellipsoid method applied to $\plpm$ also yields an exact optimal solution of $\dlpm$ with only polynomially many nonzero weights~\citep{GrotschelLovaszSchrijver1988}. Formally, we have the following proposition.

\begin{restatable}{proposition}{polytimedual}
\label{prop:polytime-dual}
For every $\m{M}\subseteq\TT$, an optimal solution of $\dlpm$ with polynomially many nonzero weights $z_{Sj}$ can be computed in time polynomial in the input size.
\end{restatable}

The proof is deferred to \Cref{apx:proofs-connection-poly}. Combined with \Cref{lemma:algfindssuitable}, this result gives the following computational consequence of the primal--dual connection.

\begin{theorem}
\label{thm:fixed-point-polytime}
Let $\m{M}$ be the set of matches returned by \Cref{algo:suitability}, and let $(\mathbf{v}^*,\mathbf{z}^*)$ be an optimal solution of $\dlpm$. Then $\mathbf{v}^*=\mathbf{v}^{\m{M}}$ is the fixed point restricted to $\m{M}$, every static priority policy consistent with the suitable optimal solution returned by \Cref{algo:suitability} is consistent with $\m{M}$ and $\mathbf{v}^*$, and every static priority policy consistent with $\m{M}$ and $\mathbf{v}^*$ satisfies \eqref{eq:chain}. Moreover, $\m{M}$, $\mathbf{v}^*$, and a static priority policy consistent with $\m{M}$ and $\mathbf{v}^*$ can be computed in time polynomial in the input size.
\end{theorem}
\begin{proof}
By \Cref{lemma:algfindssuitable}, $\plpm\ge\plpt$ and $\plpm$ has a suitable optimal solution, so \Cref{thm:dual-suitability} gives $\mathbf{v}^*=\mathbf{v}^{\m{M}}$, \Cref{prop:primal-dual-policies} gives the consistency claim, and \Cref{lem:fixed-point-dual-feasible} and the discussion following it in \Cref{sec:restricted-values} give \eqref{eq:chain}. \Cref{algo:suitability} solves at most $|\m{T}|^2+1$ linear programs of the form $\plp(\cdot)$, since every iteration but the last removes a pair from the current set; each solve takes polynomial time by \Cref{prop:polytime}, the final solve of $\dlpm$ takes polynomial time by \Cref{prop:polytime-dual}, and sorting by score takes polynomial time as well.
\end{proof}

We close with two remarks. First, the values computed in this way solve the fixed-point equation restricted to the final set $\m{M}$, but need not solve the unrestricted equation $\mathbf{F}(\mathbf{v})=\mathbf{v}$ used in \Cref{thm:values-half}: a deleted pair may have a positive score under $\mathbf{v}^{\m{M}}$. Second, the guarantee does not depend on which suitable optimal solution to $\plpm$ is used to construct a policy consistent with it (\Cref{def:associated-greedy}). By \Cref{prop:primal-dual-policies}, every such policy is consistent with $\m{M}$ and $\mathbf{v}^{\m{M}}$ and therefore satisfies \eqref{eq:chain}. More generally, the same guarantee holds for every policy consistent with $\m{M}$ and $\mathbf{v}^{\m{M}}$, regardless of how ties among equal scores are broken. 
% Add content showing two things. If we take the dual of $LP^{ALG}(\m{M})$ and let $(v^*,z^*)$ be the dual-optimal solution.
% \begin{enumerate}
% \item Then the policy defined by $(v^*,z^*)$ is equivalent to the policy defined by $(n^*,x^*)$. \na{What is the precise statement?}
% \item The pair $(v^*,z^*)$ is a solution to the $\m{M}$-restricted fixed point equation (which we could introduce here).
% \end{enumerate}

% Pull from Sections 4.2 and 4.4 in main.tex.

% \begin{lemma}\label{lem:replacement-existence}
% For every instance \(\mathcal I\), the map \(F\) has a unique fixed point in
% \(\mathbb R_+^{\mathcal T}\).
% \end{lemma}
% \begin{proof}
% Each \(Z_{Sj}(v)\) is continuous,
% since its definition uses only finite minima, maxima, and positive parts of
% continuous functions. Thus \(F\) is continuous. Moreover,
% \(0\le\gamma_S\le1\) and \eqref{eq:replacement-telescope}, which holds for every
% \(v\ge0\), imply
% \[
%  0\le F_i(v)\le\frac1{\mu_i}\sum_j\lambda_j u_{ij}(v)
%  \le C_i:=\frac1{\mu_i}\sum_j\lambda_j r_{ij}.
% \]
% Therefore, \(F\) maps the nonempty compact convex set
% \(\prod_{i\in\mathcal T}[0,C_i]\) into itself. Brouwer's theorem yields a fixed
% point \(v\).
% \end{proof}

% By Lemmas~\ref{lem:replacement-equivalence} and~\ref{lem:replacement-existence},
% there exist \(\pi,v,z\) satisfying
% \eqref{eq:replacement-joint-reward}--\eqref{eq:replacement-joint-z}.
% Propositions~\ref{prop:replacement-certificate} and~\ref{prop:replacement-half} give
% \[
%  V(\pi,\mathcal I)\ge\sum_{i\in\mathcal T}\lambda_i v_i
%  \ge\tfrac12\textup{OMN}(\mathcal I).
% \]

\section{Conclusion}
\label{sec:conclusion}

\noindent{\textbf{Summary.}} We studied dynamic matching with abandonment, in which agents of different types arrive over time, wait in a common pool from which they may abandon, and are matched in pairs for rewards that depend on their types. We showed that a simple, deterministic static priority policy, computable in polynomial time, earns at least half of the omniscient benchmark, which is the best possible guarantee for online policies~\citep{kessel2022stationary}. Two complementary approaches lead to this guarantee. The primal approach reads priority lists from the tight constraints of a family of linear programs and uses pruning to find a program with a suitable optimal solution. The dual approach computes one value per type, as the unique solution of a fixed-point equation, and ranks matches by their scores. Both approaches rest on the same lower bound on agent availability (\Cref{lem:gammabound}), and they meet through the dual of these linear programs: the policies of the primal approach are value-based policies restricted to a set of allowed matches.

\medskip
\noindent{\textbf{Future directions.}} 
Our analysis certifies the reward of a policy through linear programs, rather than by computing the stationary distribution of the underlying Markov chain. It is this approach that makes static priority policies tractable to analyze in markets with many types, type-specific arrival and abandonment rates, and general compatibility structures.
% As we sketch in \Cref{sec:general-sojourn-time}, the guarantee also extends to sojourn times that are new better than used in expectation. 
We hope that these tools, and in particular the availability bound of \Cref{lem:gammabound}, will prove useful in other stochastic matching and resource allocation problems in which simple policies are desirable.

Our work also suggests several directions for future research. First, in the general model of dynamic matching with abandonment, a natural goal is to improve upon the competitive ratio of one half against the optimal online policy, which is a weaker benchmark than the omniscient one. Second, one could study online learning and optimization versions of the model, in which parameters such as the arrival and abandonment rates or the rewards are unknown. Finally, it would be interesting to extend our results to more general arrival processes and sojourn-time distributions, or to settings in which rewards change over time, for example, as a function of how long an agent has waited.

\section*{Disclosure of AI Usage}
\phantomsection\addcontentsline{toc}{section}{Disclosure of AI Usage}
One result in this paper was obtained with the help of an AI system. The proof that the restricted fixed point $\mathbf{v}^{\m{M}}$ can be computed in polynomial time for every set of matches $\m{M}$ (\Cref{prop:fixed-point}) was first found by OpenAI's Codex Astra, after several rounds of prompting and interaction with the authors, and was subsequently verified and simplified by the authors. All other results are the authors' own work. In particular, the authors established that, for the set $\m{M}$ returned by \Cref{algo:suitability}, the restricted fixed point $\mathbf{v}^{\m{M}}$ is given by the unique optimal solution of the dual of $\plpm$, and can therefore be computed in polynomial time by linear programming (\Cref{thm:dual-suitability,prop:polytime-dual} and \Cref{thm:fixed-point-polytime}). Apart from the proof of \Cref{prop:fixed-point}, AI tools were used only to edit the writing.

% The proof that the restricted fixed point $\mathbf{v}^{\m{M}}$ can be computed in polynomial time (\Cref{prop:fixed-point}) was first discovered by the OpenAI's Codex Astra after several prompting and interactions, and was then verified and simplified by the authors. All other mathematical and non-mathematical content of this paper is the authors' own work; we used AI tools only to polish the writing.

\newpage
\bibliographystyle{plainnat}
{\footnotesize
\bibliography{refs}
}

\newpage
\begin{APPENDIX}{}
\crefalias{section}{appendix}\crefalias{subsection}{subappendix}\crefalias{subsubsection}{subsubappendix}

\section{Further Related Work}
\label{sec:related-work}
This appendix complements \Cref{sec:intro} by positioning our model, results, and techniques within the broader literature on dynamic matching and online allocation in operations research, economics, and computer science.

\medskip
\noindent\textbf{Dynamic matching with abandonment.}
The guarantees known for our model and its special cases are summarized in \Cref{sec:framework} and \Cref{fig:positioning}; here we describe the approaches behind them. \citet{collina2020dynamic} and \citet{aouad2022dynamic} introduced the model independently. \citet{collina2020dynamic} upper-bound the omniscient benchmark by a linear program and match according to its solution with a randomized policy. \citet{aouad2022dynamic} upper-bound the optimal online policy by a linear program and use a feasible solution of a quadratically constrained variant of this program to guide a randomized policy; they also study a cost-minimization version motivated by car-pooling. \citet{li2023fully} design an LP-based randomized policy for a fully online model in which agents arrive i.i.d.\ over discrete time and have type-dependent sojourn-time distributions known only through their means; they extend its analysis to Poisson arrivals. When one side of the market must be matched upon arrival, \citet{kessel2022stationary} introduce the stationary prophet inequality and design pricing-based policies, and \citet{patel2024combinatorial} allow buyers with combinatorial demands and relate the guarantees achievable against their linear programming benchmark to offline contention resolution schemes. The relaxations in these papers have availability constraints for one waiting type at a time; \citet{amanihamedani2024improved} strengthen them to all subsets of waiting types and round the resulting program with correlated proposals. Our relaxation $\offrbrel$ extends these subset constraints to the general model (\Cref{sec:benchmark-guarantee}).

Other work approximates the optimal online policy arbitrarily well in structured settings. When one side must be matched upon arrival and all waiting agents abandon at a common rate, \citet{amanihamedani2025adaptive} use queue-length information to design adaptive policies that approximate the optimal online policy to within a factor of $1-\varepsilon$, provided that the number of waiting types is constant or the instance is Euclidean of fixed dimension. They also show by example that, for this objective, adaptive policies can substantially outperform static ones on specific instances; by contrast, our results imply that static priority policies lose nothing relative to adaptive ones in terms of the worst-case ratio to the omniscient benchmark.

\medskip
\noindent\textbf{Prophet inequalities and Bayesian online allocation.}
Competitive analysis against the Bayesian omniscient benchmark follows the tradition of prophet inequalities. For a single item, a single threshold guarantees half of the expected maximum of independent nonnegative values, and no online policy can guarantee more \citep{krengel-sucheston_1978,samuelcahn1984comparison}; our main result can be viewed as an analog of this fact for dynamic matching with abandonment, with static priority lists in place of the threshold. Prophet inequalities have since been extended to richer feasibility constraints, including matroids \citep{kleinberg2012matroid} and matchings \citep{gravin2023prophet,ezra2022prophet}; for instance, \citet{ezra2022prophet} obtain the tight ratio of $\half$ for matchings in general graphs with vertex arrivals. Several of these guarantees can be obtained through prices that charge an allocation for the future value that it forecloses, a principle formalized by the balanced prices of \citet{dutting2020prophet}; the values of our value-based policies play a similar role (\Cref{sec:values}). Another common approach solves an ex-ante relaxation and rounds its solution online, using tools such as the magician's problem \citep{alaei2014bayesian} and online contention resolution schemes \citep{feldman2016ocrs}.

A related line of work compares policies with the optimal online policy instead of the prophet. For a single item, \citet{niazadeh2018prophet} show that the tight guarantees of single-threshold policies relative to the prophet remain tight relative to the optimal online policy; in contrast, in the single-good stationary prophet inequality, a posted-price policy earns at least $0.656$ of the reward of the optimal online policy \citep{kessel2022stationary}, more than the best possible ratio of $\half$ against the omniscient benchmark. Beyond single-item settings, \citet{anari2019pricing} give a polynomial-time approximation scheme relative to the optimal online policy for Bayesian online selection under laminar matroid constraints of constant depth and under production constraints, by strengthening the ex-ante relaxation into a hierarchy of linear programs. \citet{papadimitriou2024online} give a polynomial-time $0.51$-approximation of the optimal online policy for online stochastic max-weight bipartite matching with vertex arrivals, together with PSPACE-hardness of approximating this policy within some constant.

\medskip
\noindent\textbf{Reusable resources.}
Online allocation of reusable resources, in which an allocated unit is occupied for some time and then returns to inventory, has been studied under stochastic \citep{rusmevichientong2020reusable,dickerson2021allocation,feng2024reusable} and adversarial \citep{gong2022reusable,delong2024reusable,goyal2025reusable} arrivals. With Poisson arrivals and $k$ units of each resource, \citet{aminian2026reusable} attain a $1-\sqrt{2/(\pi k)}+O(1/k)$ fraction of the omniscient benchmark, improving upon \citet{feng2024reusable}, through a propose-then-admit scheme in which each resource runs a state-dependent admission-control policy; relatedly, \citet{aminian2026stationaryocrs} introduce stationary online contention resolution schemes, in which the distribution of the selected set does not depend on the arrival order. Like ours, the analysis of \citet{aminian2026reusable} relies on Poisson arrivals, steady-state arguments, and a linear program that bounds the omniscient benchmark. In reusable-resource models, however, a fixed set of resources cycles between busy and idle states, whereas in our model matched agents leave the system and future matching opportunities are created by new arrivals on every side of the market.

\medskip
\noindent\textbf{Online matching with arrivals and departures.}
In the classical online bipartite matching problem of \citet{kvv_1990} and its variants with known arrival distributions \citep{feldman2009online,manshadi2012online,jaillet2014online}, the vertices on one side are present throughout, and each vertex on the other side must be matched upon arrival or lost, so the policy never chooses when to match (see \citealp{mehta2013online} for a survey). In min-cost bipartite perfect matching with delays \citep{ashlagi2017metric}, the algorithm does choose when to match, but waiting incurs a linear cost and no request abandons. In the adversarial \emph{fully online matching} model, all vertices arrive online and remain available until a deadline, at which the algorithm can still match them; for unweighted matching, competitive ratios above $\half$ are achievable \citep{huang2019tight,huang2018match,huang_2020}. For edge-weighted matching in general graphs with adversarial arrivals and structured departures, in which agents can be matched when they become critical, \citet{ashlagi2023windowed} obtain constant competitive ratios and show that no algorithm is better than $\half$-competitive. \citet{truong2019prophet} study a finite-horizon two-sided model with stochastic arrivals, in which supply waits for a deterministic amount of time and demand must be matched upon arrival. In all of these models, the planner never loses an agent without warning: departures are known in advance, absent, or preceded by a last opportunity to match. In our model, the planner learns of an abandonment only after the agent has left.

\medskip
\noindent\textbf{Matching queues and large-market analysis.}
Dynamic matching also has a long history in applied probability, where work on matching rates under first-come-first-served policies \citep{caldentey2009fcfs,adan2012exact} and on the stability of stochastic matching models \citep{busic2013stability,mairesse2016stability,jonckheere2023maxweight} analyzes given policies rather than designing reward-maximizing ones. A control-oriented literature designs policies that are asymptotically optimal as the market grows \citep{gurvich2014control,ozkan2020dynamic}, in some cases only within natural classes of policies \citep{blanchet2022asymptotically}. \citet{aveklouris2025matching} show that, with heterogeneous match values, holding costs, and general patience distributions on both sides, a state-independent priority policy that ranks the edges is asymptotically optimal when patience hazard rates are increasing; \citet{liu2026datadriven} study a data-driven version of this problem. In a discrete-time model with abandonment, \citet{hu2022dynamic} show that, under certain conditions on the rewards, the optimal policy matches type pairs according to a priority hierarchy, combined with thresholds. In contrast to this literature, our guarantee requires neither a large market nor structural assumptions on the rewards: it holds for every instance, including thin markets in which both the option value of waiting and the risk of abandonment matter.

\medskip
\noindent\textbf{Dynamic matching without abandonment.}
When agents never abandon, queues can grow, and policies can track the solution of a fluid relaxation. Under a general position condition, several policies attain constant regret relative to hindsight or fluid benchmarks \citep{kerimov2024dynamic,gupta2024greedy,wei2026constant}. Under the same condition, \citet{kerimov2023optimality} show that, in two-way matching networks, a greedy static priority policy is hindsight optimal when the residual network of matches is acyclic (e.g., in two-sided networks), and \citet{kerimov2026availability} show that policies that only observe which queues are nonempty, including a static priority policy for acyclic networks, can attain the optimal regret. With abandonment, the gap to the omniscient benchmark can instead grow linearly over time: even in the single-good stationary prophet inequality, no online policy can guarantee more than half of the omniscient benchmark \citep{kessel2022stationary}.

\medskip
\noindent\textbf{Market thickness, departure information, and other objectives.}
A complementary literature studies when to match in markets with random compatibility, typically measuring the fraction of agents who leave unmatched or their waiting times rather than the total weighted reward. \citet{akbarpour2020thickness} show that, when the planner knows which agents are about to leave, waiting to thicken the market can sharply reduce the fraction of agents who leave unmatched, whereas without this information, greedy matching is nearly optimal; \citet{ashlagi2023matching} show that greedy matching is also asymptotically optimal in large imbalanced markets. \citet{baccara2020optimal} characterize optimal dynamic matching with waiting costs and match values that depend on the types of both agents. Closer to the queueing formulation, \citet{kohlenberg2024cost} derive nonasymptotic scaling laws for the loss caused by abandonment in two-sided matching queues, and \citet{kohlenberg_2025} identifies conditions under which restricting the compatibility graph and then matching greedily is scale-optimal, a policy structure close in spirit to static priority policies. In contrast to these papers, our guarantee allows arbitrary type-dependent rewards, is multiplicative, and holds for every instance.

\medskip
\noindent\textbf{Applications.}
The applications that motivate our model have extensive literatures of their own, including work on kidney exchange \citep{roth2004kidney,roth2005pairwise,unver2010dynamic,anderson2014dynamic,ashlagi2018matchrun,agarwal2019market,ashlagi2021kidney}, deceased-donor organ allocation \citep{bertsimas2013fairness,agarwal2021equilibrium}, public and affordable housing \citep{arnosti2020design,waldinger2021targeting,leshno2022dynamic}, and homeless services \citep{azizi2018designing,arnosti2026target}. These papers address application-specific questions, such as incentives, fairness, equilibrium behavior, and empirical performance, which complement our worst-case analysis of a common stochastic matching primitive.

\section[Poisson Processes: Definitions and Proofs]{Poisson Processes: Definitions and Proofs}
\label{sec:Poisson process}

% \na{This reads as quite technical. It introduces a lot of notation that the reader will not need in the rest of the paper. I think the main things we need are that the number of present agents follows a Poisson distribution, and some version of PASTA.
% If we need some of the remaining facts later, we can put them in the appendix. }
This appendix complements \Cref{sec:poisson-background}: it gives the formal definition of a Poisson process and the marking property, proves \Cref{fact:poisson-occupancy}, derives the consequence of \Cref{fact:poisson-arrival-calculus} that our analysis uses, and proves \Cref{cor:pasta}. For further background on Poisson processes, see \citet{lastpenrose2017lectures}.

A Poisson process with rate $\lambda(s)$ is a way to formalize the idea of creating a random set $\Phi$ in $\mathbb{R}^n$ by taking each tiny volume $ds$ in $\mathbb{R}^n$  and including it in $\Phi$ with probability $\lambda(s) ds $ independently.

\begin{definition}[Poisson process]\label{def:Poisson process}
A \textbf{Poisson process} $\Phi$ on a Borel set $\mathcal{X}$ $\subseteq \mathbb R^n$  with locally integrable rate function $\lambda:\mathcal{X}\to[0,\infty)$ is a random set of countably many points in $\mathcal{X}$ such that:
\begin{enumerate}
    \item For any Borel-measurable set $I \mathrel{\subseteq} \mathcal{X}$, $|\Phi \cap I|$ is a Poisson random variable with mean $\int_{I}\lambda(s) ds$.
    \item For any finite collection of disjoint Borel-measurable sets $I_1, I_2, \ldots, I_k$ each contained in $\mathcal{X}$, the random variables $|\Phi \cap I_j|$ are independent of each other.
\end{enumerate}
\end{definition}
We allow general rate functions on $\reals^n$ because marking (\Cref{fact:productppp}) produces Poisson processes on $\reals^2$ with nonconstant rates, such as the process of arrival times and sojourn times of the agents of a given type.

% \na{I am still a little confused about why we need an inhomogeneous process $\lambda(s)$ on an $n$-dimensional space. I guess it makes the proof of \ref{fact:productppp} succinct? But it seems more natural to focus on homogeneous processes on $\mathbb{R}$ (since that is all we need for the arrival process). From there, we should be able to prove \ref{fact:productppp} directly.}\pn{It is also useful in the proof of the stochastic lemma.}
The theory of Poisson processes shows that such a $\Phi$ exists. When we say that the agents arrive according to a Poisson process at rate $\lambda(s) = \lambda$, we are formalizing the idea that an average of $\lambda$ agents arrive in every unit of time, but the arrival of agents is independent across time.

The following property concerns points that carry independent random marks, such as the sojourn times of arriving agents \citep[see, e.g.,][Section~5.2]{lastpenrose2017lectures}.

\begin{fact}[Marking]
\label{fact:productppp}
If $\Phi$ is a Poisson process on a Borel set $\mathcal{X}\subseteq\mathbb R^n$ with rate function $\lambda(s)$, and, given $\Phi$, the points $s\in\Phi$ are associated with i.i.d.\ random variables $Y_s \in \mathbb{R}$ with density $f$, then $\{(s,Y_s):s\in\Phi\}$ is a Poisson process on $\mathcal{X}\times \mathbb R$ with rate function $\lambda(s)f(y)$.
\end{fact}

In the setting of our paper, agents of type $i$, arriving according to a Poisson process, are associated with independent sojourn times that are exponentially distributed with mean $\frac{1}{\mu_i}$. Thus, by \Cref{fact:productppp}, the pairs of arrival times and sojourn times of type-$i$ agents form a Poisson process as well.

It is useful for us to understand the number of agents of any particular type present at any time. This number is described by \Cref{fact:poisson-occupancy} in \Cref{sec:poisson-background}, which is a simple consequence of \Cref{fact:productppp}. We restate it here and give its proof.

\occupancy*
\begin{proof}
By \Cref{fact:productppp}, the type-$i$ pairs $(s,\ell)$ of arrival times $s$ and sojourn times $\ell$ (so that the departure time is $s+\ell$) form a Poisson process on $[0,\infty)\times[0,\infty)$ with rate $(s,\ell)\mapsto \arate{i}\drate{i}e^{-\drate{i}\ell}$.

% \na{I think this is missing an indicator ${\bf 1}(\ell \ge s)$, if I am understanding correctly.}\pn{It's arrival time, soujourn-time pairs, not arrival time departure time.}
The agents that are present at time $t$ are exactly the points in:
\[
  R=\{(s,\ell):0\le s\mathrel{\le} t,\ \ell>t-s\}.
\]
By \Cref{def:Poisson process}, the number of such points is Poisson with mean
\[
  \int_{R}\arate{i}\drate{i}e^{-\drate{i}\ell}\,d\ell\,ds
  =
  \arate{i}\int_0^t\int_{t-s}^{\infty}
      \drate{i}e^{-\drate{i}\ell}\,d\ell\,ds
  =
  \arate{i}\int_0^t e^{-\drate{i}(t-s)}\,ds
  =
  \tfrac{\arate{i}}{\drate{i}}\big(1-e^{-\drate{i}t}\big).
\]
The counts are independent across types because the arrival and sojourn processes are independent across types.
\end{proof}
The formula in \Cref{fact:poisson-arrival-calculus} is an expression of the fact that, conditional on a Poisson process having a point at a given location, the remaining points still form a Poisson process with the same rate function~\citep{slivnyak1962some, mecke1967stationare}. See \citet{coeurjolly2017tutorial} for a modern exposition, where this formula appears as equation $(17)$. A particularly useful corollary that we will apply is the following:

% \pn{Some explanation for my choices here: In the end, I settled on putting this Campbell-Mecke in the preliminaries because there's the next corollary that's also needed in the omniscient benchmark upper bound proof. The reason we have the next corollary and not just PASTA is because, at least formally speaking, we have to analyze the omniscient policy  for finite horizons $t$. I also decided on this deletion version of this Campbell Mecke formula, because eq(17) in Coerurjolly is exactly what we need, but it's not stated in the same way in Last and Penrose. Finally PASTA itself is a corollary of the next corollary, so it's nice to have it all here.}
\begin{corollary}
\label{cor:poisson-arrival-count}
Fix a type $j$ of agents with arrival rate $\lambda_j$. If, for each $s\ge0$, $E(s)$ is an event determined by the history $\hist(s)$ of the agents that arrived before time $s$, then for every $t\ge 0$,
\[
  \E{\sum_{a\in\agentset_j(t)} \indic{E(\atime{a})}}
  =
  \arate{j}\int_0^t \PP{E(s)}ds.
\]
\end{corollary}
\begin{proof}
Let $\Phi$ be the Poisson process of pairs of arrival times and sojourn times of type-$j$ agents, which has rate function $(s,\ell)\mapsto\lambda_j\mu_je^{-\mu_j\ell}$ by \Cref{fact:productppp}, and let $\Psi$ collect the arrival times and sojourn times of the agents of all other types. Since $\Psi$ is independent of $\Phi$, we may condition on $\Psi$ and apply \Cref{fact:poisson-arrival-calculus} to $\Phi$ with
\[
 g\big((s,\ell),\varphi\big)=\indic{s<t}\,\indic{E(s)\text{ occurs when the type-$j$ agents are given by }\varphi\text{ and the others by }\Psi}.
\]
For a type-$j$ agent $a$, whose arrival and sojourn times form the point $x\in\Phi$, the event $E(t_a)$ depends only on the agents that arrived before $t_a$, which do not include $a$ itself, so $g(x,\Phi\setminus\{x\})=\indic{t_a<t}\indic{E(t_a)}$. Similarly, $g\big((s,\ell),\Phi\big)=\indic{s<t}\indic{E(s)}$ by definition. \Cref{fact:poisson-arrival-calculus} therefore gives
\[
 \E{\sum_{a\in\agentset_j(t)}\indic{E(t_a)}\ \Big|\ \Psi}
 =\int_0^t\int_0^\infty\PP{E(s)\mid\Psi}\,\lambda_j\mu_je^{-\mu_j\ell}\,d\ell\,ds
 =\lambda_j\int_0^t\PP{E(s)\mid\Psi}\,ds ,
\]
and taking expectations over $\Psi$ completes the proof.
\end{proof}

\pasta*
\begin{proof}
Suppose first that $\pol$ is deterministic. Since $\pol$ is online, the state $\mathbf{N}^{\pol}(s)$ just before time $s$ is determined by the arrivals and departures that occur before time $s$. Hence, for each $s$, the event $\{\mathbf{N}^{\pol}(s)\in C\}$ depends only on the history before time $s$, and the first claim follows from \Cref{cor:poisson-arrival-count}. If $\pol$ uses randomization that is independent of the arrival and sojourn processes, the same argument applies conditionally on this randomization. For the second claim, divide the first by $t$ and note that $\tfrac1t\int_0^t\PP{\mathbf{N}^{\pol}(s)\in C}\,ds\to p$ whenever $\PP{\mathbf{N}^{\pol}(s)\in C}\to p$.
\end{proof}

\section{Omitted Proofs}
\label{sec:omitted-proofs}
\subsection{Omitted Proofs from \Cref{sec:preliminaries}}
\label{apx:omittedproofs-prelim}
\offlinebound*

\begin{proof}
Recall that for a policy $\pi$, $X_{ij}^\pi(t)$ is the number of $(i,j)$ matches made before time $t$, and
\[\xrate{i}{j}^{\pol}(t) = \tfrac{1}{t} \E{\Xmatch{i}{j}^{\pol}(t)}.\]

It suffices to show that the match rates $\xrate{i}{j}^{\pol}(t)$ of an arbitrary policy $\pol$ with hindsight information satisfy  \eqref{eq:omn-balance} and \eqref{eq:omn-match}.

\paragraph{Constraint~\eqref{eq:omn-balance}.} Fix $j \in \types$. Since each agent of type $j$ is matched at most once in every sample path, we have:

\[\sum_{i \in \types} \Xmatch{i}{j}^{\pol}(t) + \sum_{i \in \types} \Xmatch{j}{i}^{\pol}(t) \leq |\agentset_j(t)|.\]
Taking expectations and dividing by $t$  gives
$\sum_{i}  \xrate{i}{j}^{\pol}(t) + \sum_{i} \xrate{j}{i}^{\pol}(t) \le \arate{j}$, since $\tfrac1t \E{|\agentset_j(t)|} = \arate{j}$, which results in \eqref{eq:omn-balance}.

\paragraph{Constraint~\eqref{eq:omn-match}.} Fix $j \in \types$ and
$S \subseteq \types$. Any match counted in $\Xmatch{i}{j}^{\pol}(t)$ pairs a type-$j$ agent with a type-$i$ agent that arrived earlier and was still present when the type-$j$ agent arrived. For an arrival time $s \in [0,t]$ let
\[
  P_i(s) = \big|\{ b \in \agentset_i(s) : \dtime{b} > s \}\big|,
  \qquad
  P_S(s) = \sum_{i \in S} P_i(s)
\]
count the number of  type-$i$ agents, and of agents with types in $S$, that are present at time $s$. Thus, in every sample path, we have:
\[  \sum_{i \in S}\Xmatch{i}{j}^{\pol}(t) \leq  \big|\{ a \in \agentset_j(t) : P_S(\atime{a}) > 0 \}\big|.\]

We showed in \Cref{fact:poisson-occupancy} that for each fixed $s$, the counts $\{P_i(s)\}_{i \in \types}$ are Poisson with mean $\tfrac{\arate{i}}{\drate{i}}\big(1 - e^{-\drate{i} s}\big).$ Since this mean is at most $\arate{i}/\drate{i}$, and the counts are independent across types, we conclude that:
\[
  \PP{P_S(s) > 0} \le 1- e^{-\sum_{i \in S}\arate{i}/\drate{i}}.
\]
By \Cref{cor:poisson-arrival-count}, we conclude that:
\[\E{\sum_{i \in S}\Xmatch{i}{j}^{\pol}(t)} \leq \E{\big|\{ a \in \agentset_j(t) : P_S(\atime{a}) > 0 \}\big|}= \arate{j}\int_0^t \PP{P_S(s) > 0} ds \leq \lambda_j t \left( 1- e^{-\sum_{i \in S}\arate{i}/\drate{i}}\right),\]
as needed.
\end{proof}

\subsection{Omitted Proofs from \Cref{sec:priority-policies}}
\label{apx:omittedproofs-priority}
\stationaryctmc*
Fix a static priority policy $\pol$. The state only changes when an agent arrives or abandons the system, so from a state $\mathbf{n}$, the time to the next transition follows an exponential distribution with rate $\sum_{i \in \types}(\arate{i} + n_i \drate{i})$, and the next state depends only on $\mathbf{n}$. Therefore, $\mathbf{N}^{\pol}(t)$ is a continuous-time Markov chain. Note that the state changes only finitely many times (almost surely) in any finite interval.

Let $\m{N}^{\pol}$ be the set of states reachable from the empty state under policy $\pi$. (Note that $\m{N}^\pi$ does indeed depend on $\pi$: for example, if $\pi$ makes all matches as soon as possible, then $\m{N}^\pi$ is the set of states with at most one agent available.) The chain is irreducible on $\m{N}^{\pol}$, since from any state, the empty state is reached with positive probability within any interval of positive length (if all available agents abandon before the next arrival), and every state of $\m{N}^{\pol}$ is reachable from the empty state by definition.

It remains to prove ergodicity, for which we apply the Foster--Lyapunov drift
criterion of \citet[Theorem~2.8]{prieto2016uniform}: it suffices to exhibit a
nonnegative function $W$ on the state space, a finite set $C$, and constants
$c > 0$ and $b \in \reals$ such that
\begin{equation}
  \sum_{\mathbf{n}'} q_{\mathbf{n} \mathbf{n}'}\big(W(\mathbf{n}') - W(\mathbf{n})\big) \le -c + b\indic{\mathbf{n} \in C}
  \quad \text{for every state } \mathbf{n},
  \label{eq:lyapunov}
\end{equation}
where $q_{\mathbf{n}\mathbf{n}'}$ is the transition rate from $\mathbf{n}$ to $\mathbf{n}'$. Take
$W(\mathbf{n}) = \sum_{i \in \types} n_i$, the total number of available agents. From a
state $\mathbf{n}$, the quantity $W$ increases by one at total rate at most
$\sum_{i \in \types}\arate{i}$ (an arrival) and decreases by one at total rate
at least $\sum_{i \in \types} n_i \drate{i}$ (an abandonment), so
\[
  \sum_{\mathbf{n}'} q_{\mathbf{n} \mathbf{n}'}\big(W(\mathbf{n}') - W(\mathbf{n})\big)
   \le \sum_{i \in \types}\arate{i} - \sum_{i \in \types} n_i\drate{i}
   \le \sum_{i \in \types}\arate{i} - W(\mathbf{n})\min_{i}\drate{i}.
\]
Fix any $c > 0$ and put $C = \{ \mathbf{n} : W(\mathbf{n}) < k \}$ with
$k > (\sum_{i}\arate{i} + c)/\min_i \drate{i}$. For $\mathbf{n} \notin C$ we have $W(\mathbf{n}) \ge k$, so the
right side is at most $\sum_i \arate{i} - k\min_i\drate{i} < -c$; for
$\mathbf{n} \in C$ it is at most $\sum_i \arate{i} = -c + b$ with
$b = c + \sum_i \arate{i}$. Thus \eqref{eq:lyapunov} holds, and the chain is
ergodic and has a unique stationary distribution.

Now, we know that as $t \to \infty$, $\mathbf{N}^{\pol}(t)$ converges in distribution to $\mathbf{N}^{\pol}$, while $N_i^{\pol}(t)$ is uniformly integrable---it is dominated by
the number of type-$i$ agents present at time $t$, a Poisson variable with mean $\tfrac{\arate{i}}{\drate{i}}\left(1-e^{-\mu_i t}\right) \leq \frac{\lambda_i}{\mu_i}$. We conclude that:
\[\E{N_i^{\pol}}=
  \lim_{t \to \infty} \E{\Npres{i}^{\pol}(t)}.\]
Finally,  let $C_{ij}^{\pol}$ be the set of states in which an arriving type-$j$ agent is matched to an available type-$i$ agent. Since every $(i,j)$ match is made when a type-$j$ agent arrives, based on the state just before its arrival, we have $\Xmatch{i}{j}^{\pol}(t)=\sum_{a\in\agentset_j(t)}\indic{\mathbf{N}^{\pol}(t_a)\in C_{ij}^{\pol}}$. Moreover, $\PP{\mathbf{N}^{\pol}(s)\in C_{ij}^{\pol}}\to\PP{\mathbf{N}^{\pol}\in C_{ij}^{\pol}}$ as $s\to\infty$, because the state space is countable and $\mathbf{N}^{\pol}(s)$ converges in distribution. Since $\pol$ is online, applying \Cref{cor:pasta} with
$C=C_{ij}^{\pol}$ gives
\[
  \xrate{i}{j}^{\pol}
  =
  \lim_{t \to \infty} \tfrac{1}{t}
      \E{\Xmatch{i}{j}^{\pol}(t)}
  =
  \arate{j}\PP{\mathbf{N}^{\pol}\in C_{ij}^{\pol}}.\]
Consequently, the limit $\val{\pol}=\lim_{t\to\infty}\valt{\pol}$ also exists. Indeed, ${\valt{\pol}=\sum_{i,j\in\types}\rew{i}{j}\,\xrate{i}{j}^{\pol}(t)}$ is a finite linear combination of the rates $\xrate{i}{j}^{\pol}(t)$, each of which converges to $\xrate{i}{j}^{\pol}$, so ${\val{\pol}=\sum_{i,j\in\types}\rew{i}{j}\,\xrate{i}{j}^{\pol}}$.

\balancemultiple*

\begin{proof} Recall that $N_i^\pi(t)$ represents the number of type-$i$ agents available just before time $t$ and $X_{ij}^\pi(t)$ is the number of $(i,j)$ matches made before time $t$. The proof of Equation~\eqref{eq:Vpi} was already shown using the linearity argument at the end of the proof of \Cref{lemma:stationaryctmc}, so it remains to prove \eqref{eq:balance} and \eqref{eq:matchrate}.

\paragraph{Constraint~\eqref{eq:balance}.} Let $D_i^{\pi}(t)$ denote the number of type-$i$ agents that have abandoned the system (departed unmatched) by time $t$. Then for any time $t$ and any $i \in \m{T}$, we have:
\begin{equation}
     D_i^{\pi}(t) + N_i^{\pi}(t) + \sum_{j \in \mathcal{T}} X^\pi_{ij}(t) + \sum_{j \in \mathcal{T}} X_{ji}^\pi(t) =|\mathcal{A}_i(t)|.\label{departures-balance}
\end{equation}
This says that each agent who arrives before time $t$ has either abandoned, been matched with another agent, or is still available. Note that:
\[ \E{D_i^{\pi}(t)} = \E{\mu_i\int_0^t N_i^\pi(s)ds},\] because for any static priority policy $\pi$, $D_i^{\pi}(t) - \mu_i\int_0^t N_i^\pi(s)ds$ is a martingale \citep[see, e.g.,][]{pang2007martingale}. So by taking expectations, dividing by $t$, and taking the limit as $t \to \infty$ in \eqref{departures-balance}, constraint~\eqref{eq:balance} follows.

\paragraph{Constraint~\eqref{eq:matchrate}.} Let $C_{ij}^{\pol}$ be the set of states in which an arriving type-$j$ agent is matched to an available type-$i$ agent. Note that we showed (using \Cref{cor:pasta})
in the proof of \cref{lemma:stationaryctmc} that:
\[ \xrate{i}{j}^{\pol}
  = \arate{j}\PP{\mathbf{N}^{\pol}\in C_{ij}^{\pol}}.\] 
The events $\{C_{ij}^{\pol}:i\in\types\}$ are disjoint, so
\[\sum_{i \in S}\xrate{i}{j}^{\pol}
  =
  \arate{j}
  \PP{\mathbf{N}^{\pol}\in \bigcup_{i \in S} C_{ij}^{\pol}}.
\]
The event on the right can occur only if some type in $S$ is available, so
\eqref{eq:matchrate} follows. Furthermore, if
$S \in \prefixes{\porder{j}}$, then $S$ consists of the highest-ranked acceptable types for $j$, so whenever some type in $S$ is available the best
acceptable available type already lies in $S$ and the arriving type-$j$ agent is matched into $S$. Thus
\[
  \left\{\textstyle\sum_{i \in S}\Npres{i}^{\pol} > 0\right\}
  =
  \left\{\mathbf{N}^{\pol}\in \bigcup_{i \in S} C_{ij}^{\pol}\right\},
\]
 so \eqref{eq:matchrate} holds with equality.
\end{proof}

Before proving \Cref{lem:gammabound}, we define the IIA property that all static priority policies satisfy, and prove a crucial lemma that holds for any IIA policy.

Fix a policy $\pi$.  Throughout our discussion, we will omit $\pi$ from the notation. Recall that for $t\geq 0$, we write $\mathbf{N}(t)=(N_i(t))_{i\in\m{T}}$ for the state at time $t$ after running policy $\pi$, with $N_i(t)$ representing the number of type-$i$ agents available just before time $t$. Furthermore, $\mathcal{A}_i(t)$ is the set of agents of type $i$ that arrive before time $t$, and
\[
    \m{H}(t)=\{(i_a,t_a,d_a)\}_{a\in\m{A}(t)}
\]
is the history up to time $t$. Let $\mathbf{e}_i$ denote the $i$-th unit vector. 

Let us abuse the notation slightly and define $\mathbf{N}(t;\m{H})$ to be the state at time $t$ obtained by applying the deterministic policy $\pi$ to a deterministic history $\m{H}$. For an additional agent $a$ of type $i$, write
$\hist \cup \{a\}$ for the history with $a$ added, and let $Q_a(t;\m{H})$ be the indicator that this added agent is available at time $t$ when the policy is run on $\m{H}\cup\{a\}$.

\begin{definition}[IIA property]
\label{def:iia}
A policy satisfies \textbf{independence of irrelevant alternatives (IIA)} if, for every deterministic history $\m{H}$, every $t\geq 0$, and every added agent $a$:
\begin{equation}
  Q_a(t;\hist) = 1
  \quad\Longrightarrow\quad
  \mathbf{N}(t;\hist \cup \{a\}) = \mathbf{N}(t;\hist) + e_{i_a}.
  \label{eq:iia}
  \tag{IIA}
\end{equation}
In other words, if the added agent $a$ is still available at time $t$, the rest of the system is at exactly the state it would have reached had agent $a$ never arrived.
\end{definition}

% fix the earliest-arrival rule (or any rule under which removing an
% unselected agent does not change the selection); by footnote 8 this is without loss of generality for Lemma 3.3, since the distribution
% of N^pi does not depend on the tie-breaking rule.

Suppose we fix a deterministic rule for choosing among agents of the same type, under which removing an
unselected agent does not change the selection; for example, always choose the earliest-arriving agent of the selected type. By \Cref{footnote:tie-break} this is without loss of generality for \Cref{lem:gammabound}, since the distribution of $\mathbf{N}^\pol$ does not depend on the tie-breaking rule.  Once
we fix such a rule, we claim that any static priority policy satisfies \Cref{def:iia}.

To see this, compare two systems: one run on $\m{H}$, and the other run on $\m{H}\cup\{a\}$, with both using the same arrival and sojourn times for all agents in $\m{H}$. Up to time $t_a$, the two systems coincide.  Since $Q_a(t;\m{H})=1$, agent $a$ does not get matched upon  its arrival, so the systems continue to coincide immediately after time $t_a$, except for the addition of agent $a$.

Between time $t_a$ and $t$, the abandonments coincide in both systems. Consider an arrival after time $t_a$. Since $Q_a(t;\m{H})=1$, this arrival is not matched to $a$ in the system with $\m{H}\cup\{a\}$.  Otherwise, the static priority policy selects the same agent to match in both systems. Thus all original agents have the same status at time $t$ in the two systems, and the only difference is the additional available type-$i$ agent $a$.  This is exactly what is needed in \Cref{def:iia}.

Thus, every static priority policy satisfies the IIA property. The following lemma about policies satisfying the IIA property is the main stochastic step in the proof of \Cref{lem:gammabound}.

\begin{lemma}
\label{lemma:iia}
Suppose that a policy $\pi$ satisfies IIA.  Then, for every $i\in\m{T}$, every $t\geq 0$, and every bounded nonnegative measurable function $h:\mathbb{N}^{\m{T}}\to \Rnn$,
\begin{equation}
     \E{N_i(t) h(\mathbf{N}(t)-\mathbf{e}_i)}
    \le
    \frac{\lambda_i}{\mu_i}\left(1-e^{-\mu_i t}\right)\E{h(\mathbf{N}(t))}, 
    \label{eq:iiabound}
\end{equation}
where the term $N_i(t)h(\mathbf{N}(t)-\mathbf{e}_i)$ is defined to be zero if $N_i(t)=0$. Hence, if for the policy $\pi$ its induced $\mathbf{N}(t)$  converges in distribution to $\mathbf{N}$, where $\mathbf{N}$ is drawn from the stationary distribution (for example, if $\pi$ is a static priority policy, \Cref{lemma:stationaryctmc}),
\[
    \E{N_i h(\mathbf{N}-\mathbf{e}_i)}
    \le
    \frac{\lambda_i}{\mu_i}\E{h(\mathbf{N})}.
\]
\end{lemma}

\begin{proof} To start, notice that $N_i(t) h(\mathbf{N}(t)-\mathbf{e}_i)$ is just a sum of $h(\mathbf{N}(t)-\mathbf{e}_i)$ over the available type-$i$ agents, and furthermore, that every available type-$i$ agent $a$ at time $t$ satisfies:
\[\mathbf{N}(t)-\mathbf{e}_i = \mathbf{N}(t;\hist(t)\setminus\{a\})\] by \eqref{eq:iia} (applied with
$\hist = \hist(t)\setminus\{a\}$). Therefore:
\begin{align*}
   N_i(t) h(\mathbf{N}(t)-\mathbf{e}_i)
    = \sum_{a\in \m{A}_i(t)}
    Q_a(t;\m{H}(t)\setminus\{a\})
    h\left(\mathbf{N}(t;\m{H}(t)\setminus\{a\})\right), 
\end{align*}
Next, we take the expectation on both sides, using the Slivnyak--Mecke formula (\Cref{fact:poisson-arrival-calculus}). Let $\Phi$ be the Poisson process on $[0,\infty)^2$ of type-$i$ arrival and sojourn time pairs (\Cref{fact:productppp}), and let $\Psi$ collect the arrival and sojourn times of the agents of all other types, which are independent of $\Phi$. Conditional on $\Psi$, we apply \Cref{fact:poisson-arrival-calculus} to $\Phi$ with $g(a,\varphi)=Q_a(t;\varphi\cup\Psi)\,h(\mathbf{N}(t;\varphi\cup\Psi))$, where $\varphi\cup\Psi$ denotes the history formed by the type-$i$ agents in $\varphi$ and the other agents in $\Psi$, and then average over $\Psi$, as in the proof of \Cref{cor:poisson-arrival-count}. This gives:
\begin{align*}
    \E{N_i(t) h(\mathbf{N}(t)-\mathbf{e}_i)} =\lambda_i\int_0^t\int_0^\infty
    \E{Q_{a}(t;\m{H}(t))h(\mathbf{N}(t;\m{H}(t)))}
    \mu_i e^{-\mu_i\ell}d\ell ds, 
\end{align*}
where, in the integrand, $a$ denotes a type-$i$ agent with $t_{a}=s$, $s_{a}=\ell$, and $d_{a}=s+\ell$. $\lambda_i \mu_i e^{-\mu_i\ell}d\ell ds$ is the rate of the Poisson process associated with type-$i$ arrival and sojourn time pairs.  Now clearly, an agent can be available at $t$ only if it has not yet abandoned, so
\[Q_{a}(t;\hist(t)) \le \indic{\dtime{a} > t} = \indic{s+\ell > t}.\]
It follows that,
\begin{align*}
    \E{N_i(t) h(\mathbf{N}(t)-\mathbf{e}_i)}
    &\le
    \E{h(\mathbf{N}(t))}
    \lambda_i\mu_i\int_0^t\int_0^\infty \indic{s+\ell>t}\,e^{-\mu_i\ell}\,d\ell\,ds \\
    &=    \E{h(\mathbf{N}(t))}\lambda_i\int_0^t e^{-\mu_i(t-s)}ds
    \\&=
    \frac{\lambda_i}{\mu_i}\left(1-e^{-\mu_i t}\right)\E{h(\mathbf{N}(t))},
\end{align*}
as desired. As $t \to \infty$, $\mathbf{N}(t)$ converges in distribution to $\mathbf{N}$, while $N_i(t)$ is uniformly integrable, since it is dominated by the number of type-$i$ agents present at time $t$, a Poisson variable with mean $\tfrac{\arate{i}}{\drate{i}}\left(1-e^{-\mu_i t}\right) \leq \frac{\lambda_i}{\mu_i}$. So both sides of \eqref{eq:iiabound} converge and the inequality passes to the stationary distribution:
\[
  \E{N_i h(\mathbf{N} - \mathbf{e}_i)} \leq \tfrac{\arate{i}}{\drate{i}}\E{h(\mathbf{N})},
\]as needed.
\end{proof}

We are now ready to complete the proof of \Cref{lem:gammabound}.

\gammabound*

\begin{proof}
Fix a nonempty $S\subseteq\m{T}$ (for $S=\emptyset$, both sides of \eqref{eq:gammabound} are zero) and define $N_S \triangleq\sum_{i\in S}N_i$. Let $p_m=\PP{N_S=m},$ and apply \Cref{lemma:iia}  with
\[ h(\mathbf{N})=\indic{\textstyle\sum_{k\in S}N_k=m}.\]
For each $i\in S$, we conclude that:
\[\E{N_i\indic{N_S=m+1}}\le \frac{\lambda_i}{\mu_i}p_m,\]
so summing over $i\in S$ gives
\begin{equation}
    (m+1)p_{m+1}\le \rho(S) p_m. \label{eq:steady-state-recurrence}
\end{equation}

Let $L_S$ be a Poisson random variable with mean $\rho(S)$, and let $l_m=\PP{L_S=m}$.  Since $L_S$ satisfies \eqref{eq:steady-state-recurrence} with equality, i.e., $(m+1)l_{m+1}=\rho(S) l_m$, the likelihood ratios $q_m=\frac{p_m}{l_m}$ are nonincreasing in $m$. So $N_S$, conditioned on being positive, is stochastically dominated by $L_S$, i.e., for any $k\geq 1$,
\begin{align}
    \PP{N_S\geq  k\mid N_S>0}
    &=
    \frac{\sum_{m\geq  k}q_ml_m}{\sum_{m\geq 1}q_ml_m}
    \le
    \frac{\sum_{m\geq  k}l_m}{\sum_{m\geq 1}l_m}
    =
    \PP{L_S\geq  k\mid L_S>0}, \label{eq:tail-domination}
\end{align}
where the inequality follows because $r_{m'} \leq r_{m}$ for $m' \geq k > m$.  Summing \eqref{eq:tail-domination} over $k\geq 1$ turns this tail comparison into a comparison of conditional means:
\[\frac{\E{N_S}}{\PP{N_S>0}}=\E{N_S\mid N_S>0} \le \E{L_S\mid L_S>0}=\frac{\rho(S)}{1-e^{-\rho(S)}}\mathrel{=}\frac{1}{\gam{S}}. \]
Since $\E{N_S}=\sum_{i\in S}\nss{i}^{\pol}$, rearranging gives \eqref{eq:gammabound}. This proves the desired result. 
\end{proof}

\subsection{Omitted Proofs from \Cref{sec:primal}}
\label{apx:omittedproofs-primal}

\subsubsection{Proofs from \Cref{sec:primal-alg}}
\label{apx:proofs-primal-alg}

\plplowerbound*

\begin{proof}
Write $(\bn,\bx)=(\bn^*,\bx^*)$ for the optimal solution in the statement. Let $({\bn}^\pi, {\bx}^\pi)$ be the steady-state quantities of $\pi$ defined in \eqref{eq:nxpi}. We show that $(\bn, \bx)$ can be expressed as a convex combination of $({\bn}^\pi, {\bx}^\pi)$ and a feasible point  $(\hat \bn, \hat \bx)$ of $\plpm$. That is, there is $\beta \in (0,1)$ such that 
\begin{equation} \beta {\bx}^\pi + (1-\beta) \hat{{\bx}} = {\bx}.\label{eq:beta}\end{equation} 
By optimality of $(\bn, \bx)$ we have $\mathbf{r} \circ \bx \ge \mathbf{r} \circ {\bf \hat{x}}$, so it follows from \eqref{eq:beta} that $\mathbf{r} \circ \bx^\pi \ge \mathbf{r} \circ \bx$. Since $\pi$ makes matches only in $\m{M}$, \eqref{eq:Vpi} gives $\val{\pi}=\sum_{(i,j)\in\m{M}}r_{ij}x^\pi_{ij}=\mathbf{r}\circ\bx^\pi$, and therefore $\val{\pi}\ge \mathbf{r}\circ\bx=\plpm$.

Given $\beta \in (0,1)$ define 
    \begin{align}
        \hat{x}_{ij} &= \frac{x_{ij} - \beta x_{ij}^\pi}{1-\beta} & \quad \forall (i,j) \in \m{M}, \label{eq:hatx}\\
        \hat{n}_i & = \frac{n_i - \beta n_i^\pi}{1-\beta} & \quad \forall i \in \m{T}. \label{eq:hatn}
        %\\ \hat{n}_i &= \frac{\lambda_i - \sum_{j \in \m{T}: (i,j) \in \m{M}}\hat{x}_{ij} - \sum_{j\in\m{T} : (j,i) \in \m{M}}\hat{x}_{ji} }{\mu_i} & \quad \forall i \in \m{T},\label{eq:hatni}\\
        %\hat{\psi}_{Sj} &= \lambda_j \gamma_S \sum_{i \in S} \hat{n}_i - \sum_{i \in S} \hat{x}_{ij} & \quad \forall j \in \m{T}, S \subseteq T(M,j). \label{eq:hatpsiSj}
    \end{align}
    Then \eqref{eq:beta} holds by construction. We will show that $(\hat{{\bn}}, \hat{{\bx}})$ is feasible for all sufficiently small $\beta$. 

    We first show that $\hat{\bx}$ is nonnegative. If $x_{ij} = 0$, then by the definition of consistency, $i$ is not acceptable to $j$ according to $\pi$. Therefore $x_{ij}^\pi = 0$ and $\hat{x}_{ij} = 0$. Meanwhile, if $x_{ij} > 0$, then $\hat{x}_{ij} > 0$ for all sufficiently small $\beta$. Similarly, $\hat{\bn}\ge0$ for all sufficiently small $\beta$, because $n_i>0$ for every $i\in\m{T}$, as we show in the proof of \Cref{lemma:suitablestructure} in \Cref{apx:proofs-suitable}, which does not use \Cref{prop:plplowerbound}.
    
    Turning to \eqref{eq:balance-lp-glb}, this holds for both $(\bn, \bx)$ (because $(\bn, \bx)$ is feasible for $\plpm$) and for $(\bn^\pi, \bx^\pi)$ (by \Cref{lemma:balance:multiple}, since $x^\pi_{ij}=0$ for $(i,j)\notin\m{M}$). Because \eqref{eq:balance-lp-glb} is a linear equality constraint, it also holds for any affine combination of $(\bn, \bx)$ and $(\bn^\pi, \bx^\pi)$, that is, any combination whose coefficients sum to one, and in particular for $(\hat{{\bn}}, \hat{{\bx}})$.

    It remains to show that \eqref{eq:availability-alg} holds. We wish to show that for all $j \in \m{T}$ and all nonempty $S \subseteq \m{M}_j$, the following quantity is nonnegative:
    \begin{equation}        (1-\beta)(\lambda_j \gamma_S \sum_{i \in S} \hat{n}_i - \sum_{i \in S} \hat{x}_{ij}) = (\lambda_j \gamma_S \sum_{i \in S} n_i - \sum_{i \in S} x_{ij}) - \beta(\lambda_j\gamma_S \sum_{i \in S} n_i^\pi - \sum_{i \in S} x_{ij}^\pi). \label{eq:hatpsiSj-rewritten}
    \end{equation}
    (Note that this equality follows immediately from \eqref{eq:hatx} and \eqref{eq:hatn}.) 
    If $(\bn, \bx)$ satisfies \eqref{eq:availability-alg} strictly, then the first term on the right side of \eqref{eq:hatpsiSj-rewritten} is positive and thus the entire expression is positive for sufficiently small $\beta$. If $(\bn, \bx)$ satisfies \eqref{eq:availability-alg} with equality, then the first term on the right side of \eqref{eq:hatpsiSj-rewritten} is zero and $S \in \m{P}_j(\bn,\bx)$. Because $\pi$ is consistent with $(\bn,\bx)$, $S$ must be a prefix of $\succ_j$. Applying \Cref{lemma:balance:multiple} and \Cref{lem:gammabound}, we conclude that 
    \[\sum_{i \in S} x_{ij}^\pi = \lambda_j\PP{\sum_{i \in S} N_i^\pi > 0} \ge \lambda_j \gamma_S \sum_{i \in S} n_i^\pi,\]
    and therefore the expression in \eqref{eq:hatpsiSj-rewritten} is nonnegative. Since $\plpm$ has finitely many constraints, all of them hold for every sufficiently small $\beta$, so $(\hat\bn,\hat\bx)$ is feasible, as claimed.
% [Version 4: the paragraph below proved the last claim of the suitability lemma, not of this proposition, and used the obsolete slack variables psi; a rewritten version now ends the proof of that lemma.]
%     We now turn to the final claim of the lemma, that the condition \eqref{eq:suitability} is necessary and sufficient for a consistent $\pi$ to exists. 
%     
%  This condition is clearly necessary: if there were $i,j,S$ such that $\psi_{Sj} = 0$, $x_{ij} = 0$ and $i \in S$, then the definition of consistency would require both that $(i,j)$ matches are unacceptable (since $x_{ij} = 0$) and that $S$ is a prefix of $\succ_j$.
% To see that it is also sufficient, apply the conclusion that the sets in $\m{P}_j(\bn,\bx)$ are nested. Therefore, it is possible to create a single ordering $\succ_j$ that has all of these sets as prefixes. Furthermore, if \eqref{eq:suitability} holds, then this does not require that $\succ_j$ rank as acceptable any $i$ with $x_{ij} = 0$. Therefore, it is possible to construct a static priority policy $\pi$ consistent with $(\bn,\bx)$.
\end{proof}

\subsubsection{Proofs from \Cref{sec:suitable}}
\label{apx:proofs-suitable}

\suitablestructure*

We first recall the definition of submodular and supermodular set functions.

\begin{definition}[Submodularity/supermodularity] A set function $\varphi : 2^\m{T} \to \mathbb{R}$ is \textbf{submodular} if for all sets $S$ and $S'$ 
\begin{equation} \varphi(S) +  \varphi(S') \ge \varphi(S \cap S') + \varphi(S \cup S'). \label{eq:submodular} \end{equation}
It is {\bf strictly submodular} if \eqref{eq:submodular} holds strictly whenever $S$ and $S'$ are incomparable (meaning that $S \nsubseteq S'$ and $S' \nsubseteq S$). 
The function $\varphi$ is {\bf supermodular} if $-\varphi$ is submodular, and {\bf strictly supermodular} if $-\varphi$ is strictly submodular.
\end{definition}

\begin{proofof}{\Cref{lemma:suitablestructure}}
Fix $(\bn,\bx)$ and $j\in\m{T}$, and define $x_j(S)\triangleq\sum_{i\in S}x_{ij}$ and
\[f(S)\triangleq\frac{1}{\gamma_S}=\frac{\rho(S)}{1-e^{-\rho(S)}},\]
with $f(\emptyset)=1$. Then
\[\Psi_j(S)=\lambda_j\sum_{i\in S}n_i-x_j(S)f(S).\]
We wish to show that $\Psi_j$ is submodular. Since $S\mapsto \lambda_j\sum_{i\in S}n_i$ is modular, it suffices to show that $S\mapsto x_j(S)f(S)$ is supermodular.

We first establish that $f$ is strictly increasing and strictly supermodular. Let $g(y)\triangleq y/(1-e^{-y})$, with $g(0)=1$, so that $f(S) = g(\rho(S))$. The function $g$ is strictly increasing and strictly convex on $[0,\infty)$.\footnote{For $y>0$, $g'(y)=e^y(e^y-1-y)/(e^y-1)^2>0$. Also, $g''(y)=e^y h(y)/(e^y-1)^3$, where $h(y)=y(e^y+1)-2e^y+2$. Since $h(0)=h'(0)=0$ and $h''(y)=ye^y>0$ for $y>0$, we have $h(y)>0$ for $y>0$.} Because $\rho$ is additive and positive on every nonempty set, $f(S)=g(\rho(S))$ is strictly increasing. To see that $f$ is strictly supermodular, suppose that $S,S'$ are incomparable and let $a=\rho(S\cap S')$, $b=\rho(S\setminus S')$, $c=\rho(S'\setminus S)$. Then $b,c>0$, and
\[f(S\cup S')+f(S\cap S')-f(S)-f(S')=[g(a+b+c)-g(a+c)]-[g(a+b)-g(a)]>0,\]
where the strict inequality follows from strict convexity of $g$, since increments of length $b>0$ are strictly increasing in the starting point and $c>0$. 

% [Version 5, fifth round: the substitutions, previously in a footnote, are now stated in the text, and the column-aligned array is replaced by a standard multi-line display.]
Now fix arbitrary $S,S'\subseteq\m{T}$. Since $x_j$ is additive,
\[
\begin{aligned}
x_j(S\cup S')&=x_j(S\cap S')+x_j(S'\setminus S)+x_j(S\setminus S'),\\
x_j(S)&=x_j(S\cap S')+x_j(S\setminus S'),\\
x_j(S')&=x_j(S\cap S')+x_j(S'\setminus S).
\end{aligned}
\]
Substituting these expressions on the left side below and collecting the terms that multiply ${x_j(S\cap S')}$, ${x_j(S'\setminus S)}$, and ${x_j(S\setminus S')}$ gives
\begin{equation}\label{eq:decomposition}
\begin{split}
&x_j(S\cup S')\,f(S\cup S')+x_j(S\cap S')\,f(S\cap S')-x_j(S)\,f(S)-x_j(S')\,f(S')\\
&\qquad=x_j(S\cap S')\,\big[f(S\cup S')+f(S\cap S')-f(S)-f(S')\big]\\
&\qquad\quad+x_j(S'\setminus S)\,\big[f(S\cup S')-f(S')\big]\\
&\qquad\quad+x_j(S\setminus S')\,\big[f(S\cup S')-f(S)\big].
\end{split}
\end{equation}
Because $f$ is supermodular and  increasing, each of the three quantities in square brackets above is nonnegative, implying that the entire expression is nonnegative. Thus, the function $S \mapsto x_j(S)f(S)$ is supermodular, and $\Psi_j$ is submodular.

We now use strict supermodularity of $f$ and submodularity of $\Psi_j$ to prove the second claim. Let $S,S'\in\m{P}_j(\bn,\bx)$, which is equivalent to $\Psi_j(S)=\Psi_j(S')=0$. Since $\Psi_j$ is nonnegative by feasibility of $(\bn,\bx)$, submodularity of $\Psi_j$ gives
\[0=\Psi_j(S)+\Psi_j(S')\geq\Psi_j(S\cap S')+\Psi_j(S\cup S')\geq0.\]
Hence
\[\Psi_j(S\cap S')=\Psi_j(S\cup S')=0.\]
It follows that the expression in \eqref{eq:decomposition} must be equal to zero, since, by modularity of $S\mapsto\lambda_j\sum_{i\in S}n_i$, it equals $\Psi_j(S)+\Psi_j(S')-\Psi_j(S\cap S')-\Psi_j(S\cup S')$.

We now show that either $S \subseteq S'$ or $S' \subseteq S$. Seeking a contradiction, suppose that $S$ and $S'$ were incomparable. Strict supermodularity and strict monotonicity of $f$ would make all three bracketed quantities on the right side of \eqref{eq:decomposition} strictly positive. Because the expression in \eqref{eq:decomposition} equals zero, this implies that $x_j(S \cap S'), x_j(S' \setminus S)$, and $x_j(S \setminus S')$ must all be zero, and thus that $x_{j}(S \cup S') = 0$. But $\Psi_j(S \cup S') = 0$ implies 
\[x_j(S\cup S')f(S\cup S')=\lambda_j\sum_{i\in S\cup S'}n_i.\]
To complete our proof we establish that $n_i > 0$ for every $i\in\m{T}$, contradicting the conclusion that $x_{j}(S \cup S') = 0$.

Applying \eqref{eq:availability-alg} with $S=\{j : (j,i) \in \m{M}\} = \m{M}_i$ and using $\sum_{j\in S}n_j\leq\rho(S)$, which follows from \eqref{eq:balance-lp-glb}, gives
\[\sum_{j \in \m{M}_i}x_{ji}\leq\lambda_i\gamma_{\m{M}_i}\sum_{j\in\m{M}_i}n_j\leq\lambda_i(1-e^{-\rho(\m{M}_i)})<\lambda_i\]
(if $\m{M}_i=\emptyset$, the left side is zero).
Because \eqref{eq:balance-lp-glb} holds, we must have 

\[n_i \mu_i + \sum_{j : (i,j) \in \m{M}} x_{ij} > 0.\]
If $n_i = 0$, then applying \eqref{eq:availability-alg} with $S=\{i\}$ requires that $x_{ij}=0$ for every term in the sum. Thus, we must have $n_i > 0$.

% \rn{This pragraph is somewhat, not fully redundent; happy to explain}

We now turn to the final claim of \Cref{lemma:suitablestructure}, that \eqref{eq:suitability} is necessary and sufficient for $(\bn,\bx)$ to be suitable. The condition is necessary: if there were $j\in\m{T}$, $S\in\m{P}_j(\bn,\bx)$, and $i\in S$ with $x_{ij}=0$, then the definition of consistency would require both that $(i,j)$ matches are unacceptable (since $x_{ij}=0$) and that $S$ is a prefix of $\succ_j$ (so that $i$ is acceptable). To see that it is also sufficient, recall that the sets in $\m{P}_j(\bn,\bx)$ are nested. Therefore, for each $j\in\m{T}$, there is an order $\succ_j$ that has all of these sets as prefixes and ranks as acceptable only the types in their union. If \eqref{eq:suitability} holds, then every type $i$ that $\succ_j$ ranks as acceptable belongs to a set in $\m{P}_j(\bn,\bx)$, so $(i,j)\in\m{M}$ and $x_{ij}>0$. The resulting static priority policy is therefore consistent with $(\bn,\bx)$.
\end{proofof}

\polytime*
\begin{proof}
The program $\plpm$ has $|\m{T}|+|\m{M}|$ variables, the $|\m{T}|$ balance constraints \eqref{eq:balance-lp-glb}, nonnegativity constraints, and one availability constraint \eqref{eq:availability-alg} for every $j\in\m{T}$ and nonempty $S\subseteq\m{M}_j$; only the last family is exponentially large. We show how to construct an exact polynomial-time separation oracle for it. Given a point $(\hat\bn,\hat\bx)$, we first check the balance and nonnegativity constraints explicitly. If they hold, then for each $j\in\m{T}$ we consider the set function
\[
 \Psi_j(S)=\frac{1}{\gamma_S}\Big(\lambda_j\gamma_S\sum_{i\in S}\hat n_i-\sum_{i\in S}\hat x_{ij}\Big)
 =\lambda_j\sum_{i\in S}\hat n_i-\frac{1}{\gamma_S}\sum_{i\in S}\hat x_{ij},
 \qquad S\subseteq\m{M}_j,
\]
with $\Psi_j(\emptyset)=0$, which is submodular by \Cref{lemma:suitablestructure} (the submodularity argument there uses only $\hat\bx\ge0$). Since $\gamma_S>0$, the availability constraint of $(S,j)$ holds if and only if $\Psi_j(S)\ge0$, so all availability constraints of $j$ hold if and only if $\min_{S\subseteq\m{M}_j}\Psi_j(S)\ge0$. A submodular function can be minimized in strongly polynomial time \citep{schrijver2000combinatorial}, and each evaluation of $\Psi_j$ takes polynomial time. If the minimum is negative, the minimizing set $S$ identifies a violated constraint, whose defining inequality is a separating hyperplane; otherwise $(\hat\bn,\hat\bx)$ is feasible. By the equivalence of separation and optimization \citep[Theorem~6.4.9]{GrotschelLovaszSchrijver1988}, $\plpm$ can therefore be solved in polynomial time.
\end{proof}

\subsection{Omitted Proofs from \Cref{sec:values}}
\label{apx:omittedproofs-dual}
\subsubsection{Proofs from \Cref{sec:value-certificate}}
\label{apx:value-certificate-proofs}
\replacementhalf*
\begin{proof}
The dual of \(\offrbrel\) is
\begin{equation}\label{eq:replacement-omn-dual}
\begin{aligned}
 \textup{DLP}^{\textup{OMN}}
 \mathrel{\triangleq}\min_{\mathbf{v},\mathbf{z}\ge0}\quad
 &\sum_i\lambda_i v_i
   +\sum_{j\in\mathcal T}\sum_{\emptyset\ne S\subseteq\mathcal T}
           \lambda_j(1-e^{-\rho(S)})z_{Sj}\\
 \textrm{s.t.}\quad
 &v_i+v_j+\sum_{\emptyset\ne S\subseteq\mathcal T}z_{Sj}\mathbf1(i\in S)
       \ge r_{ij}\qquad\forall i,j\in\mathcal T.
\end{aligned}
\end{equation}
Extend the prefix weights by setting \(z_{Sj}=0\) for
\(S\notin\mathcal P_j(\pi)\). Constraint
\eqref{eq:replacement-reverse-reward} then makes \((\mathbf{v},\mathbf{z})\) feasible for
\(\textup{DLP}^{\textup{OMN}}\). By weak duality,
\begin{equation}\label{eq:replacement-offline-dual}
 \offrbrel
 \le\sum_i\lambda_i v_i
     +\sum_j\sum_{S\in\mathcal P_j(\pi)}\lambda_j(1-e^{-\rho(S)})z_{Sj}.
\end{equation}
Multiply \eqref{eq:replacement-reverse-balance} by \(\lambda_i/\mu_i\), sum over
all types, and use \(\gamma_S\rho(S)=1-e^{-\rho(S)}\):
\begin{align*}
 \sum_i\lambda_i v_i
 &\ge\sum_i\frac{\lambda_i}{\mu_i}
          \sum_j\sum_{S\in\mathcal P_j(\pi)}
               \lambda_j\gamma_S z_{Sj}\mathbf1(i\in S)\\
 &=\sum_j\sum_{S\in\mathcal P_j(\pi)}\lambda_j\gamma_S z_{Sj}
           \sum_{i\in S}\frac{\lambda_i}{\mu_i}\\
 &=\sum_j\sum_{S\in\mathcal P_j(\pi)}\lambda_j(1-e^{-\rho(S)})z_{Sj}.
\end{align*}
Substitution into \eqref{eq:replacement-offline-dual} gives
\( \offrbrel\le2\sum_i\lambda_i v_i\).
\end{proof}

\subsubsection{Proofs from \Cref{sec:fixed-point-dual}}
\label{apx:tight-value-proofs}
Before turning to the remaining proofs, we record a representation of the weights $Z^{\m{M}}_{Sj}(\mathbf{v})$ of \eqref{eq:replacement-Z} that is used repeatedly. Fix $\m{M}\subseteq\TT$, a vector $\mathbf{v}$, and an arriving type $j$, and for $y>0$ let
\begin{equation}\label{eq:score-prefix}
 S_j(y)\triangleq\{i\in\m{M}_j : u_{ij}(\mathbf{v})\ge y\}
\end{equation}
be the set of types whose residual score for an arriving type-$j$ agent is at least $y$; we call these sets the \emph{positive-score prefixes} of $j$. Since $S_j(y)=S$ exactly when $\max_{i\in\m{M}_j\setminus S}u_{ij}(\mathbf{v})<y\le\min_{i\in S}u_{ij}(\mathbf{v})$, the definition \eqref{eq:replacement-Z} gives
\begin{equation}\label{eq:Z-integral}
 Z^{\m{M}}_{Sj}(\mathbf{v})=\int_0^\infty \mathbf{1}(S=S_j(y))\,dy
 \qquad\text{for every nonempty } S\subseteq\m{M}_j.
\end{equation}
Thus $Z^{\m{M}}_{Sj}(\mathbf{v})>0$ only if $S$ is a positive-score prefix, i.e., $S=S_j(y)$ for some $y>0$; summing \eqref{eq:Z-integral} over the sets containing a type $i\in\m{M}_j$ gives the telescoping identity \eqref{eq:telescope-M}; and multiplying by $\gamma_S\mathbf{1}(i\in S)$ before summing gives, for every $i\in\m{M}_j$,
\begin{equation}\label{eq:gamma-integral}
 \sum_{\emptyset\ne S\subseteq\m{M}_j}\gamma_S Z^{\m{M}}_{Sj}(\mathbf{v})\,\mathbf{1}(i\in S)
 =\int_0^{u_{ij}(\mathbf{v})}\gamma_{S_j(y)}\,dy ,
\end{equation}
where the integrand is well defined because $i\in S_j(y)$ for $y\le u_{ij}(\mathbf{v})$. Finally, we recall that $\gamma_S=1/g(\rho(S))$, where $g(y)=y/(1-e^{-y})$ is the strictly increasing function from the proof of \Cref{lemma:suitablestructure}, and $\rho$ is a weakly increasing set function. As a result, $\gamma_S$ is a weakly decreasing set function, that is,  $\gamma_S\ge\gamma_{S'}$ whenever $\emptyset\ne S\subseteq S'$.

The proof of \Cref{prop:fixed-point} relies on two properties of $\mathbf{F}^{\m{M}}$, which we establish first. Throughout, we write $\rho_i=\lambda_i/\mu_i$, $\Lambda=\sum_{i\in\types}\lambda_i$, $\lambda_{\min}=\min_{i\in\types}\lambda_i$, and $\mu_{\min}=\min_{i\in\types}\mu_i$, and we use the weighted inner product and norm
\[
 \langle\mathbf{a},\mathbf{b}\rangle_{\boldsymbol{\lambda}}=\sum_{i\in\types}\lambda_ia_ib_i,
 \qquad
 \|\mathbf{a}\|_{\boldsymbol{\lambda}}=\langle\mathbf{a},\mathbf{a}\rangle_{\boldsymbol{\lambda}}^{1/2},
\]
which satisfy $\sqrt{\lambda_{\min}}\,\|\mathbf{a}\|_\infty\le\|\mathbf{a}\|_{\boldsymbol{\lambda}}\le\sqrt{\Lambda}\,\|\mathbf{a}\|_\infty$ for every $\mathbf{a}\in\reals^{\types}$.

\begin{lemma}\label{lem:F-one-sided}
Fix $\m{M}\subseteq\TT$, write $\mathbf{F}=\mathbf{F}^{\m{M}}$, and let $L=1+2\Lambda^2/(\mu_{\min}\lambda_{\min})$. For all $\mathbf{v},\mathbf{w}\in\reals^{\types}$,
\begin{enumerate}[(i)]
\item $\langle\mathbf{F}(\mathbf{v})-\mathbf{F}(\mathbf{w}),\,\mathbf{v}-\mathbf{w}\rangle_{\boldsymbol{\lambda}}\le\tfrac12\|\mathbf{v}-\mathbf{w}\|_{\boldsymbol{\lambda}}^2$;
\item $\|\mathbf{F}(\mathbf{v})-\mathbf{F}(\mathbf{w})\|_\infty\le(2\Lambda/\mu_{\min})\,\|\mathbf{v}-\mathbf{w}\|_\infty$;
\item $\|\mathbf{F}(\mathbf{v})-\mathbf{F}(\mathbf{w})\|_{\boldsymbol{\lambda}}\le L\,\|\mathbf{v}-\mathbf{w}\|_{\boldsymbol{\lambda}}$.
\end{enumerate}
\end{lemma}
\begin{proof}
Fix $\mathbf{v},\mathbf{w}\in\reals^{\types}$, let $\boldsymbol{\delta}=\mathbf{v}-\mathbf{w}$, and let $\mathbf{v}(t)=\mathbf{w}+t\boldsymbol{\delta}$ for $t\in[0,1]$. For each $j\in\types$ and $i\in\m{M}_j$, define
\[
 \Phi_{ij}(t)=\sum_{\emptyset\ne S\subseteq\m{M}_j}\gamma_S\,Z^{\m{M}}_{Sj}(\mathbf{v}(t))\,\mathbf{1}(i\in S),
 \qquad\text{so that}\qquad
 F_i(\mathbf{v}(t))=\frac{1}{\mu_i}\sum_{j:\,i\in\m{M}_j}\lambda_j\Phi_{ij}(t)
\]
by \eqref{eq:replacement-F}. Each $\Phi_{ij}$ is continuous, since each $Z^{\m{M}}_{Sj}$ is a continuous function of $\mathbf{v}$.

\emph{Structure of $\Phi_{ij}$.} Each residual score $u_{ij}(\mathbf{v}(t))$ is the positive part of an affine function of $t$. Hence, $[0,1]$ can be divided into finitely many intervals such that, in the interior of each interval and for each arriving type $j$, neither the set of types $i\in\m{M}_j$ with a positive score $u_{ij}(\mathbf{v}(t))$ nor the order of these scores changes. (It suffices to cut $[0,1]$ at the points where one of the affine functions $t\mapsto r_{ij}-v_i(t)-v_j(t)$ that is not identically zero vanishes, and at the points where two of these functions that do not coincide are equal.) Fix such an interval and an arriving type $j$. Writing $u_{ij}$ for $u_{ij}(\mathbf{v}(t))$, list the types with positive scores as $i_1,\ldots,i_m$, so that $u_{i_1j}\ge\cdots\ge u_{i_mj}>0$ throughout the interval, and let $S_k=\{i_1,\ldots,i_k\}$. Recall that a set carries positive weight only if its members have positive scores that are strictly larger than all scores outside it. Hence, by \eqref{eq:replacement-Z}, $Z^{\m{M}}_{S_kj}=u_{i_kj}-u_{i_{k+1}j}$ for $k=1,\ldots,m$, where $u_{i_{m+1}j}\triangleq 0$, and $Z^{\m{M}}_{Sj}=0$ for every other nonempty $S\subseteq\m{M}_j$. Consequently, $\Phi_{ij}=0$ on the interval for every type $i\in\m{M}_j$ with a zero score, while, for $k=1,\ldots,m$,
\[
 \Phi_{i_kj}
 =\sum_{\ell=k}^{m}\gamma_{S_\ell}\big(u_{i_\ell j}-u_{i_{\ell+1}j}\big)
 =\gamma_{S_k}u_{i_kj}-\sum_{\ell=k+1}^{m}\beta_\ell\,u_{i_\ell j},
 \qquad\text{where }\beta_\ell=\gamma_{S_{\ell-1}}-\gamma_{S_\ell}\ \ (\ell=2,\ldots,m).
\]
Here, $\beta_\ell\ge0$ because $\gamma$ is a decreasing set function. On the interval, $u_{i_\ell j}=r_{i_\ell j}-v_{i_\ell}(t)-v_j(t)$ has derivative $-(\delta_{i_\ell}+\delta_j)$, also when $i_\ell=j$. Differentiating and using $\sum_{\ell=k+1}^m\beta_\ell=\gamma_{S_k}-\gamma_{S_m}$, we obtain
\begin{equation}\label{eq:Phi-derivative}
 \Phi'_{i_kj}
 =-\gamma_{S_k}\delta_{i_k}+\sum_{\ell=k+1}^{m}\beta_\ell\,\delta_{i_\ell}-\gamma_{S_m}\delta_j ,
 \qquad k=1,\ldots,m.
\end{equation}
In particular, each $\Phi_{ij}$ is continuous and piecewise linear, so $\Phi_{ij}(1)-\Phi_{ij}(0)=\int_0^1\Phi'_{ij}(t)\,dt$.

\emph{Part (i).} By the formula for $F_i(\mathbf{v}(t))$,
\[
 \langle\mathbf{F}(\mathbf{v})-\mathbf{F}(\mathbf{w}),\boldsymbol{\delta}\rangle_{\boldsymbol{\lambda}}
 =\sum_{i\in\types}\lambda_i\delta_i\big(F_i(\mathbf{v}(1))-F_i(\mathbf{v}(0))\big)
 =\sum_{j\in\types}\lambda_j\int_0^1\sum_{i\in\m{M}_j}\rho_i\,\delta_i\,\Phi'_{ij}(t)\,dt .
\]
Since $\tfrac12\|\boldsymbol{\delta}\|_{\boldsymbol{\lambda}}^2=\sum_{j\in\types}\lambda_j\cdot\tfrac12\delta_j^2$, it suffices to show that $\sum_{i\in\m{M}_j}\rho_i\delta_i\Phi'_{ij}\le\tfrac12\delta_j^2$ for each $j$ in the interior of each interval above. Only the types $i_1,\ldots,i_m$ contribute to this sum, and if $m=0$, there is nothing to prove. Multiplying \eqref{eq:Phi-derivative} by $2\delta_{i_k}$ and applying $2ab\le a^2+b^2$ to the products $\delta_{i_k}\delta_{i_\ell}$ and $-\delta_{i_k}\delta_j$, whose coefficients $\beta_\ell$ and $\gamma_{S_m}$ are nonnegative, gives
\[
 2\delta_{i_k}\Phi'_{i_kj}
 \le-2\gamma_{S_k}\delta_{i_k}^2+\sum_{\ell=k+1}^{m}\beta_\ell\big(\delta_{i_k}^2+\delta_{i_\ell}^2\big)+\gamma_{S_m}\big(\delta_{i_k}^2+\delta_j^2\big)
 =-\gamma_{S_k}\delta_{i_k}^2+\sum_{\ell=k+1}^{m}\beta_\ell\,\delta_{i_\ell}^2+\gamma_{S_m}\delta_j^2 ,
\]
where the equality again uses $\sum_{\ell=k+1}^m\beta_\ell=\gamma_{S_k}-\gamma_{S_m}$. We now multiply by $\rho_{i_k}$ and sum over $k=1,\ldots,m$. Recall that $\rho(S)\gamma_S=1-e^{-\rho(S)}$. The total coefficient of $\delta_{i_1}^2$ is $-\rho_{i_1}\gamma_{S_1}=e^{-\rho(S_1)}-1\le0$, and, for $\ell=2,\ldots,m$, the total coefficient of $\delta_{i_\ell}^2$ is
\[
 -\rho_{i_\ell}\gamma_{S_\ell}+\beta_\ell\sum_{k=1}^{\ell-1}\rho_{i_k}
 =\rho(S_{\ell-1})\gamma_{S_{\ell-1}}-\rho(S_\ell)\gamma_{S_\ell}
 =e^{-\rho(S_\ell)}-e^{-\rho(S_{\ell-1})}\le0 .
\]
Dropping these nonpositive terms, which is valid also if one of the types $i_\ell$ is $j$ itself, yields
\[
 2\sum_{k=1}^{m}\rho_{i_k}\delta_{i_k}\Phi'_{i_kj}\le\rho(S_m)\gamma_{S_m}\delta_j^2=\big(1-e^{-\rho(S_m)}\big)\delta_j^2\le\delta_j^2 ,
\]
as needed.

\emph{Part (ii).} By \eqref{eq:Phi-derivative} and $\beta_\ell\ge0$, in the interior of each interval,
\[
 |\Phi'_{i_kj}|\le\Big(\gamma_{S_k}+\sum_{\ell=k+1}^{m}\beta_\ell+\gamma_{S_m}\Big)\|\boldsymbol{\delta}\|_\infty=2\gamma_{S_k}\|\boldsymbol{\delta}\|_\infty\le2\|\boldsymbol{\delta}\|_\infty ,
\]
while $\Phi'_{ij}=0$ for the types $i$ with a zero score. Hence, $|\Phi_{ij}(1)-\Phi_{ij}(0)|\le2\|\boldsymbol{\delta}\|_\infty$ for all $j$ and $i\in\m{M}_j$, and therefore
\[
 |F_i(\mathbf{v})-F_i(\mathbf{w})|\le\frac{1}{\mu_i}\sum_{j:\,i\in\m{M}_j}\lambda_j\cdot2\|\boldsymbol{\delta}\|_\infty\le\frac{2\Lambda}{\mu_{\min}}\|\boldsymbol{\delta}\|_\infty
 \qquad\forall i\in\types.
\]

\emph{Part (iii).} By the norm inequalities above and part (ii),
\[
 \|\mathbf{F}(\mathbf{v})-\mathbf{F}(\mathbf{w})\|_{\boldsymbol{\lambda}}
 \le\sqrt{\Lambda}\,\frac{2\Lambda}{\mu_{\min}}\|\boldsymbol{\delta}\|_\infty
 \le\frac{2\Lambda}{\mu_{\min}}\sqrt{\frac{\Lambda}{\lambda_{\min}}}\,\|\boldsymbol{\delta}\|_{\boldsymbol{\lambda}}
 \le\frac{2\Lambda^2}{\mu_{\min}\lambda_{\min}}\|\boldsymbol{\delta}\|_{\boldsymbol{\lambda}}
 \le L\,\|\boldsymbol{\delta}\|_{\boldsymbol{\lambda}} ,
\]
where the third inequality uses $\Lambda/\lambda_{\min}\ge1$.
\end{proof}

\begin{lemma}\label{lem:damped-contraction}
Fix $\m{M}\subseteq\TT$, let $L$ be as in \Cref{lem:F-one-sided}, let $\theta=1/(2L^2)$, and define the damped map $\mathbf{F}_\theta(\mathbf{v})=(1-\theta)\mathbf{v}+\theta\mathbf{F}^{\m{M}}(\mathbf{v})$. Then $\mathbf{F}_\theta(\mathbf{v})=\mathbf{v}$ if and only if $\mathbf{F}^{\m{M}}(\mathbf{v})=\mathbf{v}$. Moreover, for all $\mathbf{v},\mathbf{w}\in\reals^{\types}$,
\[
 \|\mathbf{F}_\theta(\mathbf{v})-\mathbf{F}_\theta(\mathbf{w})\|_{\boldsymbol{\lambda}}\le \kappa\,\|\mathbf{v}-\mathbf{w}\|_{\boldsymbol{\lambda}},
 \qquad\text{where } \kappa=1-\frac{1}{8L^2}<1 .
\]
\end{lemma}
\begin{proof}
The first claim holds because $\mathbf{F}_\theta(\mathbf{v})-\mathbf{v}=\theta\big(\mathbf{F}^{\m{M}}(\mathbf{v})-\mathbf{v}\big)$ and $\theta>0$. For the second claim, write $\boldsymbol{\delta}=\mathbf{v}-\mathbf{w}$ and $\boldsymbol{\Delta}=\mathbf{F}^{\m{M}}(\mathbf{v})-\mathbf{F}^{\m{M}}(\mathbf{w})$. Since $L\ge1$, we have $\theta\le\half$, so $1-\theta\ge0$, and parts (i) and (iii) of \Cref{lem:F-one-sided} give
\[
 \|\mathbf{F}_\theta(\mathbf{v})-\mathbf{F}_\theta(\mathbf{w})\|_{\boldsymbol{\lambda}}^2
 =(1-\theta)^2\|\boldsymbol{\delta}\|_{\boldsymbol{\lambda}}^2+2\theta(1-\theta)\langle\boldsymbol{\Delta},\boldsymbol{\delta}\rangle_{\boldsymbol{\lambda}}+\theta^2\|\boldsymbol{\Delta}\|_{\boldsymbol{\lambda}}^2
 \le\big((1-\theta)^2+\theta(1-\theta)+\theta^2L^2\big)\|\boldsymbol{\delta}\|_{\boldsymbol{\lambda}}^2 .
\]
The factor in parentheses equals $1-\theta+\theta^2L^2=1-1/(4L^2)$, and the claim follows from $\sqrt{1-x}\le1-x/2$ for $x\in[0,1]$.
\end{proof}

\fixedpointunique*
\begin{proof} Fix $\mathcal{M}\subseteq \TT$, and write $\mathbf{F}=\mathbf{F}^{\m{M}}$. We prove the existence, uniqueness, and polynomial-time computation of the fixed point of $\mathbf{F}^\mathcal{M}$.

\medskip
\noindent\underline{\emph{Existence and uniqueness}}: By \Cref{lem:damped-contraction}, the damped map $\mathbf{F}_\theta$ is a contraction with respect to $\|\cdot\|_{\boldsymbol{\lambda}}$, and $\reals^{\types}$ is complete under this norm. Banach's fixed-point theorem~\citep{banach1922operations} therefore shows that $\mathbf{F}_\theta$ has a unique fixed point. Since $\mathbf{F}_\theta$ and $\mathbf{F}$ have the same fixed points (\Cref{lem:damped-contraction}), $\mathbf{F}$ has a unique fixed point, which we denote by $\mathbf{v}^{\m{M}}$. Moreover, $\mathbf{F}^{\m{M}}(\mathbf{v})\ge0$ for every $\mathbf{v}\in\reals^{\types}$, so $\mathbf{v}^{\m{M}}=\mathbf{F}(\mathbf{v}^{\m{M}})\ge0$.
% [Version 2: existence and uniqueness now both follow from the contraction property (Lemma lem:damped-contraction and Banach's theorem). The Brouwer argument for existence, which followed the bound below, and the direct uniqueness argument of the previous draft were removed.]

\medskip
\noindent\underline{\emph{Computation}}: We first locate the fixed point in a box. Since $0\le\gamma_S\le1$, the telescoping identity \eqref{eq:telescope-M} gives, for every $\mathbf{v}\ge0$,
\[
 F^{\m{M}}_i(\mathbf{v})\le\frac{1}{\mu_i}\sum_{j:\,i\in\m{M}_j}\lambda_j\sum_{\emptyset\ne S\subseteq\m{M}_j}Z^{\m{M}}_{Sj}(\mathbf{v})\,\mathbf{1}(i\in S)
 =\frac{1}{\mu_i}\sum_{j:\,i\in\m{M}_j}\lambda_ju_{ij}(\mathbf{v})
 \le C_i\triangleq\frac{1}{\mu_i}\sum_{j\in\m{T}}\lambda_jr_{ij}.
\]
Therefore $\mathbf{F}^{\m{M}}$ maps the nonempty compact convex set $\prod_{i\in\m{T}}[0,C_i]$ into itself, and so does $\mathbf{F}_\theta$, because $\mathbf{F}_\theta(\mathbf{v})$ is a convex combination of $\mathbf{v}$ and $\mathbf{F}(\mathbf{v})$. The same holds for the box $B=[0,U]^{\types}$ with $U=1+\max_{i\in\types}C_i$, which has nonempty interior, as the algorithm below requires. Let $D$ be the diagonal matrix with diagonal entries $\sqrt{\lambda_i}$. The change of variables $\mathbf{y}=D\mathbf{v}$ turns $\|\cdot\|_{\boldsymbol{\lambda}}$ into the Euclidean norm, so the map $\mathbf{y}\mapsto D\,\mathbf{F}_\theta(D^{-1}\mathbf{y})$ is a contraction with factor $\kappa$ (as in \Cref{lem:damped-contraction}) with respect to the Euclidean norm, and it maps the box $DB$ into itself. The algorithm of \citet[Theorem~3.4 with $p=2$]{daskalakis2026polynomial}, or the ellipsoid algorithm of \citet{sikorski1993ellipsoid}, therefore computes, for every $\eta>0$, a vector $\tilde{\mathbf{v}}\in B$ with $\|\mathbf{F}_\theta(\tilde{\mathbf{v}})-\tilde{\mathbf{v}}\|_{\boldsymbol{\lambda}}\le\eta$, using a number of arithmetic operations and evaluations of $\mathbf{F}$ that is polynomial in the input size and $\log(1/\eta)$. Since $\mathbf{F}_\theta(\mathbf{v}^{\m{M}})=\mathbf{v}^{\m{M}}$, the triangle inequality and \Cref{lem:damped-contraction} give
\[
 \|\tilde{\mathbf{v}}-\mathbf{v}^{\m{M}}\|_{\boldsymbol{\lambda}}
 \le\|\tilde{\mathbf{v}}-\mathbf{F}_\theta(\tilde{\mathbf{v}})\|_{\boldsymbol{\lambda}}+\|\mathbf{F}_\theta(\tilde{\mathbf{v}})-\mathbf{F}_\theta(\mathbf{v}^{\m{M}})\|_{\boldsymbol{\lambda}}
 \le\eta+\kappa\,\|\tilde{\mathbf{v}}-\mathbf{v}^{\m{M}}\|_{\boldsymbol{\lambda}} ,
\]
so $\|\tilde{\mathbf{v}}-\mathbf{v}^{\m{M}}\|_{\boldsymbol{\lambda}}\le\eta/(1-\kappa)=8L^2\eta$. Since $\|\mathbf{a}\|_\infty\le\|\mathbf{a}\|_{\boldsymbol{\lambda}}/\sqrt{\lambda_{\min}}$ for every $\mathbf{a}$, an appropriate choice of $\eta$ yields, for every $\varepsilon\in(0,1)$, a vector $\tilde{\mathbf{v}}$ with $\|\tilde{\mathbf{v}}-\mathbf{v}^{\m{M}}\|_\infty\le\varepsilon$ in time polynomial in the input size and $\log(1/\varepsilon)$.

It remains to show that $\mathbf{F}$ can be evaluated efficiently and that a sufficiently accurate approximation determines $\mathbf{v}^{\m{M}}$ exactly. For the first claim, fix $\mathbf{v}$ and an arriving type $j$, and sort the types $i\in\m{M}_j$ with positive scores as $u_{i_1j}(\mathbf{v})\ge\cdots\ge u_{i_mj}(\mathbf{v})>0$. As in the proof of \Cref{lem:F-one-sided}, the only sets that can carry positive weight are $S_k=\{i_1,\ldots,i_k\}$, $k=1,\ldots,m$, with $Z^{\m{M}}_{S_kj}(\mathbf{v})=u_{i_kj}(\mathbf{v})-u_{i_{k+1}j}(\mathbf{v})$, where $u_{i_{m+1}j}(\mathbf{v})\triangleq 0$. Hence, the contribution of arriving type $j$ to $F_{i_k}(\mathbf{v})$ is $(\lambda_j/\mu_{i_k})\sum_{\ell=k}^{m}\gamma_{S_\ell}Z^{\m{M}}_{S_\ell j}(\mathbf{v})$. All of these contributions can be computed from the sorted scores and the partial sums $\rho(S_\ell)$ using $O(|\types|\log|\types|)$ arithmetic operations and $O(|\types|)$ evaluations of the exponential function. Summing over arriving types, $\mathbf{F}(\mathbf{v})$ can be evaluated with $O(|\types|^2\log|\types|)$ operations.

\medskip
\noindent\underline{\emph{Exact computation}}: The formula above shows that $\mathbf{F}$ is piecewise affine. Consider the finitely many affine functions $\mathbf{v}\mapsto r_{ij}-v_i-v_j$ for $(i,j)\in\m{M}$, together with their differences $\mathbf{v}\mapsto(r_{ij}-v_i-v_j)-(r_{i'j}-v_{i'}-v_j)$ for distinct $i,i'\in\m{M}_j$. On each connected component $Q$ of the set of vectors $\mathbf{v}\in\reals^{\types}$ at which none of these functions vanishes, the set of positive scores and their order are fixed for every arriving type. Hence, $\mathbf{F}(\mathbf{v})=A_Q\mathbf{v}+\mathbf{b}_Q$ for all $\mathbf{v}\in Q$, where the entries of the matrix $A_Q$ and of the vector $\mathbf{b}_Q$ are formed from $\boldsymbol{\lambda}$, $\boldsymbol{\mu}$, $\mathbf{r}$, and at most $|\types|^2$ of the constants $\gamma_S$. Since there are finitely many components and their union is dense in $\reals^{\types}$, we can choose a component $Q$ whose closure contains $\mathbf{v}^{\m{M}}$, and since $\mathbf{F}$ is continuous (\Cref{lem:F-one-sided}), $\mathbf{v}^{\m{M}}=\mathbf{F}(\mathbf{v}^{\m{M}})=A_Q\mathbf{v}^{\m{M}}+\mathbf{b}_Q$. Moreover, $I-A_Q$ is invertible. Indeed, on the open set $Q$, the map $\mathbf{F}_\theta$ coincides with the affine map $\mathbf{v}\mapsto\big((1-\theta)I+\theta A_Q\big)\mathbf{v}+\theta\mathbf{b}_Q$, and since $\mathbf{F}_\theta$ is a contraction with factor $\kappa$, the matrix $(1-\theta)I+\theta A_Q$ has operator norm at most $\kappa<1$ with respect to $\|\cdot\|_{\boldsymbol{\lambda}}$. Hence $\theta(I-A_Q)=I-\big((1-\theta)I+\theta A_Q\big)$ is invertible.

Thus, $\mathbf{v}^{\m{M}}$ is the unique solution of the linear system $(I-A_Q)\mathbf{v}=\mathbf{b}_Q$. By standard bounds on the encoding length of solutions of linear systems (see, e.g., \citealp{GrotschelLovaszSchrijver1988}), every coordinate of $\mathbf{v}^{\m{M}}$ is therefore a rational number with denominator at most $2^{K}$, for some $K$ that is polynomial in the input size and does not depend on $Q$; here we use the convention of footnote~\ref{fn:arithmetic} on the constants $\gamma_S$. Two distinct rational numbers with denominators at most $2^K$ differ by at least $2^{-2K}$. Hence, it suffices to compute $\tilde{\mathbf{v}}$ as above with $\|\tilde{\mathbf{v}}-\mathbf{v}^{\m{M}}\|_\infty\le2^{-2K-2}$, and to round each coordinate of $\tilde{\mathbf{v}}$ to the unique rational number with denominator at most $2^K$ within distance $2^{-2K-2}$, which the continued fraction method does in polynomial time (see, e.g., \citealp{GrotschelLovaszSchrijver1988}). This recovers $\mathbf{v}^{\m{M}}$ exactly, in time polynomial in the input size.
\end{proof}

\equivalence*

\begin{proof}
Suppose that \eqref{eq:replacement-joint-reward}--\eqref{eq:replacement-joint-z} hold. Fix an arriving type $j$, list its acceptable partner types as
$i_1\succ_j\cdots\succ_j i_m$, and write $P_k=\{i_1,\ldots,i_k\}$.
First consider the case $m\ge1$. By \eqref{eq:replacement-joint-reward},
\begin{equation}\label{eq:replacement-cumulative}
 q_k\triangleq r_{i_kj}-v_{i_k}-v_j
 =\sum_{\ell=k}^m z_{P_\ell j}\ge0,
 \qquad k=1,\ldots,m.
\end{equation}
Set $q_{m+1}=0$. Nonnegativity of the weights implies
\begin{equation}\label{eq:replacement-differences}
 q_1\ge\cdots\ge q_m\ge0,
 \qquad z_{P_kj}=q_k-q_{k+1}.
\end{equation}
Thus, acceptable matches have nonnegative scores and are ordered by nonincreasing score. Moreover, \eqref{eq:replacement-joint-rejected} implies that every positive-score match is acceptable. Any additional acceptable matches have score zero and appear at the end of the list.

For each prefix $P_k$, the minimum residual score $u_{ij}(\mathbf{v})$ inside the set is $q_k$, and the maximum outside is $q_{k+1}$, since all rejected matches have zero residual score. Here, we use the convention that the maximum and minimum over an empty set are zero. Therefore, \eqref{eq:replacement-Z} and \eqref{eq:replacement-differences} give
\[
 Z_{P_kj}(\mathbf{v})=q_k-q_{k+1}=z_{P_kj}.
\]
If $Z_{Sj}(\mathbf{v})>0$, every type in $S$ has a positive residual score strictly larger than every residual score outside $S$. Such a set must be a prefix of the acceptable list, so $Z_{Sj}(\mathbf{v})=0$ for every nonprefix set.

If $m=0$, \eqref{eq:replacement-joint-rejected} implies that $u_{ij}(\mathbf{v})=0$ for every type $i$, and hence $Z_{Sj}(\mathbf{v})=0$ for every nonempty set $S$. Thus, in either case, extending the prefix weights by zero gives
\[
 z_{Sj}=Z_{Sj}(\mathbf{v})
 \qquad\forall j\in\mathcal T,\ \emptyset\ne S\subseteq\mathcal T.
\]
The preceding argument also shows that $\pi$ is consistent with $\mathbf{v}$. Substituting these weights into \eqref{eq:replacement-joint-balance} gives $v_i=F_i(\mathbf{v})$ for every type $i$.

Conversely, suppose that $\mathbf{v}\ge0$ is a fixed point, that is, it satisfies $\mathbf{F}(\mathbf{v})=\mathbf{v}$, and let $\pi$ be any static priority policy consistent with $\mathbf{v}$. By \Cref{defn:dual-consistent-policy}, $\pi$ accepts every match with positive score $r_{ij}-v_i-v_j$ and ranks these matches by nonincreasing score, with ties broken arbitrarily. It may also accept matches with score exactly zero, which it ranks below all positive-score matches. Fix an arriving type $j$, list its distinct positive residual scores $u_{ij}(\mathbf{v})=(r_{ij}-v_i-v_j)^+$ as $c_1>\cdots>c_d>0$, set $c_{d+1}=0$, and define
\[
 S_k=\{i\in\mathcal T:u_{ij}(\mathbf{v})\ge c_k\},
 \qquad k=1,\ldots,d.
\]
The only positive weights in \eqref{eq:replacement-Z} (after setting $\mathcal{M}=\TT$) are
$Z_{S_kj}(\mathbf{v})=c_k-c_{k+1}$, and each $S_k$ is a prefix of $\pi$.
Summing these successive differences gives
\begin{equation}\label{eq:replacement-telescope}
 \sum_{\emptyset\ne S\subseteq\mathcal T}
 Z_{Sj}(\mathbf{v})\mathbf{1}(i\in S)
 =u_{ij}(\mathbf{v})
 \qquad\forall i,j\in\mathcal T.
\end{equation}
If there are no positive residual scores, all the weights are zero, and the same identity holds.

Set $z_{Sj}=Z_{Sj}(\mathbf{v})$ for every prefix $S\in\mathcal P_j(\pi)$. Since acceptable matches have nonnegative scores, \eqref{eq:replacement-telescope} establishes \eqref{eq:replacement-joint-reward}, while the consistency of $\pi$ gives \eqref{eq:replacement-joint-rejected}. Finally, $\mathbf{F}(\mathbf{v})=\mathbf{v}$ gives \eqref{eq:replacement-joint-balance}, and the required nonnegativity conditions hold by assumption and construction.
\end{proof}

\subsubsection{Proofs from \Cref{sec:restricted-values}}
\label{apx:proofs-restricted}
The following lemma is used in the proof of \Cref{prop:restricted-certificate}.

\begin{lemma}
\label{lem:prefix-of-consistent}
Let $\pi=\{\porder{j}\}_{j\in\types}$ be consistent with $\m{M}$ and $\mathbf{v}$, and fix $j\in\m{T}$. If a nonempty set $S\subseteq\m{M}_j$ satisfies $Z^{\m{M}}_{Sj}(\mathbf{v})>0$, then $S\in\prefixes{\porder{j}}$.
\end{lemma}
\begin{proof}
By \eqref{eq:Z-integral}, there exists $y>0$ such that
\[
 S=\{i\in\m{M}_j:u_{ij}(\mathbf{v})\ge y\}.
\]
Thus, $u_{ij}(\mathbf{v})\ge y>0$ for every $i\in S$, while
$u_{i'j}(\mathbf{v})<y$ for every $i'\in\m{M}_j\setminus S$.

Fix $i\in S$. Since $(i,j)\in\m{M}$ and
$r_{ij}-v_i-v_j=u_{ij}(\mathbf{v})>0$, the second condition of
\Cref{defn:dual-consistent-policy} gives $i\porder{j}\emptyset$.
Next, let $i'\notin S$ be acceptable to $j$, so that
$i'\porder{j}\emptyset$. By the first condition of
\Cref{defn:dual-consistent-policy}, $(i',j)\in\m{M}$ and
$r_{i'j}-v_{i'}-v_j\ge0$. Hence, $i'\in\m{M}_j\setminus S$ and
\[
 r_{i'j}-v_{i'}-v_j
 =u_{i'j}(\mathbf{v})
 <y
 \le u_{ij}(\mathbf{v})
 =r_{ij}-v_i-v_j.
\]
The third condition of \Cref{defn:dual-consistent-policy} therefore gives
$i\porder{j}i'$. Thus, every type in $S$ is acceptable and ranks above
every acceptable type outside $S$. Letting $i_0$ be the lowest-ranked
type in $S$, we conclude that
\[
 S=\{i'\in\types:i'\succeq_j i_0\},
\]
which is a prefix of $\porder{j}$ by \Cref{def:prefix}.
\end{proof}

\restrictedcertificate*

\begin{proof}
Write $\mathbf{v}=\mathbf{v}^{\m{M}}$, so that $\mathbf{v}\ge0$ and
$\mathbf{F}^{\m{M}}(\mathbf{v})=\mathbf{v}$ by
\Cref{prop:fixed-point}. Recall that $\m{P}_j(\pi)$ denotes
the collection of nonempty prefixes of $\porder{j}$. By the first
condition of \Cref{defn:dual-consistent-policy}, every such prefix is
a subset of $\m{M}_j$. Define
\[
 z_{Sj}\triangleq Z^{\m{M}}_{Sj}(\mathbf{v})
 \qquad\forall j\in\m{T},\ S\in\m{P}_j(\pi).
\]
By \Cref{lem:prefix-of-consistent}, every nonempty set
$S\subseteq\m{M}_j$ with $Z^{\m{M}}_{Sj}(\mathbf{v})>0$ is a prefix
of $\porder{j}$. Thus, restricting the sums to prefixes omits only
zero weights. In particular, \eqref{eq:telescope-M} gives
\begin{equation}\label{eq:z-on-prefixes}
\begin{aligned}
 \sum_{S\in\m{P}_j(\pi)}z_{Sj}\,\mathbf{1}(i\in S)
 &=\sum_{\emptyset\ne S\subseteq\m{M}_j}
     Z^{\m{M}}_{Sj}(\mathbf{v})\,\mathbf{1}(i\in S)\\
 &=u_{ij}(\mathbf{v})
 \qquad\forall i\in\m{M}_j.
\end{aligned}
\end{equation}

We claim that $(\mathbf{v},\mathbf{z})$ is feasible for
$\textup{DLP}^{\textup{LB}}(\pi)$ in \Cref{sec:value-certificate-1}. Nonnegativity holds by assumption and
construction. To verify the constraints
\eqref{eq:replacement-reward}, fix $(i,j)\in\m{M}(\pi)$. The first
condition of \Cref{defn:dual-consistent-policy} gives
$(i,j)\in\m{M}$ and $r_{ij}-v_i-v_j\ge0$, so
$u_{ij}(\mathbf{v})=r_{ij}-v_i-v_j$. Therefore,
\eqref{eq:z-on-prefixes} implies
\[
 v_i+v_j+\sum_{S\in\m{P}_j(\pi)}z_{Sj}\,\mathbf{1}(i\in S)
 =r_{ij},
\]
which establishes \eqref{eq:replacement-reward} with equality.

To verify the constraints \eqref{eq:replacement-balance},
fix $i\in\m{T}$. Using the fixed-point equation, the definition
\eqref{eq:replacement-F}, and the fact that all positive weights
are supported on prefixes, we obtain
\begin{align*}
 \mu_i v_i
 &=\mu_i F^{\m{M}}_i(\mathbf{v})\\
 &=\sum_{j\in\m{T}}\lambda_j
   \sum_{\emptyset\ne S\subseteq\m{M}_j}
   \gamma_S Z^{\m{M}}_{Sj}(\mathbf{v})\,\mathbf{1}(i\in S)\\
 &=\sum_{j\in\m{T}}\sum_{S\in\m{P}_j(\pi)}
   \lambda_j\gamma_S z_{Sj}\,\mathbf{1}(i\in S).
\end{align*}
Thus, \eqref{eq:replacement-balance} also holds with equality.
The pair $(\mathbf{v},\mathbf{z})$ is therefore feasible, and
\Cref{prop:replacement-certificate} gives
$V(\pi)\ge\sum_{i\in\m{T}}\lambda_i v_i$, as claimed.
\end{proof}

% [Version 5: moved here from the former Appendix B.4.1. The lemma now states only the bound, and the dual feasibility of the restricted fixed point is shown in the proof.]
\fixedpointdualfeasible*
\begin{proof}
We use the dual $\dlpm$ of $\plpm$, stated in \Cref{sec:suitable-dual}. Write $\mathbf{v}=\mathbf{v}^{\m{M}}$, which is nonnegative by \Cref{prop:fixed-point}, and
$\mathbf{z}=\mathbf{Z}^{\m{M}}(\mathbf{v})\ge0$. By the telescoping identity
\eqref{eq:telescope-M}, the left side of
\eqref{eq:dual-alg-reward} equals
\[
 v_i+v_j+u_{ij}(\mathbf{v})
 \ge v_i+v_j+(r_{ij}-v_i-v_j)
 =r_{ij}
 \qquad\forall (i,j)\in\m{M}.
\]
Thus, the constraints \eqref{eq:dual-alg-reward} hold. Moreover, by the definition
\eqref{eq:replacement-F} of $\mathbf{F}^{\m{M}}$ and the fixed-point
equation, the right side of \eqref{eq:dual-alg-balance} equals
$\mu_i F^{\m{M}}_i(\mathbf{v})=\mu_i v_i$ for every $i\in\m{T}$,
so the constraints \eqref{eq:dual-alg-balance} hold with equality. Hence,
$(\mathbf{v},\mathbf{z})$ is feasible for $\dlpm$, and weak duality gives $\sum_{i\in\m{T}}\lambda_i v_i\ge\plpm$, which finishes the proof of the lemma.
\end{proof}

\medskip
\noindent{\textbf{Failure of direct application of \Cref{prop:replacement-half}  in the restricted case.}} A direct application of \Cref{prop:replacement-half}, following the approach in \Cref{sec:value-certificate}, requires the reversed constraints \eqref{eq:replacement-reverse-reward} to hold for \emph{all} pairs of types. For pairs in $\m{M}$, they follow from \eqref{eq:telescope-M}. However, if $(i,j)\notin\m{M}$, no prefix of $\pi$ for arrivals of type $j$ contains $i$, so the corresponding constraint reduces to
\[
 v^{\m{M}}_i+v^{\m{M}}_j\ge r_{ij}.
\]
Thus, to apply \Cref{prop:replacement-half} with the restricted fixed point $\mathbf{v}^{\m{M}}$ and its associated prefix weights, every excluded match must have a nonpositive score under $\mathbf{v}^{\m{M}}$.
Under this condition, however, the restriction does not change the fixed point:
\begin{equation}\label{eq:equivalence-fixed}
 \mathbf{v}^{\m{M}}=\mathbf{v}^{\TT}.
\end{equation}
To see this, note that all excluded matches have nonpositive scores at $\mathbf{v}^{\m{M}}$, so $u_{ij}(\mathbf{v}^{\m{M}})=0$ for every $(i,j)\notin\m{M}$. Removing them therefore changes neither the sets carrying positive weight nor their weights, so the restricted and unrestricted maps agree at this vector:
\[
 \mathbf{F}(\mathbf{v}^{\m{M}})
 =\mathbf{F}^{\m{M}}(\mathbf{v}^{\m{M}})
 =\mathbf{v}^{\m{M}}.
\]
Uniqueness of the unrestricted fixed point, established in \Cref{prop:fixed-point}, then gives \eqref{eq:equivalence-fixed}. Moreover, if both constructions reject zero-score matches and use the same tie-breaking rule among equal positive scores, the resulting value-based static priority policies make the same matching decisions.

\subsection{Omitted Proofs from \Cref{sec:connections}}
\label{apx:omittedproofs-connections}

Throughout this subsection, we fix $\m{M}\subseteq\TT$ and use the notation of the previous subsection. We also use the following consequences of complementary slackness for an optimal primal solution $(\bn^*,\bx^*)$ to $\plpm$ and an optimal dual solution $(\mathbf{v}^*,\mathbf{z}^*)$ to $\dlpm$. If $z^*_{Sj}>0$, then the availability constraint \eqref{eq:availability-alg} indexed by $(S,j)$ is tight at $(\bn^*,\bx^*)$, that is, $S\in\m{P}_j(\bn^*,\bx^*)$. Similarly, if $x^*_{ij}>0$, then the constraint \eqref{eq:dual-alg-reward} indexed by $(i,j)$ is tight at $(\mathbf{v}^*,\mathbf{z}^*)$. Finally, since $n^*_i>0$ for every $i\in\m{T}$ (see the proof of \Cref{lemma:suitablestructure}), the  constraints \eqref{eq:dual-alg-balance} hold with equality at every optimal solution to $\dlpm$.

% [Version 5: the proof of Lemma 6.1 (now Lemma 5.7) was moved to Appendix B.3.3.]

\subsubsection{Proofs from \Cref{sec:suitable-dual}}
\label{apx:proofs-dual-suitable}

\dualsuitability*
\begin{proof}
\emph{(a) $\Rightarrow$ (b).} Let $(\bn^*,\bx^*)$ be a suitable optimal solution to $\plpm$, and let $(\mathbf{v}^*,\mathbf{z}^*)$ be any optimal solution to $\dlpm$. Fix $j\in\m{T}$. For each $i\in\m{M}_j$, write $w_i=r_{ij}-v_i^*-v_j^*$ for its score, and let
\[
 \m{Z}_j\triangleq\{\emptyset\ne S\subseteq\m{M}_j:z^*_{Sj}>0\}.
\]
By complementary slackness, every set in $\m{Z}_j$ has a tight availability constraint, so
$\m{Z}_j\subseteq\m{P}_j(\bn^*,\bx^*)$. Suitability implies that $x^*_{ij}>0$ for every type $i$ contained in one of these sets. Applying complementary slackness again gives
\begin{equation}\label{eq:score-is-weight}
 w_i=\sum_{S'\ni i}z^*_{S'j}
 \qquad\forall i\in\bigcup_{S\in\m{Z}_j}S,
\end{equation}
where the sum ranges over subsets of $\m{M}_j$ containing $i$. Moreover, by \Cref{lemma:suitablestructure}, the tight availability sets are nested, and hence so are the sets in $\m{Z}_j$.

If $\m{Z}_j=\emptyset$, the reward constraints \eqref{eq:dual-alg-reward} imply that $w_i\le0$ for every $i\in\m{M}_j$. All residual scores are therefore zero, so
$z^*_{Sj}=Z^{\m{M}}_{Sj}(\mathbf{v}^*)=0$ for every nonempty $S\subseteq\m{M}_j$. We may thus assume that $\m{Z}_j\ne\emptyset$.

Fix $S_1\in\m{Z}_j$. For any $i\in S_1$ and $i'\in\m{M}_j\setminus S_1$, every positive-weight set containing $i'$ must also contain $i$. Indeed, if $S\in\m{Z}_j$ contains $i'$, then $S\not\subseteq S_1$, so nestedness implies $S\supsetneq S_1$. Since $S_1$ contains $i$ but not $i'$, we obtain
\begin{align*}
 w_i
 =\sum_{S\ni i}z^*_{Sj}
 \ge z^*_{S_1j}+\sum_{S\ni i'}z^*_{Sj}
>\sum_{S\ni i'}z^*_{Sj}
 \ge w_{i'}.
\end{align*}
The equality follows from \eqref{eq:score-is-weight}, and the last inequality follows from the reward constraint \eqref{eq:dual-alg-reward} for $(i',j)$. Moreover, \eqref{eq:score-is-weight} gives $w_i\ge z^*_{S_1j}>0$ for every $i\in S_1$. Thus, every type in $S_1$ has a positive score strictly larger than the score of every type outside $S_1$. Equivalently,
\[
 S_1=\{i\in\m{M}_j:w_i\ge c\},
 \qquad c\triangleq\min_{i\in S_1}w_i>0.
\]

Let $c_1>\cdots>c_K>0$ be the distinct positive scores among
$\{w_i\}_{i\in\m{M}_j}$, set $c_{K+1}=0$, and define
\[
 P_k\triangleq\{i\in\m{M}_j:w_i\ge c_k\},
 \qquad k=1,\ldots,K.
\]
We have shown that $\m{Z}_j\subseteq\{P_1,\ldots,P_K\}$. For each $k$, choose a type $i$ with $w_i=c_k$. Its reward constraint gives
$\sum_{S\ni i}z^*_{Sj}\ge c_k>0$, so $i$ belongs to a positive-weight set and \eqref{eq:score-is-weight} applies. Since $i\in P_\ell$ if and only if $\ell\ge k$,
\[
 c_k=w_i=\sum_{S\ni i}z^*_{Sj}
 =\sum_{\ell=k}^{K}z^*_{P_\ell j}
 \qquad k=1,\ldots,K.
\]
Taking successive differences yields
$z^*_{P_kj}=c_k-c_{k+1}$ for $k=1,\ldots,K$, while every other weight is zero. Since $u_{ij}(\mathbf{v}^*)=w_i^+$, the definition \eqref{eq:replacement-Z} gives exactly the same weights:
\[
 Z^{\m{M}}_{P_kj}(\mathbf{v}^*)=c_k-c_{k+1}
 \qquad k=1,\ldots,K,
\]
with $Z^{\m{M}}_{Sj}(\mathbf{v}^*)=0$ for every other nonempty set $S$. Hence,
$z^*_{Sj}=Z^{\m{M}}_{Sj}(\mathbf{v}^*)$ for every nonempty $S\subseteq\m{M}_j$. Since $j$ was arbitrary, we conclude that
$\mathbf{z}^*=\mathbf{Z}^{\m{M}}(\mathbf{v}^*)$.

\medskip\noindent
\emph{(a) $\Leftarrow$ (b).} By the strict complementary slackness theorem of \citet{goldman1956theory}, there exist optimal solutions $(\bn^*,\bx^*)$ to $\plpm$ and $(\mathbf{v}^*,\mathbf{z}^*)$ to $\dlpm$ that are strictly complementary. In particular, for every $(i,j)\in\m{M}$, we have $x^*_{ij}>0$ if and only if the corresponding reward constraint is tight at $(\mathbf{v}^*,\mathbf{z}^*)$. Similarly, for every $j$ and nonempty $S\subseteq\m{M}_j$, we have $z^*_{Sj}>0$ if and only if $S\in\m{P}_j(\bn^*,\bx^*)$.

By (b), $\mathbf{z}^*=\mathbf{Z}^{\m{M}}(\mathbf{v}^*)$. Fix $j\in\m{T}$ and $S\in\m{P}_j(\bn^*,\bx^*)$. Strict complementarity gives
\[
 z^*_{Sj}=Z^{\m{M}}_{Sj}(\mathbf{v}^*)>0.
\]
By \eqref{eq:replacement-Z}, every $i\in S$ therefore has a positive residual score, so
$u_{ij}(\mathbf{v}^*)=r_{ij}-v_i^*-v_j^*>0$. The telescoping identity \eqref{eq:telescope-M} then gives
\[
 \sum_{S'\ni i}z^*_{S'j}
 =\sum_{S'\ni i}Z^{\m{M}}_{S'j}(\mathbf{v}^*)
 =u_{ij}(\mathbf{v}^*)
 =r_{ij}-v_i^*-v_j^*.
\]
Thus, the reward constraint for $(i,j)$ is tight, and strict complementarity implies $x^*_{ij}>0$. This holds for every member of every tight availability set, establishing \eqref{eq:suitability}. Hence, $(\bn^*,\bx^*)$ is suitable by \Cref{lemma:suitablestructure}.

\medskip\noindent
\emph{Uniqueness.} Suppose that (a) and (b) hold, and let $(\mathbf{v}^*,\mathbf{z}^*)$ be any optimal solution to $\dlpm$. By (b),
$\mathbf{z}^*=\mathbf{Z}^{\m{M}}(\mathbf{v}^*)$. Since the constraints \eqref{eq:dual-alg-balance} hold with equality at every dual optimum,
\[
 \mu_i v_i^*
 =\sum_{j\in\m{T}}\lambda_j
   \sum_{\emptyset\ne S\subseteq\m{M}_j}
   \gamma_S Z^{\m{M}}_{Sj}(\mathbf{v}^*)\,\mathbf{1}(i\in S)
 \qquad\forall i\in\m{T}.
\]
The right side is nonnegative, so $\mathbf{v}^*\ge0$. By the definition of $\mathbf{F}^{\m{M}}$, these equalities are precisely the restricted fixed-point equation
$\mathbf{F}^{\m{M}}(\mathbf{v}^*)=\mathbf{v}^*$. Therefore, \Cref{prop:fixed-point} gives
$\mathbf{v}^*=\mathbf{v}^{\m{M}}$, and consequently
$\mathbf{z}^*=\mathbf{Z}^{\m{M}}(\mathbf{v}^{\m{M}})$. The optimal dual solution is thus unique. Finally, strong duality yields
\[
 \plpm=\dlpm=\sum_{i\in\m{T}}\lambda_i v^{\m{M}}_i,
\]
which finishes the proof.
\end{proof}

\subsubsection{Proofs from \Cref{sec:primal-is-dual}}
\label{apx:proofs-primal-is-dual}

\primaldualpolicies*

\begin{proof}
Since $(\bn^*,\bx^*)$ is a suitable optimal solution to $\plpm$, \Cref{thm:dual-suitability} shows that $\dlpm$ has a unique optimal solution $(\mathbf{v}^{\m{M}},\mathbf{z}^*)$, with
$\mathbf{z}^*=\mathbf{Z}^{\m{M}}(\mathbf{v}^{\m{M}})$. Fix an arriving type $j\in\m{T}$, and write
$w_i=r_{ij}-v^{\m{M}}_i-v^{\m{M}}_j$ for the score of each $i\in\m{M}_j$. The telescoping identity \eqref{eq:telescope-M} gives
\[
 \sum_{S\ni i}z^*_{Sj}=u_{ij}(\mathbf{v}^{\m{M}})
 \qquad\forall i\in\m{M}_j,
\]
where the sum ranges over subsets of $\m{M}_j$ containing $i$.
For $y>0$, let $S_j(y)$ be the set defined in \eqref{eq:score-prefix} with $\mathbf{v}=\mathbf{v}^{\m{M}}$.
Whenever $w_i>0$, the set $S_j(w_i)$ contains $i$ and has positive dual weight. Indeed, the minimum residual score within this set is $w_i$, while the maximum outside it is strictly smaller, with the maximum and minimum over an empty set taken to be zero. Thus, \eqref{eq:replacement-Z} implies
\[
 z^*_{S_j(w_i)j}
 =Z^{\m{M}}_{S_j(w_i)j}(\mathbf{v}^{\m{M}})>0.
\]

Now let $\pi$ be consistent with $(\bn^*,\bx^*)$ in the sense of \Cref{def:associated-greedy}. Complementary slackness applies to this primal solution and the optimal dual solution $(\mathbf{v}^{\m{M}},\mathbf{z}^*)$. We verify the three conditions of \Cref{defn:dual-consistent-policy}.

First, if $i\porder{j}\emptyset$, consistency with $(\bn^*,\bx^*)$ implies that $(i,j)\in\m{M}$ and $x^*_{ij}>0$. The reward constraint for $(i,j)$ is therefore tight, giving
\[
 w_i=\sum_{S\ni i}z^*_{Sj}\ge0.
\]

Second, suppose that $(i,j)\in\m{M}$ and $w_i>0$. Since
$z^*_{S_j(w_i)j}>0$, complementary slackness gives
$S_j(w_i)\in\m{P}_j(\bn^*,\bx^*)$. Consistency then implies
$S_j(w_i)\in\prefixes{\porder{j}}$. Since $i\in S_j(w_i)$ and every member of a prefix is acceptable, we have $i\porder{j}\emptyset$.

Third, suppose that $(i,j),(i',j)\in\m{M}$ and $w_i>w_{i'}\ge0$. As above, $S_j(w_i)$ is a prefix of $\porder{j}$. It contains $i$ but not $i'$, since
$u_{i'j}(\mathbf{v}^{\m{M}})=w_{i'}<w_i$. By \Cref{def:prefix}, every member of this prefix ranks above every type outside it, so $i\porder{j}i'$. Since $j$ was arbitrary, $\pi$ is consistent with $\m{M}$ and $\mathbf{v}^{\m{M}}$.

%\rn{Please check this proof}
\end{proof}

\subsubsection{Proofs from \Cref{sec:connection-poly}}
\label{apx:proofs-connection-poly}

% [Version 5: the first part of this proof (on LP^ALG(M)) was moved to Appendix B.2 with Proposition 4.4.]
\polytimedual*
\begin{proof}
By the proof of \Cref{prop:polytime} in  \Cref{apx:proofs-suitable}, the ellipsoid method, together with an exact polynomial-time separation oracle, solves $\plpm$ in polynomial time. Although $\dlpm$ has a variable $z_{Sj}$ for each availability constraint, the ellipsoid method applied to $\plpm$ generates only polynomially many constraints, and a basic optimal solution of $\dlpm$ that is supported on these constraints---hence with polynomially many nonzero weights $z_{Sj}$---can be computed in polynomial time \citep[Theorem~6.5.14; see also Lemma~6.5.15 and Theorem~6.5.17]{GrotschelLovaszSchrijver1988}. We omit further details for brevity, as this is a standard result.
\end{proof}

% \section{Uniqueness of Fixed-point Restricted to $\m{M}$}
% \label{sec:uniqueness}
% \input{tex/apx-fixedpoint-uniqueness}

\section{Additional Examples}
\label{sec:additional-examples}
The following instance shows that the vector of values in an optimal solution to $\dlplong$ need not be a fixed point of $\mathbf{F}$, and a value-based policy constructed from these values can earn less than $\plpt$; see \Cref{sec:suitable-dual}.

\begin{example}
\label{ex:not_suitable_failure}
    Consider an instance with $\m{T}  = \{1,2,3\}$,
    \begin{align*}
        \lambda_1 = 1,\ \lambda_2 &= 0.1,\ \lambda_3 = 1,\\
        \mu_1 = 0.5,\ \mu_2 &=0.05,\ \mu_3 = 0.5,
    \end{align*}
    and rewards
    \begin{equation*}
        \mathbf{r} = \begin{pmatrix}
            0.1 & 5 & 0.6\\
            5 & 20 & 5.5\\
            0.6 & 5.5 & 1.1
            \end{pmatrix}.
    \end{equation*}
\end{example}

If we solve $\dlplong$, we get the vector of values $\mathbf{v}^* = (0.033, 4.965, 0.531)$. Every score $r_{ij}-v^*_i-v^*_j$ is positive, so the only policy $\pi = \{\succ_j\}_{j\in \m{T}}$ consistent with these values is: 
    \begin{align*}
        3 \succ_1 1 &\succ_1 2 \succ_1 \emptyset \\
        2 \succ_2 3 &\succ_2 1 \succ_2 \emptyset\\
        3 \succ_3 1 &\succ_3 2 \succ_3 \emptyset
    \end{align*}
 In this instance, we then have that\footnote{We can check that the exact value of the policy is $\frac{4267}{4060}$. The calculation is simplified by the fact that every match is acceptable, so there is never more than one available agent.}
    \[
    \plpt = \dlplong = \sum_{i \in \m{T}} \lambda_i v_i^* \approx 1.061 > 1.051 \approx \sum_{i,j \in \m{T}} r_{ij} x^\pi_{ij} = V(\pi).
    \]

    To understand why this policy underperforms, consider the optimal primal solution $(\mathbf{n}^*,\mathbf{x}^*)$ for $\plpt$, where $x^*_{32}=0$ and $x^*_{22}\approx0.0281$, compared with $x^\pi_{32}\approx0.0212$ and $x^\pi_{22}\approx0.0026$. Thus, $\pi$ makes the most valuable $(2,2)$ matches at less than one-tenth of the rate prescribed by the primal LP. A $(2,2)$ match requires a type-$2$ agent to remain available until another type-$2$ agent arrives. By matching type-$2$ agents with other types instead, $\pi$ loses opportunities to make these more valuable matches.

In the terminology of \Cref{sec:suitable}, the optimal primal solution $(\mathbf{n}^*,\mathbf{x}^*)$ is not suitable: type $3$ belongs to a tight availability set for arrivals of type $2$, even though $x^*_{32}=0$. In this instance, $\mathbf{v}^*$ is also not a fixed point of $\mathbf{F}$. \Cref{algo:suitability} removes this obstruction: it first removes $(3,2)$, then the four additional zero-rate pairs $(1,1),(1,2),(2,1),(2,3)$, and returns $\m{M}=\{(1,3),(2,2),(3,1),(3,3)\}$. The resulting program $\plpm$ has a suitable optimal solution, and its value satisfies $\plpm\approx1.083>\plpt$.

For this $\m{M}$, the vector of values in the optimal solution to $\dlpm$ is the restricted fixed point $\mathbf{v}^{\m{M}}\approx(0.074,6.336,0.376)$. Under the value-based policy consistent with $\m{M}$ and $\mathbf{v}^{\m{M}}$, an arriving type $1$ accepts only type $3$, an arriving type $2$ accepts only type $2$, and an arriving type $3$ ranks type $3$ above type $1$. This policy reserves type-$2$ agents for matches with one another and earns approximately $1.338$, well above $\plpm$.

\end{APPENDIX}

\end{document}